\documentclass[10pt,A4paper]{article}

\usepackage[normalem]{ulem}
\usepackage{amssymb}
\usepackage{amsmath}
\usepackage{amsthm}
\usepackage{latexsym}
\usepackage[dvips]{epsfig}
\usepackage{enumerate}
\usepackage{mathrsfs}
\usepackage{eufrak}
\usepackage{bm}
\usepackage{tikz}
\usepackage{authblk}
\usepackage{cancel}
\usepackage{color}
\definecolor{pink}{rgb}{1,0,1}
\newcommand{\pink}[1]{{\textcolor{pink}{#1}}}

\allowdisplaybreaks

\theoremstyle{plain}
\newtheorem{proposition}{Proposition}
\newtheorem{lemma}{Lemma}
\newtheorem{theorem}{Theorem}
\newtheorem{assumption}{Assumption}

\newtheorem{corollary}{Corollary}

\newtheorem{remark}{Remark}

\def\bmg{{\bm g}}

\def\bml{{\bm l}}
\def\bmn{{\bm n}}
\def\bmm{{\bm m}}

\def\bmsigma{{\bm \sigma}}

\def\nablasl{/\kern-0.58em\nabla}
\def\Deltasl{/\kern-0.58em\Delta}

\def\bmpartial{{\bm \partial}}

\def\bmsigma{{\bm \sigma}}

\newcounter{mnotecount}

\newcommand{\mnotex}[1]
{\protect{\stepcounter{mnotecount}}$^{\mbox{\footnotesize $\bullet$\themnotecount}}$ 
\marginpar{
\raggedright\tiny\em
$\!\!\!\!\!\!\,\bullet$\themnotecount: #1} }

\begin{document}

\title{\textbf{Improved existence for the characteristic initial value
    problem with the conformal Einstein field equations}}

\author[,1]{David Hilditch
  \footnote{E-mail address:{\tt david.hilditch@tecnico.ulisboa.pt}}}
\author[,2]{Juan A. Valiente Kroon \footnote{E-mail address:{\tt
      j.a.valiente-kroon@qmul.ac.uk}}}
\author[,2]{Peng Zhao \footnote{E-mail address:{\tt p.zhao@qmul.ac.uk}}}
\affil[1]{CENTRA, Departamento de F\'isica, Instituto Superior
  T\'ecnico – IST, Universidade de Lisboa – UL, Avenida Rovisco Pais
  1, 1049 Lisboa, Portugal.}
\affil[2]{School of Mathematical Sciences, Queen Mary, University of London,
Mile End Road, London E1 4NS, United Kingdom.}

\maketitle

\begin{abstract}
We adapt Luk's analysis of the characteristic initial value problem in
General Relativity to the asymptotic characteristic problem for the
conformal Einstein field equations to demonstrate the local existence
of solutions in a neighbourhood of the set on which the data are
given. In particular, we obtain existence of solutions along a narrow
rectangle along null infinity which, in turn, corresponds to an
infinite domain in the asymptotic region of the physical
spacetime. This result generalises work by K{\'a}nn{\'a}r on the local
existence of solutions to the characteristic initial value problem by
means of Rendall's reduction strategy. In analysing the conformal
Einstein equations we make use of the Newman-Penrose formalism and a
gauge due to J. Stewart.
\end{abstract}

\section{Introduction}

This article is the second of a series in which we study the
characteristic initial value problem (CIVP) in General Relativity in a
range of settings. In~\cite{HilValZha19} (henceforth Paper~I) it has
been shown that Luk's strategy ---see \cite{Luk12}, to obtain an
improved local existence result for the characteristic problem for the
vacuum Einstein field equations can be adapted to a different gauge
based on the Newman-Penrose (NP) formalism, the \emph{Stewart gauge},
see \cite{SteFri82,Fri91}. In the present article we consider the
\emph{asymptotic characteristic initial value problem}, a CIVP for
Friedrich's \emph{conformal Einstein field equations} (a regular
conformal representation of the vacuum Einstein field equations) in
which one of the null initial hypersurfaces is a portion of past null
infinity, and show that Luk's strategy can also be adapted to this
setting. Accordingly, we obtain a domain of existence of the solution
to the conformal Einstein field equations on a narrow rectangle having
a portion of null infinity as one of its long sides ---see
Figure~\ref{Fig:ComparisonExistenceDomains}. In doing so we improve
K\'ann\'ar's local existence result for the asymptotic CIVP in which
existence of a solution is only guaranteed in a neighbourhood of the
intersection of the initial null hypersurfaces ---see~\cite{Kan96b}
and also~\cite{CFEBook}, Chapter~18. Expressed in terms of a solution
to the Einstein field equations the improved rectangular existence
domain corresponds, in fact, to an infinite domain.  K\'ann\'ar's
result is, in turn, an extension to the setting of smooth
(i.e.~$C^\infty$) functions of Friedrich's seminal analysis of the
CIVP for the Einstein and conformal Einstein field equations in the
analytic setting ---see~\cite{Fri81a,Fri81b,Fri82}.

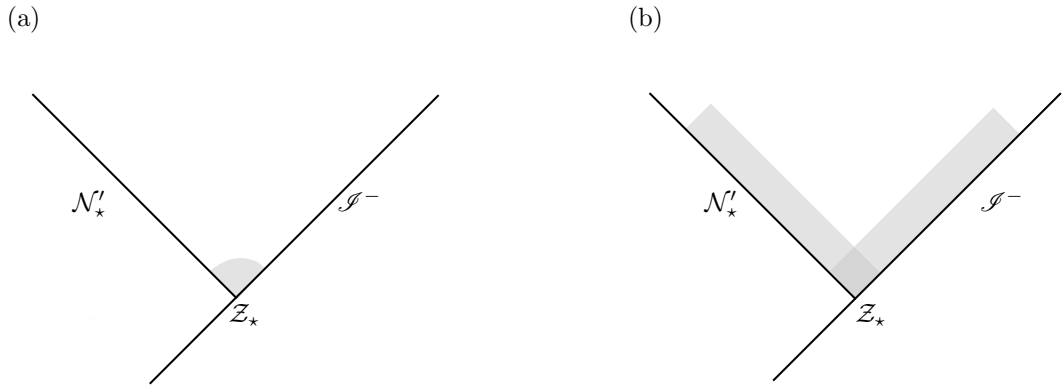
\begin{figure}[t]
\centering
\includegraphics[width=\textwidth]{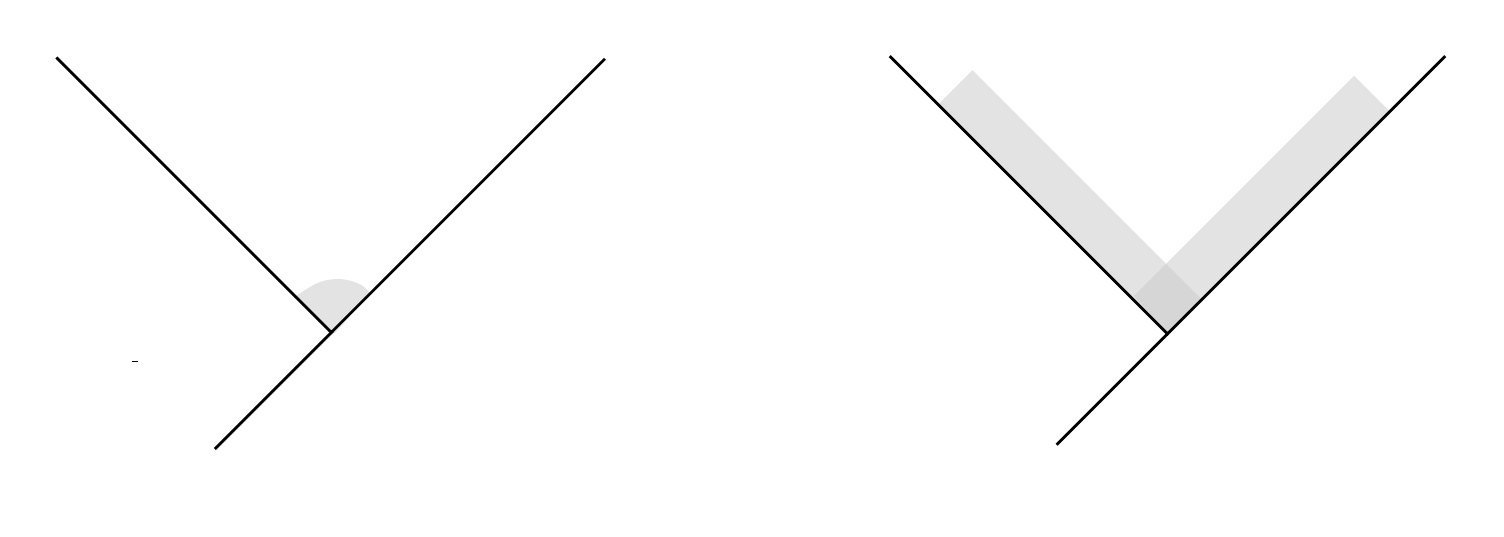}
\put(-415,160){(a)} \put(-180,160){(b)}
\put(-47,90){$\mathscr{I}^-$} \put(-290,90){$\mathscr{I}^-$}
\put(-390,90){$\mathcal{N}'_\star$}
\put(-151,90){$\mathcal{N}'_\star$} \put(-331,48){$\mathcal{Z}_\star$}
\put(-94,48){$\mathcal{Z}_\star$}
\caption{Comparison of the existence domains for the characteristic
  problem: (a) existence domain using Rendall's strategy based on the
  reduction to a standard Cauchy problem; (b) existence domain using
  Luk's strategy ---in principle, the long side of the rectangles
  extends for as much as one has control on the initial data.}
\label{Fig:ComparisonExistenceDomains}
\end{figure}

A different strategy to the proof of the main theorem in the present
article has been given in \cite{CabChrWaf14}. This alternative proof
makes use of the fact that the conformal Einstein field equations
imply a set of quasilinear wave equations for which a general theory
of \emph{improved} existence of solutions to the characteristic
problem had been developed. Thus, the approach used in the present
equation provides a different, complementary insight into the
structural properties of the conformal Einstein field equations.

An alternative asymptotic CIVP for the conformal Einstein field
equations can be formulated by prescribing initial data on a cone
representing past null infinity including its vertex ---which
corresponds to past timelike infinity. The setup of this CIVP was
given in~\cite{Fri86a,Fri88}. Statements about the existence of
solutions to this problem have been obtained
in~\cite{ChrPae13a,ChrPae13b}, see also~\cite{Fri14b}. These results
also provide improved existence results ---that is, the existence of
solutions is not restricted to the tip of the cone of past null
infinity but extends to a neighbourhood of the cone away from the
vertex as long as there is control of the initial data. An important
open problem in this respect is to obtain a solution which includes
spatial infinity. The analysis in~\cite{ChrPae13a,ChrPae13b} builds on
the theory of the CIVP for characteristic cones developed
in~\cite{Cag80,Cag81}. This theory makes use of systems of second
order hyperbolic equations and is, in principle, different from Luk's
approach.

\subsubsection*{New insights}

The NP formalism has played an influential role in mathematical
relativity. Nevertheless, its use in the formulation of existence
results has been limited. The main motivations behind the analysis in
this article is to understand the structural properties of the
conformal Einstein field equations. In this spirit, we avoid the use
of general theory as in \cite{CabChrWaf14} and give a detailed
discussion of the various arguments to convince the reader that they
indeed follow through. In order to keep the length of the article at
bay, we focus primarily on those aspects of the analysis which are
novel or not present in the discussion of Paper I. In this article we
use the phrase \emph{improved existence result} to mean an improvement
on K{\'a}nn{\'a}r's existence result in a neighbourhood of the
intersection of past null infinity and an incoming null cone. This
improved result is, in a sense, optimal in that it provides existence
in a neighbourhood of the initial hypersurfaces as long as one has
control on the initial data —see
Figure~\ref{Fig:ComparisonExistenceDomains}. Our main result is
presented in Theorem~\ref{Section:Summary} in
Section~\ref{Theorem:Main}. While the results of the present article
have been stated, for conciseness with past null infinity and an
incoming null cone as initial hypersurfaces, analogous results can be
readily obtained, \emph{mutatis mutandi}, for a setting where the data
is prescribed on future null infinity and an outgoing light cone. In
this case the development of the data is located in the past of the
union of the initial hull hypersurfaces.

\subsubsection*{Overview and main results}

In Section \ref{Section:CFE} we start by providing a brief review of
the main technical tool of this work: the conformal Einsten field
equations and its relation to the vacuum Einstein field equations. In
Section \ref{Section:GeometryProblem} we give a discussion of the
geometric formulation of the asymptotic CIVP and the gauge choice
behind this formulation ---i.e. \emph{Stewart's gauge}. Section
\ref{Section:FormulationCIVP} is concerned with the basic local
existence theory of the asymptotic CIVP. In particular, we discuss the
choices of freely specifiable initial data, the hyperbolic reduction
procedure, the computation of formal derivatives and the propagation
of the constraints and the gauge. The basic existence theorem is given
in
Theorem~\ref{Theorem:BasicExistence}. Section~\ref{Section:ImprovedSetting}
provides the basic setting for the formulation of the improved
existence result. Section~\ref{Section:Estimates} discusses the
construction of the estimates required to establish the improved
existence result. Finally, Section~\ref{Section:Summary} provides the
precise statement of the main result of our analysis.

The article contains one appendix listing in explicit form the NP
formulation of the conformal Einstein field equations used in the
analysis.

\subsubsection*{Notation and conventions}

We take~$\{ {}_a , {}_b , {}_c , \dots\}$ to denote abstract tensor
indices whereas~$\{_\mu , _\nu , _\lambda , \dots\}$ will be used as
spacetime coordinate indices with the values~${ 0, \dots, 3 }$. Our
conventions for the curvature tensors are fixed by the relation
\begin{align}
(\nabla_a \nabla_b -\nabla_b \nabla_a) v^c = R^c{}_{dab} v^d.
\end{align}
We make systematic use of the NP formalism as described, for example,
in~\cite{Ste91,PenRin84}. In particular, the signature of Lorentzian
metrics is~$(+---)$. Many of our derivations, although
straightforward, are fairly lengthy, so we have included in
Appendix~\ref{Appendix:CFEinNP} a complete summary of the equations of
the NP-formalism, highlighting the simplifications that occur with our
particular gauge.

\section{The vacuum conformal Einstein field equations}
\label{Section:CFE}

The purpose of this section is to provide a succinct summary of
Friedrich's conformal Einstein field equations ---see
e.g.~\cite{Fri81b,Fri83}. The reader interested in full details,
derivation and discussion is referred to~\cite{CFEBook}, Chapter 8. In
this article we restrict our attention to \emph{the 4-dimensional
  vacuum case with vanishing Cosmological constant}.

In what follows let~$(\tilde{\mathcal{M}},\tilde{g}_{ab})$ denote a
vacuum spacetime. Friedrich's conformal Einstein field equations are a
regular conformal representation of the Einstein field equations. That
is, they are equations for a conformally rescaled metric~$g_{ab}
=\Xi^2 \tilde{g}_{ab}$, a conformal factor~$\Xi$ and concommitants
which are formally regular up to the conformal boundary, i.e. the
points for which~$\Xi=0$, and which imply, for~$\Xi\neq 0$, a solution
to the Einstein field equations. In the following, denote
by~$(\mathcal{M},g_{ab})$ denote the conformal extension
of~$(\tilde{\mathcal{M}},\tilde{g}_{ab})$ generated by the conformal
factor~$\Xi$. Following the usual terminology we refer
to~$(\tilde{\mathcal{M}},\tilde{g}_{ab})$
and~$(\tilde{\mathcal{M}},\tilde{g}_{ab})$, respectively, as the
\emph{physical} and \emph{unphysical spacetimes}. The sets of points
of the conformal boundary giving rise to a hypersurface ---i.e. the
points for which~$\mathbf{d}\Xi\neq 0$--- are called null
infinity,~$\mathscr{I}$. In this article we assume, for simplicity,
that~$\mathscr{I}=\partial\mathcal{M}$. Spacetimes admitting this type
of smooth conformal extension are known as \emph{asymptotically
  simple} ---see \cite{PenRin86,Ste91,CFEBook} for the classic precise
definition. Null infinity can be shown to correspond to the endpoints
of null geodesics and, thus, it consists of two disconnected
components, past and future null infinity,~$\mathscr{I}^-$
and~$\mathscr{I}^+$. In this article, for concreteness we restrict our
discussion to a neighbourhood of~$\mathscr{I}^-$.

In what follows let~$\nabla_a$ denote the Levi-Civita connection of
the metric~$g_{ab}$ and let~$R_{ab}$, $C^a{}_{bcd}$ be its Ricci and
Weyl tensors, respectively. In order to introduce the conformal
Einstein field equations we further define the fields
\begin{align*}
L_{ab}&\equiv\frac{1}{2}R_{ab}-\frac{1}{12}Rg_{ab}, \\
d^a{}_{bcd}&\equiv\Xi^{-1}C^a{}_{bcd}, \\
s&\equiv\frac{1}{4}\nabla^a\nabla_a\Xi+\frac{1}{24}R\Xi,
\end{align*}
the Schouten tensor, the rescaled Weyl tensor and the Friedrich scalar
respectively. In terms of the latter the vacuum Einstein field
equations with vanishing Cosmological constant are given by
\begin{subequations}
\begin{align}
&\nabla_a\nabla_b\Xi=-\Xi L_{ab}+sg_{ab}, \label{CFE1}\\
&\nabla_as=-L_{ac}\nabla^c\Xi, \label{CFE2}\\
&\nabla_cL_{db}-\nabla_dL_{cb}=\nabla_a\Xi d^a{}_{bcd}, \label{CFE3}\\
&\nabla_a d^a{}_{bcd}=0, \label{CFE4}\\
&6\Xi s-3\nabla_c\Xi\nabla^c\Xi=0 \label{CFE5},\\
&R^c{}_{dab} = C^c{}_{dab} + 2 \big(\delta^c{}_{[a} L_{b]d} - g_{d[a}
   L_{b]}{}^c\big). \label{CFE6}
\end{align}
\end{subequations}
More details can be found in~\cite{CFEBook}.  In
Appendix~\ref{SpinorialCFE} we give the spinorial counterpart of these
equations and their components equations in Newman-Penrose (NP)
frame. This spinorial formulation of the equations will be used
systematically in the rest of the article. It is just recalled that
the spinoral version of the equations follows directly by suitably
contracting the tensorial equations with the soldering forms ---see
e.g. \cite{CFEBook} for more details.

The relation between the conformal Einstein field
equations~\eqref{CFE1}-\eqref{CFE6} and the Einstein field equations
is expressed in the following result:

\begin{proposition}[\textbf{\em solutions of the conformal vacuum Einstein
       field equations as solutions to the vacuum Einstein field
       equations}] Let
\begin{align}
 (g_{ab},\ \Xi,\ s, \ L_{ab},\ d^a{}_{bcd})\nonumber
\end{align}
denote a solution to equations~\eqref{CFE1}-\eqref{CFE4}
and~\eqref{CFE6} such that~$\Xi\neq0$ on an open
set~$\mathcal{U}\subset\mathcal{M}$. If, in addition, equation
\eqref{CFE5} is satisfied at a point~$p\in\mathcal{U}$, then the
metric~$\tilde{g}_{ab}=\Xi^{-2}g_{ab}$ is a solution to the vacuum
Einstein field equations on~$\mathcal{U}$.
\end{proposition}

The proof can be found in Chapter~$8$ of~\cite{CFEBook}.

\section{The geometry of the problem}
\label{Section:GeometryProblem}

In this section, we will discuss the geometric and the gauge choices
in the asymptotic CIVP on past null infinity.

\subsection{Basic setting}

Our basic geometric setting consists of and unphysical
manifold~$\mathcal{M}$ with a boundary and an edge. The boundary
consists of two null hypersurfaces:~$\mathscr{I}^-$, past null
infinity on which~$\Xi=0$; and~$\mathcal{N}_{\star}^{\prime}$, an
incoming null hypersurface with non-vacuum
intersection~$\mathcal{S}_{\star}\equiv\mathscr{I}^-\cap\mathcal{N}_{\star}^{\prime}$. We
will assume that~$\mathcal{S}_{\star}\approx\mathbb{S}^2$. In a
neighbourhood~$\mathcal{U}$ of $\mathcal{S}_{\star}$, one can
introduce coordinates~$x=(x^{\mu})$ with~$x^0=v$ and~$x^1=u$ such
that, at least in a neighbourhood of~$\mathcal{S}_{\star}$ one can
write
\begin{align}
\mathscr{I}^-=\{p\in\mathcal{U}\mid u(p)=0\}, \qquad
\mathcal{N}_{\star}^{\prime}=\{p\in\mathcal{U}\mid v(p)=0\}.\nonumber
\end{align}
Given suitable data
on~$(\mathscr{I}^-\cup\mathcal{N}_{\star}^{\prime})\cap\mathcal{U}$
one is interested in making statements about the existence and
uniqueness of solutions to the conformal Einstein field equations on
some open set
\begin{align}
\mathcal{V}\subset\{p\in\mathcal{U}\mid u(p)\geq0, v(p)\geq0\}
\end{align}
which we identify with a subset of the future domain of
dependence,~$D^{+}(\mathscr{I}^-\cup\mathcal{N}_{\star}^{\prime})$,
of~$\mathscr{I}^-\cup\mathcal{N}_{\star}^{\prime}$.

\begin{figure}[t]
\centering
\includegraphics[width=0.8\textwidth]{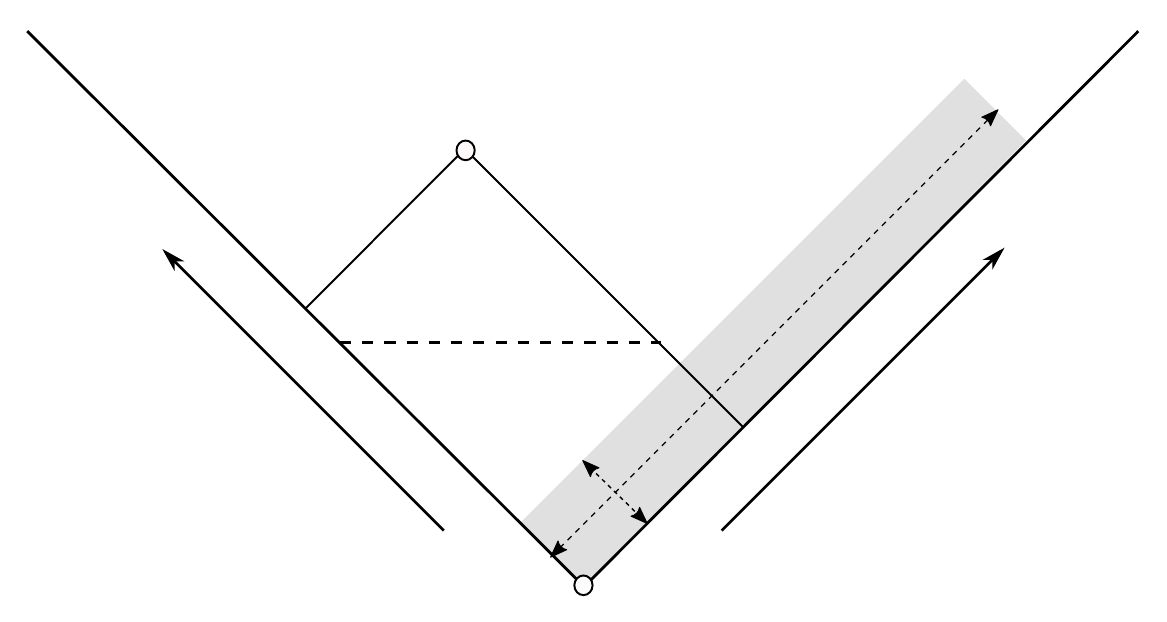}
\put(-87,60){$l^a$, $v$}
\put(-271,60){$n^a$, $u$}
\put(-40,130){$\mathscr{I}^-$}
\put(-320,130){$\mathcal{N}'_\star$}
\put(-172,-5){$\mathcal{S}_{u_\star,v_\star}$}
\put(-197,137){$\mathcal{S}_{u,v}$}
\put(-200,100){$\mathcal{D}_{u,v}$}
\put(-200,60){$\mathcal{D}^t_{u,v}$}
\put(-140,80){$t$}
\put(-160,40){$\varepsilon$}
\put(-76,155){$v_\bullet$}
\put(-245,110){$\mathcal{N}_u$}
\put(-173,110){$\mathcal{N}'_v$}
\caption{Setup for Stewart's gauge. The construction makes use of a
  double null foliation of the future domain of dependence of the
  initial hypersurface~$\mathscr{I}^-\cup\mathcal{N}_\star'$. The
  coordinates and NP null tetrad are adapted to this geometric
  setting. The analysis in this article is focused on the arbitrarily
  thin grey rectangular domain along the
  hypersurface~$\mathscr{I}^-$. The argument can be adapted, in a
  suitable manner, to a similar rectangle
  along~$\mathcal{N}_\star'$. See the main text for the definitions of
  the various regions and objects.
\label{Fig:CharacteristicSetup}}
\end{figure}
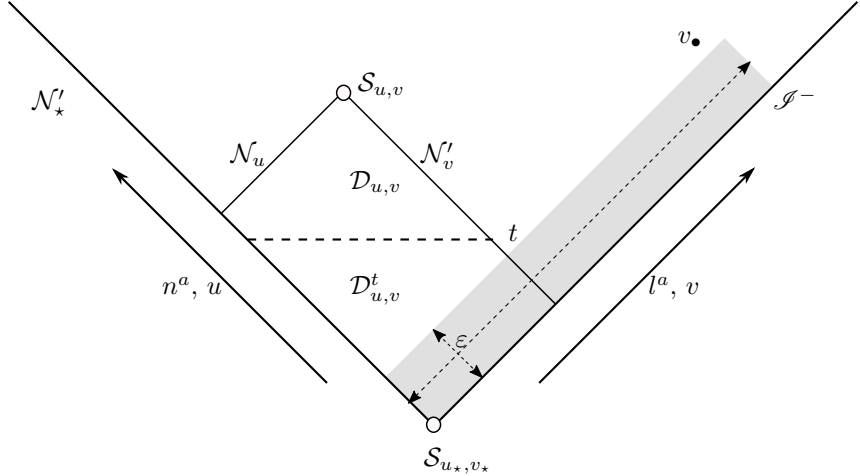

\subsection{Stewart's gauge }\label{Section:StewartGauge}

The basic geometric setting described in the previous section is
supplemented by a gauge choice first introduced by
Stewart~\cite{Ste91}.

\subsubsection{Coordinates}

It is convenient to regard the 2-dimensional
surface~$\mathcal{S}_{\star}$ as a submanifold of spacelike
hypersurface. The subsequent discussion will be restricted to the
future of the
hypersurface. As~$\mathcal{S}_{\star}\approx\mathbb{S}^2$, one has
that~$\mathcal{S}_{\star}$ divides the spacelike hypersurface into two
regions, the interior of~$\mathcal{S}_{\star}$ and the exterior
of~$\mathcal{S}_{\star}$. Now, consider a foliation of the spacelike
hypersurface by 2-dimensional surfaces with the topology
of~$\mathbb{S}^2$ which includes~$\mathcal{S}_{\star}$. At each of the
2-dimensional surfaces we assume there pass two null
hypersurfaces. Further, we assume that one of these hypersurfaces has
the property that the projection of the tangent vectors of their
generators at~$\mathcal{S}_{\star}$ point~$\mathit{outwards}$. We call
these null hypersurfaces~$\mathit{outgoing\ light\ cones}$. Moreover,
it is also assumed that one of these hypersurfaces has the property
that the projection of the tangent vectors of their generators
at~$\mathcal{S}_{\star}$ point~$\mathit{inwards}$. We call these null
hypersurfaces~$\mathit{ingoing\ light\ cones}$.

Thus, at least locally, one obtains a $1$-parameter family of outgoing
null hypersurface~$\mathcal{N}_{u}$ and a 1-parameter family of
ingoing null hypersurface~$\mathcal{N}_{v}^{\prime}$. One can then
define scalar fields~$u$ and~$v$ by the requirements, respectively,
that $u$ is constant on each of the $\mathcal{N}_{u}$ and $v$ is
constant on each~$\mathcal{N}_{v}^{\prime}$. In particular, we assume
that~$\mathcal{N}_{0}=\mathscr{I}^-$
and~$\mathcal{N}_{0}^{\prime}=\mathcal{N}_{\star}^{\prime}$. Following
standard usage, we call~$u$ a~$retarded\ time$ and~$v$ a
$advanced\ time$. The scalar fields $u$ and $v$ will be used as
coordinates in a neighbourhood of~$\mathcal{S}_{\star}$. To complete
the coordinate system, consider arbitrary
coordinates~$(x^{\mathcal{A}})$ on $\mathcal{S}_{\star}$, with the
index $^{\mathcal{A}}$ taking the values 2, 3. These coordinates are
then propagated into $\mathscr{I}^-$ by requiring them to be constant
along the generators of $\mathscr{I}^-$. Once coordinates have been
defined on $\mathscr{I}^-$, one can propagate them into $\mathcal{V}$
by requiring them to be constant along the generators of each
$\mathcal{N}_{v}^{\prime}$. In this manner one obtains a coordinate
system $(x^{\mu})=(v,\ u,\ x^{\mathcal{A}})$ in $\mathcal{V}$.

We use the notation~$\mathcal{N}_u(v_1,v_2)$ to denote the part of the
hypersurface~$\mathcal{N}_u$ with $v_1\leq v\leq
v_2$. Likewise~$\mathcal{N}'_v(u_1,u_2)$ has a similar definition. We
denote the sphere intersected by~$\mathcal{N}_{u}$
and~$\mathcal{N}_{v}'$ by~$\mathcal{S}_{u,v}$. We define the region
\begin{align*}
\bigcup_{0\leq v'\leq v, 0\leq u'\leq u}\mathcal{S}_{u',v'}
\end{align*}
as~$\mathcal{D}_{u,v}$. We also define the time function~$t\equiv
u+v$, and the {\it truncated causal diamond},
\begin{align*}
  \mathcal{D}_{u,v}^{\,\tilde{t}}\equiv\mathcal{D}_{u,v}
  \cap\{t\leq \tilde{t}\}.
\end{align*}

\begin{remark}
{\em It is observed that while the null coordinte $u$ has a compact
  range, this is, in principle, not the case for $v$. }
\end{remark}

\subsubsection{The NP frame}

A null Newman-Penrose (NP) tetrad is constructed by choosing vector
fields~$l^a$ and~$n^a$ tangent to the generators of~$\mathcal{N}_{u}$
and~$\mathcal{N}_{v}^{\prime}$ respectively. Following the standard
conventions we make use of the normalisation~$g_{ab}l^an^b=1$ is
preserved under boost transformations. This freedom can be used to
set~$n_a=\nabla_av$. This requirement still leaves some freedom left
as one can choose a relabelling of the form~$v\mapsto V(v)$. Next, we
choose the vector fields~$m^a$ and~$\bar{m}^a$ so that they are
tangent to the
surfaces~$\mathcal{S}_{u,v}\equiv\mathcal{N}_{u}\cap\mathcal{N}_{v}^{\prime}$
and satisfy the
conditions~$g_{ab}m^a\bar{m}^b=-1$,~$g_{ab}m^am^b=0$. This leaves the
freedom to perform a spin transformation at each point.

Now, observing that, by construction, on the generators of each null
hypersurface~$\mathcal{N}_{v}^{\star}$ only the coordinate $u$ varies,
one has that
\begin{align}
n^\mu\bmpartial_\mu=Q\bmpartial_u,\nonumber
\end{align} 
where~$Q$ is a real function of the position. Further, since the
vector~$l^a$ is tangent to the generators of each~$\mathcal{N}_{u}$
and~$l^an_a=l^a\nabla_av=1$, one has that
\begin{align}
  l^\mu\bmpartial_\mu =\bmpartial_v +C^{\mathcal{A}}\bmpartial_{\mathcal{A}},
  \nonumber
\end{align}
where, again, the components~$C^{\mathcal{A}}$ are real functions of
the position. By construction, the coordinates~$(x^{\mathcal{A}})$ do
not vary along the generators of~$\mathscr{I}^-$, that is, one has
that~$l^a\nabla_ax^{\mathcal{A}}=0$. Accordingly, one has
\begin{align}
C^{\mathcal{A}}=0 \quad \mbox{on}\quad \mathscr{I}^-.\nonumber
\end{align}
Finally, since~$m^a$ and~$\bar{m}^a$ span the tangent space of each
surface~$\mathcal{S}_{u,v}$ one has that
\begin{align}
m^\mu\bmpartial_\mu =P^{\mathcal{A}} \bmpartial_{\mathcal{A}},\nonumber
\end{align}
where the coefficients~$P^{\mathcal{A}}$ are complex functions.

Summarising, we make the following assumption:

\begin{assumption}
\label{Assumption:Frame}
{\em On~$\mathcal{V}$ one can find a Newman-Penrose frame
  $\{l^a,\ n^a,\ m^a,\ \bar{m}^a \}$ of the form:}
\begin{subequations}
\begin{align}
& \bml=\bmpartial_v +C^{\mathcal{A}}\bmpartial_{\mathcal{A}}, \label{framel}\\
& \bmn=Q\bmpartial_u, \label{framen}\\ 
& \bmm=P^{\mathcal{A}} \bmpartial_{\mathcal{A}}. \label{framem}
\end{align}
\end{subequations}
\end{assumption}

\begin{remark}
{\em In view of the normalisation condition~$g_{ab}m^a\bar{m}^b=-1$,
  there are only 3 independent real functions in the
  coefficients~$P^{\mathcal{A}}$. Thus, $Q$, $C^{\mathcal{A}}$
  together with $P^{\mathcal{A}}$ give six scalar fields describing
  the metric. The components $(g^{\mu\nu})$ of the contravariant form
  of the metric $g_{ab}$ are of the form
\begin{align}
(g^{\mu\nu})=\left(\begin{array}{ccc}
0&Q&0 \\
Q&0&QC^{\mathcal{A}} \\
0&QC^{\mathcal{A}}&\sigma^{\mathcal{A}\mathcal{B}} \\
\end{array}\right),\nonumber
\end{align}
where 
\begin{align}
\sigma^{\mathcal{A}\mathcal{B}}\equiv
-(P^{\mathcal{A}}\bar{P^{\mathcal{B}}}+\bar{P^{\mathcal{A}}}P^{\mathcal{B}})
\nonumber
\end{align}
is the (contravariant) induced metric on~$\mathcal{S}_{u,v}$.}
\end{remark}

On~$\mathcal{N}_{\star}^{\prime}$ one has
that~$\bmn=Q\bmpartial_u$. As the coordinates~$(x^{\mathcal{A}})$ are
constant along the generators of~$\mathscr{I}^-$
and~$\mathcal{N}_{\star}^{\prime}$, it follows that
on~$\mathcal{N}_{\star}^{\prime}$ the coefficient~$Q$ is only a
function of~$u$. Thus, without loss of generality one can
reparameterise~$u$ so as to set~$Q=1$
on~$\mathcal{N}_{\star}^{\prime}$.

\subsubsection{Spin connection coefficients}

Direct inspection of the NP
commutators~\eqref{NPCommutator1}-\eqref{NPCommutator4} applied to the
coordinates~$(v,\ u,\ x^2,\ x^3)$ taking into account together with
the remaining gauge freedom in the vector frame of
Assumption~\ref{Assumption:Frame} leads to the following:

\begin{lemma}
\label{Lemma:FrameGauge}
The NP frame of Assumption~\ref{Assumption:Frame} can be chosen such
that
\begin{subequations}
\begin{align}
& \kappa=\nu=\gamma=0, \label{spinconnection1}\\
& \rho=\bar{\rho},\ \ \mu=\bar{\mu}, \label{spinconnection2}\\
& \pi=\alpha+\bar{\beta}, \label{spinconnection3}
\end{align}
\end{subequations}
on~$\mathcal{V}$ and, furthermore, with
\begin{align*}
\epsilon-\bar{\epsilon}=0\ \ \ on\ \ \
\mathcal{V}\cap\mathscr{I}^-. 
\end{align*}
\end{lemma}

\begin{proof}
The proof of this result is analogous to that of Lemma~$1$ in Paper~I.
\end{proof}

\subsubsection{Equations for the frame coefficients}

Taking into account the conditions of the spin connection coefficients
given by~\eqref{spinconnection1}-\eqref{spinconnection3}, the
remaining commutators yield the equations
\begin{subequations}
\begin{align}
  &\Delta C^{\mathcal{A}}=-(\bar{\tau}+\pi)P^{\mathcal{A}}
  -(\tau+\bar{\pi})\bar{P}^{\mathcal{A}}, \label{framecoefficient1} \\
  &\Delta P^{\mathcal{A}}=-\mu P^{\mathcal{A}}
  -\bar{\lambda}\bar{P}^{\mathcal{A}}, \label{framecoefficient2} \\
  &DP^{\mathcal{A}}-\delta C^{\mathcal{A}}=(\bar{\rho}+\epsilon
  -\bar{\epsilon})P^{\mathcal{A}}+\sigma\bar{P}^{\mathcal{A}},
  \label{framecoefficient3} \\
  &DQ=-(\epsilon+\bar{\epsilon})Q,  \label{framecoefficient4}\\
  &\bar{\delta}P^{\mathcal{A}}-\delta\bar{P}^{\mathcal{A}}
  =(\alpha-\bar{\beta})P^{\mathcal{A}}-(\bar{\alpha}-\beta)
  \bar{P}^{\mathcal{A}}, \label{framecoefficient5} \\
  &\delta Q=(\tau-\bar{\pi})Q. \label{framecoefficient6}
\end{align}
\end{subequations}
Equations~\eqref{framecoefficient1}-\eqref{framecoefficient2} allow
the evolution of the frame coefficients~$C^{\mathcal{A}}$
and~$P^{\mathcal{A}}$ off of the null
hypersurface~$\mathcal{N}_{\star}^{\prime}$.
Equations~\eqref{framecoefficient3}-\eqref{framecoefficient4} allow to
evolve the coefficients~$Q$ and~$P^{\mathcal{A}}$ to be evolved along
the null generators
of~$\mathscr{I}^-$. Finally,~\eqref{framecoefficient5}-\eqref{framecoefficient6}
provide constraints for~$Q$ and~$P^{\mathcal{A}}$ on the
spheres~$\mathcal{S}_{u,v}$.

\subsubsection{The conformal gauge conditions}

The conformal Einstein field equations have an in-built conformal
freedom which can be exploited to simplify the geometric setting of
the problem. This freedom allows us, in particular, to select in an
indirect manner the conformal factor via a \emph{conformal gauge
  source function}. More precisely, one has the following:

\begin{lemma}[\textbf{\em conformal gauge conditions}] 
\label{Lemma:ConformalGauge}
Let~$(\tilde{\mathcal{M}},\tilde{\bmg})$ denote an asymptotically
simple spacetime 
satisfying~$\mathit{Ric}[\tilde{\bmg}]=0$ and let~$(\mathcal{M}, \bmg,
\Xi)$ with $\bmg=\Xi^2\tilde{\bmg}$ be a conformal extension for which
the condition $\Xi=0$ describes past null
infinity~$\mathscr{I}^-$. Given the previous NP
frame~\eqref{framel}-\eqref{framem}, the conformal factor $\Xi$ can be
chosen so that given a null
hypersurface~$\mathcal{N}_{\star}^{\prime}$
intersecting~$\mathscr{I}^-$
on~$\mathcal{S}_{\star}\approx\mathbb{S}^2$ one has
\begin{align*}
\Lambda =-\frac{1}{24}R(x), \qquad \mbox{in a neighourhood} \quad
\mathcal{V}\quad  \mbox{of} \quad \mathcal{S}_{\star} \quad \mbox{on}
\quad J^+(\mathcal{S}_{\star})
\end{align*}
where~$R(x)$ is an arbitrary function of the coordinates. Moreover,
one has the additional gauge conditions
\begin{subequations}
\begin{align*}
& \Sigma_2=1, \ \ \mu=\rho=0 \ \ \ \ on\ \ \mathcal{S}_{\star}, \\
& \Phi_{22}=0 \ \ \ \ on\ \ \mathcal{N}_{\star}^{\prime}, \\
& \Phi_{00}=0 \ \ \ \ on\ \ \mathscr{I}^-.
\end{align*}
\end{subequations}
\end{lemma}

\begin{proof}
The proof of the above result is analogous to that of Lemma 18.2
in~\cite{CFEBook}.
\end{proof}

\section{The formulation of the CIVP}
\label{Section:FormulationCIVP}

In this section we analyse general aspects of the asymptotic CIVP for
the conformal vacuum Einstein field equations with data on the null
hypersurface~$\mathscr{I}^-$ and~$\mathcal{N}_{\star}^{\prime}$. A key
feature of the setting is the existence of a hierarchical structure in
the reduced conformal equations which allows to identify the basic
reduced initial data set from which the full initial data
on~$\mathscr{I}^-\cup\mathcal{N}_{\star}^{\prime}$ for the conformal
Einstein field equations can be computed.

\subsection{Specifiable free data}

The following result identifies a possible choice of freely
specifiable initial data for the asymptotic CIVP:

\begin{lemma}[\textbf{\em freely specifiable data for the
    characteristic problem}]
\label{Lemma3}
 Assume that the gauge condition given by Lemma~\ref{Lemma:FrameGauge}
 and Lemma~\ref{Lemma:ConformalGauge} is satisfied in a
 neighbourhood~$\mathcal{V}$ of~$\mathcal{S}_{\star}$. Initial data
 for the conformal Einstein field equations
 on~$\mathscr{I}^-\cup\mathcal{N}_{\star}^{\prime}$ can be computed
 from the reduced data set~$\mathbf{r}_\star$ consisting of:
\begin{align*}
  &\phi_0, \; \epsilon+\bar\epsilon \quad \mbox{on}\quad
  \mathscr{I}^-,\nonumber\\
  &\phi_4\ \ on\ \ \mathcal{N}_{\star}^{\prime}, \nonumber\\
  &\lambda,\ \ \phi_2+\bar{\phi}_2,\ \ \Phi_{20},\ \ \phi_3,\ \
  P^{\mathcal{A}},\ \ on\ \ \mathcal{S}_{\star}. \nonumber
\end{align*}
\end{lemma}

\begin{remark}
{\em An alternative, less symmetric, reduced initial data set is given
  by:
\begin{align*}
  &\lambda, \; \epsilon+\bar\epsilon \quad \mbox{on} \quad
  \mathscr{I}^-,\nonumber\\
  &\phi_4\ \ \mbox{on}\ \ \mathcal{N}_{\star}^{\prime}, \nonumber\\
  &\phi_3,\ \ \phi_2+\bar{\phi}_2,\ \ P^{\mathcal{A}},\ \ \mbox{on}
  \ \ \mathcal{S}_{\star}. \nonumber
\end{align*}
}
\end{remark}

\begin{proof}
The proof of this result follows from a lengthy but straightforward
computation on the same lines of Lemma~$2$ in Paper~I. See also
Section~18.2 in~\cite{CFEBook} and~\cite{Kan96b}.
\end{proof}

\begin{remark}
{\em From the smoothness of the freely specifiable component~$\phi_4$
  on the incoming null hypersurface~$\mathcal{N}'_\star$ and, in
  particular at the intersection with~$\mathscr{I}^-$ it follows that
  that the resulting spacetime will satisfy the \emph{peeling
    behaviour} near past null infinity ---see e.g. \cite{CFEBook},
  Chapter 10. A reformulation of our characteristic problem to future
  null infinity, for which now~$\phi_0$ is freely specifiable data
  along an outgoing null hypersurface, gives rise \emph{mutatis
    mutandi} to solutions with peeling at~$\mathscr{I}^+$. }
\end{remark}

\subsection{Basic local existence}

To apply the theory of CIVP, as discussed say in Section 12.5
of~\cite{CFEBook}, one has to extract a suitable symmetric hyperbolic
evolution system out of the conformal field equations and the
structure equations. The gauge introduced in
Section~\ref{Section:StewartGauge} allows us to perform this
reduction.

\subsubsection{The reduced conformal field equations}

In what follows, we group the four directional derivatives of the
conformal factor and~$s$ as
\begin{align*}
\bm{\Sigma}^t\equiv (\Sigma_1,\ \Sigma_2,\ \Sigma_3,\ \Sigma_4,\ s),
\end{align*}
the components of the frame as 
\begin{align*}
\bm{e}^t\equiv (C^{\mathcal{A}},\ P^{\mathcal{A}},\ Q),
\end{align*}
the spin connection coefficients not fixed by the gauge as 
\begin{align*}
\bm{\Gamma}^t\equiv (\epsilon,\ \pi,\ \beta,\ \mu,\ \alpha,\ \lambda,\
\tau,\ \sigma,\ \rho),
\end{align*}
the independent components of the rescaled Weyl spinor as 
\begin{align*}
\bm{\phi}^t\equiv (\phi_0,\ \phi_1,\ \phi_2,\ \phi_3,\ \phi_4),
\end{align*}
and finally those of tracefree Ricci spinor as
\begin{align*}
\bm{\Phi}^t\equiv (\Phi_{00},\ \Phi_{01},\ \Phi_{11},\ \Phi_{02},
\Phi_{12},\ \Phi_{22}),
\end{align*}
where~${}^t$ denotes the operation of taking the transpose of a column
vector.

A suitable symmetric hyperbolic systems for the four directional
derivatives of the conformal factor, the frame components and the spin
coefficients can be obtained from
equations~\eqref{CFEfirst4}-\eqref{CFEfirst6}, \eqref{CFEfirst6}*,
\eqref{CFEsecond2}, \eqref{framecoefficient1},
\eqref{framecoefficient2}, \eqref{framecoefficient4}
and~\eqref{structureeq1}-\eqref{structureeq4}, \eqref{structureeq6},
\eqref{structureeq7}, \eqref{structureeq11}, \eqref{structureeq13},
\eqref{structureeq15}, respectively. These can be written in the
schematic form,
\begin{align*}
  &\bm{\mathcal{D}}_0\bm{\Sigma}=\bm{B}_0(\bm{\Sigma},
  \bm{\Gamma}, s)\bm{\Sigma}, \\
  & \bm{\mathcal{D}}_1\bm{e}=\bm{B}_1(\bm{\Gamma},\bm{e})\bm{e}, \\
  & \bm{\mathcal{D}}_2\bm{\Gamma}=\bm{B}_2(\bm{\Gamma},\bm{\phi},
  \bm{\Phi})\bm{\Gamma},
\end{align*}
where~$\mathcal{D}_0$, $\mathcal{D}_1$ and~$\mathcal{D}_2$ are given
by
\begin{align*}
  & \bm{\mathcal{D}}_0\equiv
  \mbox{diag}(\Delta,\ \Delta,\ \Delta,\ \Delta,\ \Delta), \\
  & \bm{\mathcal{D}}_1\equiv\mbox{diag}(\Delta,\ \Delta,\ D), \\
  & \bm{\mathcal{D}}_2\equiv\mbox{diag}
  (\Delta,\ \Delta,\ \Delta,\ \Delta,\ \Delta,\ \Delta, D,\ D,\ D),
\end{align*}
and~$\bm{B}_0$, $\bm{B}_1$, $\bm{B}_2$ are smooth matrix-valued
functions of their arguments, whose explicit form will not be
required.

The components of the third conformal field equation~\eqref{CFE3},
equations~\eqref{CFEthird1}-\eqref{CFEthird13} can be reorganised as
\begin{align*}
\bm{\mathcal{D}}_3\bm{\Phi} =\bm{B}_3\bm{\Phi}
\end{align*}
where 
\begin{align*}
\mathcal{D}_3=\left(\begin{array}{cccccc}
\Delta &-\bar{\delta} &0&0&0&0 \\
-\delta &D+2\Delta&-\delta &-\bar{\delta}&0&0 \\
0&-\bar{\delta} &D+\Delta &0&-\bar{\delta}&0 \\
0&-\delta &0&D+\Delta&-\delta &0 \\
0&0&-\delta &-\bar{\delta}&2D+\Delta &-\delta \\
0&0&0&0&-\bar{\delta}&D \\
\end{array}\right)
\end{align*}
with~$\bm{B_3}=\bm{B_3}(\Phi,\ \phi,\ \Gamma,\ \Sigma_i)$. Writing
\begin{align*}
\bm{\mathcal{D}}_3 = \bm{A}^\mu_3 \partial_\mu 
\end{align*}
one has that 
\begin{align*}
& \bm{A}_3^v=\mbox{diag}(0,\ 1,\ 1,\ 1,\ 2,\ 1),\\
& \bm{A}_3^u=\mbox{diag}(Q,\ 2Q,\ Q,\ Q,\ Q,\ 0).
\end{align*}
and 
\begin{align*}
\bm{A}_3^{\mathcal{A}}=\left(\begin{array}{cccccc}
0 &-\bar{P}^{\mathcal{A}} &0&0&0&0 \\
-P^{\mathcal{A}} &C^{\mathcal{A}} &-P^{\mathcal{A}}&-\bar{P}^{\mathcal{A}}&0&0 \\
0&-\bar{P}^{\mathcal{A}} &C^{\mathcal{A}} &0&-\bar{P}^{\mathcal{A}}&0 \\
0&-P^{\mathcal{A}} &0&C^{\mathcal{A}}&-P^{\mathcal{A}} &0 \\
0&0&-P^{\mathcal{A}} &-\bar{P}^{\mathcal{A}}&2C^{\mathcal{A}} &-P^{\mathcal{A}} \\
0&0&0&0&-\bar{P}^{\mathcal{A}}&C^{\mathcal{A}} \\
\end{array}\right). 
\end{align*}
To be specific, the equations above are obtained through the
combinations~\eqref{CFEthird1}+\eqref{CFEthird11},
\eqref{CFEthird10}+2\eqref{CFEthird2}+\eqref{CFEthird12},
\eqref{CFEthird4}*+\eqref{CFEthird8}*,
\eqref{CFEthird3}+\eqref{CFEthird9},
\eqref{CFEthird5}+2\eqref{CFEthird7}+\eqref{CFEthird12}
and~\eqref{CFEthird6}+\eqref{CFEthird13}, respectively. It can
be readily verified that the matrices~$\bm{A}_3^{\mu}$ are
Hermitian. Moreover,
\begin{align*}
\bm{A}_3^{\mu}(l_{\mu}+n_{\mu})=\mbox{diag}(1,\ 3,\ 2,\ 2,\ 3,\ 1)
\end{align*}
is likewise clearly positive definite.

The components of the fourth conformal equation~\eqref{CFE4},
\eqref{CFEforth1}-\eqref{CFEforth8}, can be grouped as
\begin{align*}
\bm{\mathcal{D}}_4\bm{\phi}=\bm{B}_4\bm{\phi}
\end{align*}
where 
\begin{align*}
\bm{\mathcal{D}}_4=\left(\begin{array}{ccccc}
\Delta&-\delta&0&0&0 \\
-\bar{\delta}&D+\Delta&-\delta&0&0 \\
0&-\bar{\delta}&D+\Delta&-\delta&0 \\
0&0&-\bar{\delta}&D+\Delta&-\delta \\
0&0&0&-\bar{\delta}&D \\
\end{array}\right),
\end{align*}
and~$\bm{B}_4=\bm{B}_4(\bm{\phi},\ \bm{\Gamma})$. Again, writing
\begin{align*}
\bm{\mathcal{D}}_4 = \bm{A}^\mu_4 \partial_\mu, 
\end{align*}
one has that
\begin{align*}
& \bm{A}_4^v=\mbox{diag}(0,\ 1,\ 1,\ 1,\ 1), \\
& \bm{A}_4^u=\mbox{diag}(Q,\ Q,\ Q,\ Q,\ 0),
\end{align*}
and 
\begin{align*}
\bm{A}_4^{\mathcal{A}}=\left(\begin{array}{ccccc}
0 &-P^{\mathcal{A}} &0&0&0 \\
-\bar{P}^{\mathcal{A}} &C^{\mathcal{A}} &-P^{\mathcal{A}}&0&0 \\
0&-\bar{P}^{\mathcal{A}} &C^{\mathcal{A}} &-P^{\mathcal{A}}&0 \\
0&0&-\bar{P}^{\mathcal{A}} &C^{\mathcal{A}} &-P^{\mathcal{A}} \\
0&0&0&-\bar{P}^{\mathcal{A}} &C^{\mathcal{A}} \\
\end{array}\right).
\end{align*}
Specifically, the above matricial expressions are obtained from the
combinations~\eqref{CFEforth1}, \eqref{CFEforth2}+\eqref{CFEforth5},
\eqref{CFEforth3}+\eqref{CFEforth6},
\eqref{CFEforth4}+\eqref{CFEforth7} and~\eqref{CFEforth8}. Again, the
matrices~$\bm{A}^{\mu}_4$ can be seen to be Hermitian and, moreover,
one has that
\begin{align*}
\bm{A}_4^{\mu}(l_{\mu}+n_{\mu})=\mbox{diag}(1,\ 2,\ 2,\ 2,\ 1)
\end{align*}
is clearly positive definite.  

We can summarise the above discussion as:

\begin{lemma}
The evolution system
\begin{subequations}
\begin{align}
& \bm{\mathcal{D}}_0\bm{\Sigma}=\bm{B}_0 \bm{\Sigma} \label{reducedeq0}\\
& \bm{\mathcal{D}}_1\bm{e}=\bm{B}_1 \bm{e}, \label{reducedeq1}\\
& \bm{\mathcal{D}}_2\bm{\Gamma}=\bm{B}_2\bm{\Gamma},\label{reducedeq2} \\
& \bm{\mathcal{D}}_3\bm{\Phi} =\bm{B}_3\bm{\Phi}, \label{reducedeq3}\\
& \bm{\mathcal{D}}_4\bm{\phi}=\bm{B}_4\bm{\phi}, \label{reducedeq4} 
\end{align}
\end{subequations}
is symmetric hyperbolic with respect to the direction given by
\begin{align*}
\tau^a = l^a + n^a.
\end{align*}
\end{lemma}

\subsubsection{Computation of the formal derivatives
  on~$\mathcal{N}_\star'\cup\mathscr{I}^-$ and propagation of
  the constraints}

As discussed in Section~12.5 of~\cite{CFEBook}, the existence and
uniqueness of solutions to a CIVP can be obtained via an auxiliary
Cauchy problem on the spacelike hypersurface
\begin{align*}
S\equiv\{p\in\mathbb{R}\times\mathbb{R}\times\mathbb{S}^2\mid
v(p)+u(p)=0\}.
\end{align*}
The formulation of this problem depends crucially on Whitney's
extension theorem, which requires being able to evaluate all
derivatives (interior and transverse) of initial data
on~$\mathcal{N}_\star'\cup\mathscr{I}^-$. One has the following:

\begin{lemma}[\textbf{\em computation of formal derivatives}]
Any arbitrary formal derivatives of the unknown
functions~$\{\bm{\Sigma},\ \bm{e},\ \bm{\Gamma},\ \bm{\Phi},\ \bm{\phi}\}$
on $\mathcal{N}_\star'\cup\mathscr{I}^-$ can be computed from the
prescribed initial data~$\bm{r}_\star$ for the reduced conformal field
equations on~$\mathcal{N}_\star'\cup\mathscr{I}^-$.
\end{lemma}

\begin{proof}
The statement follows from a careful inspection of the conformal field
equations in the present gauge, see Section 18.2 in~\cite{CFEBook} and
\cite{Kan96b} for more details.
\end{proof}

Moreover, using arguments similar to those discussed
in~\cite{CFEBook}, Section~12.5, one can establish the following
result concerning the relation between the reduced equations and the
full conformal vacuum Einstein field equations:

\begin{proposition}[\textbf{\em propagation of the constraints}]
A solution of the reduced conformal field
equations~\eqref{reducedeq0}-\eqref{reducedeq4} on a
neighbourhood~$\mathcal{V}$ of~$\mathcal{S}_{\star}$
on~$J^+(\mathcal{S}_{\star})$ that coincides with initial data
on~$\mathcal{N}_\star'\cup\mathscr{I}^-$ satisfying the conformal
equations gives rise to a solution of the conformal Einstein field
equations~\eqref{CFE1}-\eqref{CFE6} on~$\mathcal{V}$.
\end{proposition}

In addition, one has that:
\begin{corollary}[\textbf{\em preservation of the conformal gauge}] 
Let~$\textbf{u}$ denote a solution to the characteristic problem for
the conformal field equations on a neighbourhood~$\mathcal{V}$
of~$\mathcal{S}_{\star}$ on~$J^+(\mathcal{S}_{\star})$ which satisfies
the gauge conditions given in Lemmas 1 and 2. Then the metric~$\bm{g}$
constructed from the components of the solution~$\bm{u}$ satisfies the
conformal vacuum Einstein field equations in a gauge for
which~$R[\bmg]=R(x)$.
\end{corollary}

\subsubsection{Summary}

Combining the analysis above and applying the theory of the CIVP for
the symmetric hyperbolic systems of Section 12.5 of~\cite{CFEBook},
one obtains the following existence result:

\begin{theorem}[\textbf{\em existence and uniqueness to the standard
  asymptotic characteristic problem}]
\label{Theorem:BasicExistence}
 Given a smooth reduced initial
  data set~$\bm{r}_\star$ for the conformal Einstein field equations
  on~$\mathcal{N}_\star'\cup\mathscr{I}^-$, there exists a unique
  smooth solution of the conformal field equations in a
  neighbourhood~$\mathcal{V}$ of~$\mathcal{S}_{\star}$
  on~$J^+(\mathcal{S}_{\star})$ which implies the prescribed initial
  data on~$\mathcal{N}_\star'\cup\mathscr{I}^-$. Moreover, this
  solution to the conformal Einstein field equations implies, in turn,
  a solution to the vacuum Einstein field equations in a neighbourhood
  of past null infinity.
\end{theorem}

\begin{remark}
{\em Although the region~$\mathcal{V}$ is, in the unphysical picture,
  finite, from the physical point of view, it corresponds to an
  infinite domain of the asymptotic region near past null infinity.  }
\end{remark}

\section{Improved existence result}
\label{Section:ImprovedSetting}

In this section we provide the basic setting for the improved local
existence result for the asymptotic CIVP for the conformal Einstein
field equations using Luk's method. Our analysis builds on the general
formalism developed in Paper~I. Accordingly, in order to avoid tedious
repetition we focus our attention on the novel aspects and specific
challenges raised by the conformal field equations Thus, results and
definitions already appearing in Paper~I are stated without proofs and
where appropriate we quote results directly.

The main difference between the present analysis and that of Paper~I
is that when dealing with the conformal Einstein field equations one
has more unknown equations to take care of. Specifically, we now have
the conformal factor, its derivatives and the components of the
tracefree Ricci tensor.

\subsection{Integration identities and definitions of norms}

In this section we recall some of the basic definitions introduced in
Paper~I.

\smallskip
\noindent
\textbf{Integration.} In what follows let $\phi$ denote a scalar
field. We define the integration on~$\mathcal{S}_{u,v}$ as
\begin{align*}
\int_{\mathcal{S}_{u,v}}\phi\equiv\int_{\mathcal{S}_{u,v}}\phi\mathrm{d}{\bm \sigma}, 
\end{align*}
where~$\mathrm{d}{\bm\sigma}\equiv\sqrt{|\det{\bm \sigma}|}
\mathrm{d}x^2\mathrm{d}x^3$ denotes the volume element of the induced
metric~$\bm\sigma$ on~$\mathcal{S}_{u,v}$. On the truncated causal
diamonds~$\mathcal{D}_{u,v}^{\,t}$, we define integration using the
volume form of the spacetime metric as follows:
\begin{align*}
\int_{\mathcal{D}_{u,v}^{\,t}}\phi&\equiv
\int_0^u\int_0^{\tilde{v}}\int_{\mathcal{S}_{u,v}}\phi\sqrt{|\det {\bm g}|}
\mathrm{d}x^2\mathrm{d}x^3\mathrm{d}v\mathrm{d}u,\\
&=\int_0^u\int_0^{\tilde{v}}\int_{\mathcal{S}_{u,v}}
Q^{-1}\phi\sqrt{|\det{\bm\sigma}|} \mathrm{d}x^2
\mathrm{d}x^3\mathrm{d}v\mathrm{d}u,
\end{align*}
with~$\tilde{v}=\min(v,t-u)$. Integration in the full causal diamond
is denoted in the obvious way with omission of the label~$t$. As there
are no canonical volume forms on~$\mathcal{N}_u$ and~$\mathcal{N}'_v$,
we define integration on these null hypersurfaces by
\begin{align*}
&\int_{\mathcal{N}_u(0,v)}\phi \equiv \int_0^v\int_{\mathcal{S}_{u,v}}
\phi\sqrt{|\det{\bm\sigma}|} \mathrm{d}x^2\mathrm{d}x^3\mathrm{d}v',\\
&\int_{\mathcal{N}'_v(0,u)}\phi \equiv
\int_0^u\int_{\mathcal{S}_{u,v}}\phi\sqrt{|\det{\bm \sigma}|}
\mathrm{d}x^2\mathrm{d}x^3\mathrm{d}u'.
\end{align*}
We also define
\begin{align*}
\int_{\mathcal{N}_u^{\,t}}\phi\equiv\int_{\mathcal{N}_u(I^t)}\phi, \qquad
\int_{\mathcal{N}^{'\,t}_v}\phi\equiv\int_{\mathcal{N}'_v[0,\varepsilon]^t}\phi
\end{align*}
where~$I^t\equiv[0,\min(v_{\bullet},t-u)]$
for~$v_{\bullet}\in\mathbb{R}^+$
denotes the \emph{truncated long integration interval},
and~$[0,\varepsilon]^t\equiv [0,\min(\varepsilon,t-v)]$ the
\emph{truncated short integration interval}.

\smallskip
\noindent
\textbf{Norms.} For~$1\leq p<\infty$, we define the
following~$L^p$-type norms:
\begin{align*}
  ||\phi||_{L^p(\mathcal{S}_{u,v})}\equiv
  \left(\int_{\mathcal{S}_{u,v}}|\phi|^{p}\right)^{1/p},\qquad
  ||\phi||_{L^p(\mathcal{N}_u^{\,t})}\equiv
  \left(\int_{\mathcal{N}_u^{\,t}}|\phi|^{p}\right)^{1/p},\qquad
  ||\phi||_{L^p(\mathcal{N}^{'\,t}_v)}\equiv
  \left(\int_{\mathcal{N}^{'\,t}_v}|\phi|^{p}\right)^{1/p}.
\end{align*}
The corresponding~$L^{\infty}$ norm is defined by
\begin{align*}
||\phi||_{L^{\infty}(\mathcal{S}_{u,v})}\equiv\sup_{\mathcal{S}_{u,v}}|\phi|.
\end{align*}

\subsection{Estimates for the components of the frames and the conformal factor}

A first step in the analysis in Paper~I was the construction of basic
estimates for the components of the frame in terms of the initial
conditions. A similar step is required for the conformal Einstein
field equations. The main difference in this case is that one also
needs to obtain some basic control on the conformal factors and its
derivatives. These estimates are constructed presently.

\subsubsection{Definitions}

Following Paper~I, in the following it will be convenient to define
the following norm measuring the size of the initial value of the
components of the frame:
\begin{align*}
\Delta_{e_{\star}}\equiv\sup_{\mathscr{I}^-,\mathcal{N}'_\star}
\left(|Q|,|Q^{-1}|,|C^{\mathcal{A}}|,|P^{\mathcal{A}}| \right).
\end{align*}
Moreover, we define a scalar
\begin{align*}
\chi\equiv \Delta\log Q,
\end{align*}
which is at the same level of connection coefficients as it is the
derivative of a component of frame. This scalar provides a component
of the connection which does not arise in the original NP formalism,
but is needed to obtain a complete set of~$\Delta$-equations for the
frame. From the definition of~$\chi$ and the NP Ricci identities we
readily obtain
\begin{align}
\label{EqDchi}
D\chi=2\Phi_{11}+\Psi_2+\bar\Psi_2+2\alpha\tau+2\bar\beta\tau
+2\bar\alpha\bar\tau+2\beta\bar\tau+2\tau\bar\tau
-(\epsilon+\bar\epsilon)\chi.
\end{align}
Now, as a consequence of the gauge choice~$Q=1$
on~$\mathcal{N}'_{\star}$, the initial data for~$\chi$
on~$\mathcal{N}'_{\star}$ is 0. For convenience we also define
\begin{align*}
\varpi \equiv \beta-\bar\alpha
\end{align*}
corresponding to the only independent component of the connection on
the spheres~$\mathcal{S}_{u,v}$.

\subsubsection{The estimates}

Following the main strategy in Paper~I, we construct estimates for the
components of the frame and the conformal factor through the analysis
of~$\Delta$-equations under the following bootstrap assumption:

\begin{assumption}[\textbf{\em assumption to control the coefficients
      of the frame and the conformal factor}] Assume that we have a
  solution to the vacuum conformal Einstein field equations in
  Stewart's gauge satisfying
\begin{align*}
  ||
  \{
  \chi, \mu, \lambda, \alpha, \beta, \tau, \Sigma_2
  \}
  ||_{L^{\infty}(\mathcal{S}_{u,v})}\leq\Delta_{\Gamma} 
\label{Assumption:Frame}
\end{align*}
on a truncated causal diamond~$\mathcal{D}_{u,v_\bullet}^{\,t}$,
where~$\Delta_{\Gamma}$ is some (possibly large) constant.
\end{assumption}

The construction of the estimates proceeds along the following steps:

\smallskip
\noindent
\textbf{Step~1.} We integrate~$\chi=\Delta\log Q=\partial_uQ$ in the
short direction so as to obtain
\begin{align*}
|Q-Q_\star|=|\int_0^{\varepsilon}\chi \mbox{d}u|
  \leq\int_0^{\varepsilon}|\chi|\mbox{d}u\leq\int_0^{\varepsilon}
  \Delta_{\Gamma}\mbox{d}u=\Delta_{\Gamma}\varepsilon
\end{align*}
for any~$v$. Then we have
\begin{align*}
||Q-Q_\star||_{L^{\infty}(\mathcal{S}_{u,v})}\leq\Delta_{\Gamma}\varepsilon .
\end{align*}
So there is a constant~$C$ depending on the initial data such that
\begin{align*}
Q^{-1}, \quad Q\leq C(\Delta_{e_{\star}}).
\end{align*}
  
\smallskip
\noindent
\textbf{Step~2.}: To estimate the conformal factor~$\Xi$, we
integrate~$\Xi$ along the short direction
\begin{align*}
  |\Xi|=|\int_0^{\varepsilon}Q^{-1}\Sigma_2 \mbox{d}u|\leq
  C(\Delta_{\Gamma})\varepsilon.
\end{align*}
Accordingly, we have the following lemma:
\begin{lemma}[\textbf{\em control of conformal factor}]
\label{conformal factor}
Under Assumption~\ref{Assumption:Frame}, if~$\varepsilon>0$ is
sufficiently small, there exists a constant~$C$ depending on the size
of the initial data such that
\begin{align*}
||\Xi||_{L^{\infty}(\mathcal{S}_{u,v})}\leq C(\Delta_{\Gamma})\varepsilon
\end{align*}
on~$\mathcal{D}_{u,v_\bullet}^{\,t}$.
\end{lemma}

\smallskip
\noindent
\textbf{Step~3.} Integrating~$P^{\mathcal{A}}$ in the short direction
using equation~\eqref{framecoefficient2} one readily obtains the
following lemma:

 \begin{lemma}[\textbf{\em control on the components of the frame, I}]
\label{controlframe1}
Under Assumption~\ref{Assumption:Frame}, if~$\varepsilon>0$ is
sufficiently small, there exists a constant~$C$ depending on the size
of the initial data such that
\begin{align*}
  ||\{
  P^{\mathcal{A}},(P^{\mathcal{A}})^{-1}
  \}
  ||_{L^{\infty}(\mathcal{S}_{u,v})}\leq
C(\Delta_{e_{\star}}),
\end{align*}
on~$\mathcal{D}_{u,v_\bullet}^{\,t}$. Moreover, since 
\begin{align*}
\sigma^{\mathcal{AB}}=-P^{\mathcal{A}}\bar
P^{\mathcal{B}}-P^{\mathcal{B}}\bar P^{\mathcal{A}}, 
\end{align*}
we also obtain that
\begin{align*}
&|\sigma^{\mathcal{AB}}|, |\sigma_{\mathcal{AB}}|\leq C(\Delta_{e_{\star}}), \\
& c(\Delta_{e_{\star}})\leq\det\sigma\leq C(\Delta_{e_{\star}}).
\end{align*}
Thus, for any vector~$v^a$ on~$\mathcal{S}_{u,v}$, we have that the
norms
\begin{align*}
\int_{\mathcal{S}_{u,v}}(\sigma_{\mathcal{AB}}v^{\mathcal{A}}v^{\mathcal{B}})^{p/2},
\qquad \mbox{and} \qquad \int_{\mathcal{S}_{u,v}}((v^1)^2+(v^2)^2)^{p/2},
\end{align*}
are equivalent. Finally, one also has
\begin{align*}
\sup_{u,v}|\mbox{\em Area}(\mathcal{S}_{u,v})-\mbox{\em Area}
(\mathcal{S}_{0,v})|\leq C\Delta_{\Gamma}\varepsilon.
\end{align*}
\end{lemma}

\smallskip
\noindent
\textbf{Step~4.} Integrating~$C^{\mathcal{A}}$ in the short direction
using equation~\eqref{framecoefficient1} yields the lemma

\begin{lemma}[\textbf{\em control on the components of the frame, II}]
\label{controlframe2}
Choosing~$\varepsilon$ suitably, since~$C^{\mathcal{A}}_\star=0$
on~$\mathscr{I}^-$ one has that
\begin{align*}
||C^{\mathcal{A}}||_{L^{\infty}(\mathcal{S}_{u,v})}\leq C\Delta_{\Gamma}\varepsilon
\end{align*}
on~$\mathcal{D}_{u,v_\bullet}^{\,t}$.
\end{lemma}

\subsection{Some further tools from Paper~I}

In this section we introduce some frequently used inequalities in the
main estimates for the transport equations on hypersurface of
constant~$u$ or~$v$. A detailed proof of these results can be found in
Paper~I.

\subsubsection{General estimates for transport equations}

A key tool in our analysis are formulae that allow us to follow the
evolution of integrals along null hypersurfaces.

\begin{lemma}[\textbf{\em computing derivatives of integrals
      over~$\mathcal{S}_{u,v}$}]
\label{Lemma:DerivativeIntegrals}
  Given a scalar~$\phi$ one has that
\begin{align*}
&\frac{\mathrm{d}}{\mathrm{d}v}\int_{\mathcal{S}_{u,v}}
\phi=\int_{\mathcal{S}_{u,v}}\left(D\phi-2\rho\phi\right),\\
&\frac{\mathrm{d}}{\mathrm{d}u}\int_{\mathcal{S}_{u,v}}
\phi=\int_{\mathcal{S}_{u,v}}Q^{-1}\left(\Delta\phi+2\mu\phi\right),
\end{align*}
along the outgoing and incoming null geodesics that
rule~$\mathcal{N}_{v}'$ and~$\mathcal{N}_{u}$.
\end{lemma}

\begin{lemma}[\textbf{\em integral over causal diamonds of derivatives of
    a scalar}]
\label{Lemma:IntegralIdentities}
Let~$\phi$ be a scalar field. One has
\begin{align*}
  &\int_{\mathcal{D}_{u,v}}D\phi=\int_{\mathcal{N}'_v}Q^{-1}\phi
  -\int_{\mathcal{N}'_0}Q^{-1}\phi
  +\int_{\mathcal{D}_{u,v}}(2\rho+\epsilon+\bar\epsilon)\phi ,\\
  &\int_{\mathcal{D}_{u,v}}\Delta
  \phi=\int_{\mathcal{N}_u}\phi-\int_{\mathcal{N}_0}\phi
  -\int_{\mathcal{D}_{u,v}}2\mu \phi .
\end{align*}
\end{lemma}

\subsubsection{Gr\"onwall-type inequalities}

The following estimates will be used repeatedly in our main
analysis. Again, for proofs and further discussion, the reader is
referred to Paper~I.

\begin{proposition}[\textbf{\em control of the~$L^p$-norm with
      transport equations}]
  \label{Proposition:TransportLpEstimates}
We work under Assumption~\ref{Assumption:Frame}. Let~$\xi$ denote a
tensorfield on~$\mathcal{S}_{u,v}$ of arbitrary type and assume
on~$\mathcal{D}_{u,v_{\bullet}}^{\,t}$ that
\begin{align*}
\sup_{u,v}||\{\rho,\mu\}||_{L^{\infty}(\mathcal{S}_{u,v})}\leq\mathcal{O}.
\end{align*}
Then there
exists~$\varepsilon_{\star}=\varepsilon_{\star}(\Delta_{e_{\star}},\mathcal{O})$
such that for all~$\varepsilon\leq\varepsilon_{\star}$ and for
every~$1\leq p<\infty$, we have the estimates
\begin{align*}
  &||\xi||_{L^p(\mathcal{S}_{u,v})}\leq C(\mathcal{O},I)
  \left(||\xi||_{L^p(S_{u,0})}+\int_0^v
  ||D\xi||_{L^p(\mathcal{S}_{u,v'})}\mathrm{d}v' \right),\\
  &||\xi||_{L^p(\mathcal{S}_{u,v})}\leq2
  \left(||\xi||_{L^p(\mathcal{S}_{0,v})}
  +C(\Delta_{e_{\star}},\mathcal{O})\int_0^u
  ||\Delta\xi||_{L^p(\mathcal{S}_{u',v})}\mathrm{d}u'\right)
\end{align*}
where~$I$ denotes the long direction interval.
\end{proposition}

\begin{proposition}[\textbf{\em supremum norm of solutions to transport
      equations}]
  \label{Proposition:SupremumNormTransportEquations}
  Work under Assumption~\ref{Assumption:Frame}. There
  exists~$\varepsilon_{\star}$ such that for
  all~$\varepsilon\leq\varepsilon_{\star}$, we have
\begin{align*}
  &||\xi||_{L^{\infty}(\mathcal{S}_{u,v})}\leq ||\xi||_{L^{\infty}(\mathcal{S}_{u,0})}
  +\int_0^v||D\xi||_{L^{\infty}(\mathcal{S}_{u,v'})}\mathrm{d}v'; \\
  &||\xi||_{L^{\infty}(\mathcal{S}_{u,v})}\leq ||\xi||_{L^{\infty}(\mathcal{S}_{0,v})}
  +C({\Delta_{e_{\star}}})\int_0^u
  ||\Delta\xi||_{L^{\infty}(\mathcal{S}_{u',v})}\mathrm{d}u'.
\end{align*} 
\end{proposition}

\begin{proposition}
[\textbf{\em $L^4$ norm of solutions to transport equations}]
\label{Proposition:L4NormTransportEquations} Work under
Assumption~\ref{Assumption:Frame} and assume
on~$\mathcal{D}_{u,v_{\bullet}}^{\,t}$ that
\begin{align*}
\sup_{u,v}||\{\rho,\mu\}||_{L^{\infty}(\mathcal{S}_{u,v})}\leq\mathcal{O}.
\end{align*}
There exists~$\varepsilon_{\star}$ such that for
all~$\varepsilon\leq\varepsilon_{\star}$, we have the estimates
\begin{align*}
  &||\xi||_{L^4(\mathcal{S}_{u,v})}\leq C(\Delta_{e_{\star}},\mathcal{O})
  \left(||\xi||_{L^4(\mathcal{S}_{u,0})}+||D\xi||_{L^2(\mathcal{N}_u(0,v))}^{1/2}
  \left(||\xi||^2_{L^2(\mathcal{N}_u(0,v))}
  +||\nablasl\xi||^2_{L^2(\mathcal{N}_u(0,v))}\right)^{1/4}\right), \\
  &||\xi||_{L^4(\mathcal{S}_{u,v})}\leq 2\left(||\xi||_{L^4(\mathcal{S}_{0,v})}
  +C(\Delta_{e_{\star}})||\Delta\xi||_{L^2(\mathcal{N}'_v(0,u))}^{1/2}
  \left(||\xi||^2_{L^2(\mathcal{N}'_v(0,u))}
  +||\nablasl\xi||^2_{L^2(\mathcal{N}'_v(0,u))}
  \right)^{1/4}\right).
\end{align*} 
\end{proposition}

\subsubsection{Sobolev inequalities}

In this subsection we present some useful estimates of the~$L^p$-norms
of a scalar in terms of its~$L^2$-norms and those of its
derivatives. These inequalities are all derived from the isoperimetric
Sobolev inequality of~$\mathcal{S}_{u,v}$. 

\begin{proposition}
[\textbf{\em Sobolev-type inequality, I}]
\label{Proposition:Sobolev} 
Work under Assumption~\ref{Assumption:Frame}. Let~$\phi$ be a scalar
field on~$\mathcal{S}_{u,v}$ which is square-integrable with
square-integrable first covariant derivatives. Then for
each~$2<p<\infty$, $\phi\in L^p(\mathcal{S}_{u,v})$, there
exists~$\varepsilon_\star=\varepsilon_\star(\Delta_{e_\star},\Delta_{\Gamma})$
such that as long as~$\varepsilon\leq\varepsilon_\star$, we have
\begin{align*}
||\phi||_{L^p(\mathcal{S}_{u,v})}\leq
G_p(\bmsigma)\left(||\phi||_{L^2(\mathcal{S}_{u,v})}
+||\nablasl\phi||_{L^2(\mathcal{S}_{u,v})}\right)
\end{align*}
where~$G_p(\bmsigma)$ is a constant also depends on the isoperimetric
constant~$\mathcal{I}(\mathcal{S}_{u,v})$ and~$p$, but is controlled
by some~$C(\Delta_{e_\star})$, $\nablasl$ is the induced connection
on~$\mathcal{S}_{u,v}$ associated with the metric~$\bm\sigma$.
\end{proposition}

Moreover, we have the following:

\begin{proposition}[\textbf{\em Sobolev-type inequality, II}]
  \label{Proposition:EstimateInfinityNorm} Work under
  Assumption~\ref{Assumption:Frame}. There
  exists~$\varepsilon_\star=\varepsilon_\star(\Delta_{e_\star},\Delta_{\Gamma})$
  such that as long as~$\varepsilon\leq\varepsilon_\star$, we have
\begin{align*}
  ||\phi||_{L^{\infty}(\mathcal{S}_{u,v})}\leq G_{p}(\bmsigma)
  \left(||\phi||_{L^p(\mathcal{S}_{u,v})}+||
  \nablasl\phi||_{L^p(\mathcal{S}_{u,v})}\right),
\end{align*}
with~$2<p<\infty$ and~$G_{p}(\bmsigma)\leq C(\Delta_{e_\star})$ as above.
\end{proposition}

\begin{corollary}[\textbf{\em Sobolev-type inequality, III}]
  \label{Corollary:SobolevEmbedding} Work under
  Assumption~\ref{Assumption:Frame}. There
  exists~$\varepsilon_\star=\varepsilon_\star(\Delta_{e_\star},\Delta_{\Gamma})$
  such that as long as~$\varepsilon\leq\varepsilon_\star$, we have
\begin{align*}
  &||\phi||_{L^4(\mathcal{S}_{u,v})}\leq G(\bmsigma)
  \left(||\phi||_{L^2(\mathcal{S}_{u,v})}
  +||\nablasl\phi||_{L^2(\mathcal{S}_{u,v})}\right), \\
  &||\phi||_{L^{\infty}(\mathcal{S}_{u,v})}\leq G(\bmsigma)
  \left(||\phi||_{L^2(\mathcal{S}_{u,v})}
  +||\nablasl\phi||_{L^2(\mathcal{S}_{u,v})}
  +||\nablasl^2\phi||_{L^2(\mathcal{S}_{u,v})}\right),
\end{align*} 
again with~$G(\bmsigma)\leq C(\Delta_{e_\star})$.
\end{corollary}

\begin{remark}
The isoperimetric constant can be shown to be controlled by the area
of the surface~$\mathcal{S}_{u,v}$.
\end{remark}

\section{Main estimates}
\label{Section:Estimates}

In this section we discuss the construction of the main estimates to
obtain the improved existence results for the asymptotic CIVP for the
conformal Einstein field equations. The strategy of the arguments
resemble that in Einstein field equations. As many of the ideas and
techniques are similar to those in Paper~I, as elsewhere, in this
section we focus our attention on the particular aspects of arising
from the use of the conformal Einstein equations.

\subsection{Norms}

The argument in this and subsequent sections relies on the use of a
number of tailor-made norms. We define the following:

\begin{enumerate}[(i)]
\item Norm for the initial value of the connection coefficients, given
  by
\begin{align*}
  \Delta_{\Gamma_{\star}}\equiv\sup_{\mathcal{S}_{u,v}\subset\mathscr{I}^-,\mathcal{N}'_\star}
  \sup_{\Gamma\in\{\mu,\lambda,\rho,\sigma,\alpha,\beta,\tau,\epsilon\}}
  \max\{1,||\Gamma||_{L^{\infty}(\mathcal{S}_{u,v})},\sum_{i=0}^1
  ||\nablasl^i\Gamma||_{L^4(\mathcal{S}_{u,v})},\sum_{i=0}^2
  ||\nablasl^i\Gamma||_{L^2(\mathcal{S}_{u,v})}\}.
\end{align*}

\item Norm for the initial value of the derivative of conformal
  factor~$\Sigma_a$, given by
\begin{align*}
  \Delta_{\Sigma_{\star}}\equiv\sup_{\mathcal{S}_{u,v}\subset\mathscr{I}^-}
  \max\{1,||\Sigma_2||_{L^{\infty}(\mathcal{S}_{u,v})},\sum_{i=0}^1
  ||\nablasl^i\Sigma_2||_{L^4(\mathcal{S}_{u,v})},\sum_{i=0}^2||
  \nablasl^i\Sigma_2||_{L^2(\mathcal{S}_{u,v})}\}.
\end{align*}

\item Norm for the initial value of the Ricci curvature components,
  given by
\begin{align*}
  &\Delta_{\Phi_{\star}}\equiv
    \sup_{\mathcal{S}_{u,v}\subset\mathscr{I}^-,\mathcal{N}'_\star}
    \sup_{\Phi\in\{\Phi_{00},\Phi_{01},\Phi_{02},\Phi_{11},\Phi_{12}\}}
    \max\{1,\sum_{i=0}^1||\nablasl^i\Phi||_{L^4(\mathcal{S}_{u,v})},
    \sum_{i=0}^2||\nablasl^i\Phi||_{L^2(\mathcal{S}_{u,v})}\} \\
  &+\sum_{i=0}^3\sup_{\Phi\in\{\Phi_{00},\Phi_{01},\Phi_{02},\Phi_{11},\Phi_{12}\}}
    ||\nablasl^i\Phi||_{L^2(\mathscr{I}^-)}
    +\sup_{\Phi\in\{\Phi_{01},\Phi_{02},\Phi_{11},\Phi_{12},\Phi_{22}\}}
    ||\nablasl^i\Phi||_{L^2(\mathcal{N}'_\star)}.
\end{align*}

\item Norm for the initial value of the rescaled Weyl curvature
  components, given by
\begin{align*}
  &\Delta_{\phi_{\star}}\equiv
    \sup_{\mathcal{S}_{u,v}\subset\mathscr{I}^-,\mathcal{N}'_\star}
    \sup_{\phi\in\{\phi_0,\phi_1,\phi_2,\phi_3,\phi_4\}}
    \max\{1,\sum_{i=0}^1||\nablasl^i\phi
    ||_{L^4(\mathcal{S}_{u,v})},\sum_{i=0}^2
    ||\nablasl^i\phi||_{L^2(\mathcal{S}_{u,v})}\} \\
  &+\sum_{i=0}^3\sup_{\phi\in\{\phi_0,\phi_1,\phi_2,\phi_3\}}
    ||\nablasl^i\phi||_{L^2(\mathscr{I}^-)}
    +\sup_{\phi\in\{\phi_1,\phi_2,\phi_3,\phi_4\}}
    ||\nablasl^i\phi||_{L^2(\mathcal{N}'_\star)}.
\end{align*}

\item Norm for the components of the Ricci curvature components at
  later null hypersurfaces, given by
\begin{align*}
  \Delta_{\Phi}\equiv\sum_{i=0}^3
    \sup_{\Phi\in\{\Phi_{00},\Phi_{01},\Phi_{02},\Phi_{11},\Phi_{12}\}}
    ||\nablasl^i\Phi||_{L^2(\mathcal{N}_u^t)}+
    \sup_{\Phi\in\{\Phi_{01},\Phi_{02},\Phi_{11},\Phi_{12},\Phi_{22}\}}
    ||\nablasl^i\Phi||_{L^2(\mathcal{N}'_v{}^t)},
\end{align*}
where the suprema in~$u$ and~$v$ are taken
over~$\mathcal{D}^t_{u,v_\bullet}$.

\item Supremum-type norm over the $L^2$-norm of the components of the
  Ricci curvature at spheres of constant~$u$,~$v$, given by
\begin{align*}
  \Delta_{\Phi}(\mathcal{S})\equiv\sum_{i=0}^2
    \sup_{u,v}||\nablasl^i\{\Phi_{00},\Phi_{01},\Phi_{02},
    \Phi_{11},\Phi_{12}\}||_{L^2(\mathcal{S}_{u,v})},
\end{align*}
where the supremum is taken over~$\mathcal{D}^t_{u,v_\bullet}$.

\item Norm for the components of the Weyl tensor at later null
  hypersurfaces, given by
\begin{align*}
  \Delta_{\phi}\equiv\sum_{i=0}^3
    \sup_{\phi\in\{\phi_0,\phi_1,\phi_2,\phi_3\}}
    ||\nablasl^i\phi||_{L^2(\mathcal{N}_u^t)}
    +\sup_{\phi\in\{\phi_1,\phi_2,\phi_3,\phi_4\}}
    ||\nablasl^i\phi||_{L^2(\mathcal{N}'_v{}^t)}, 
\end{align*}
where the suprema in~$u$ and~$v$ are taken
over~$\mathcal{D}^t_{u,v_\bullet}$.

\item Supremum-type norm over the~$L^2$-norm of the components of the
  rescaled Weyl curvature at spheres of constant~$u$, $v$, given by,
\begin{align*} 
  \Delta_{\phi}(\mathcal{S})\equiv\sum_{i=0}^2\sup_{u,v}
    ||\nablasl^i\{\phi_0,\phi_1,\phi_2,\phi_3\}
    ||_{L^2(\mathcal{S}_{u,v})},
\end{align*}
with the supremum taken over~$\mathcal{D}^t_{u,v_\bullet}$ and in
which~$u$ will be taken sufficiently small to apply our estimates.

\end{enumerate}

\subsection{Estimates for the connection coefficients and the
  derivative of conformal factor }

In this subsection, we prove estimates for connection coefficients and
derivatives of the conformal factor. We assume first that the norms of
curvature are bounded and prove that the short range~$\varepsilon$ can
be chosen such that connection coefficients and the derivative of
conformal factor can be controlled by initial data
and~$\Delta_{\Phi}(\mathcal{S})$. This can be achieved by considering
the transport equations. For the connection coefficients~$\tau$
and~$\chi$, we only have their long direction~$D$ equations. However,
the fact that there is no quadratic term in~$\tau$ or~$\chi$
themselves allows us to regard these as linear equations for~$\tau$
and~$\chi$. Then the Gr\"onwall-type inequalities will show us that
these two connection coefficients are bounded. Accordingly, except
for~$\tau$ and~$\chi$, we can analyse the~$\Delta$-equations for the
connection coefficients and the derivatives of conformal. The small
range of~$\varepsilon$ does not let them drift too far from their
initial data on~$\mathscr{I}^-$. Consequently, we find that
although~$\Sigma_1$, $\Sigma_3$ and~$\Sigma_4$ are all small, the
component~$\Sigma_1$ has a different power of~$\varepsilon$
than~$\Sigma_3$ and~$\Sigma_4$ in our estimates.

\begin{proposition}[\textbf{\em control on the supremum norm of the
      connection coefficients and the derivatives of the conformal factor}]
  \label{Proposition:CEFEFirstEstimateConnectionSigma}
  Assume that we have a solution of the vacuum conformal Einstein
  field equations in Stewart's gauge
  in a region~$\mathcal{D}_{u,v_\bullet}^{\,t}$ with 
 \begin{align*}
  \sup_{u,v}||\{\mu, \lambda, \alpha, \beta, \epsilon, \rho, \sigma,
  \tau, \chi,\Sigma_1,\Sigma_2,\Sigma_3,\Sigma_4\}
  ||_{L^\infty(\mathcal{S}_{u,v})}\leq \Delta_{\Gamma,\Sigma}\,,
\end{align*}
for some positive~$\Delta_{\Gamma,\Sigma}$.
Assume also that
\begin{align*}
 \ \Delta_{\Phi}(\mathcal{S})<\infty,\quad  \Delta_{\Phi}<\infty,
  \quad \Delta_{\phi}(\mathcal{S})<\infty,\quad
  \Delta_{\phi}<\infty,\quad
  \sup_{u,v}||\nablasl^i\tau||_{L^2(\mathcal{S}_{u,v})}<\infty,
  \quad i=2,3,
\end{align*}
on the same domain. Then there exists
\begin{align*}
\varepsilon_{\star}=\varepsilon_{\star}(I, \Delta_{e_{\star}},
\Delta_{\Gamma_{\star}}, \sup_{u,v}||\nablasl^2\tau
||_{L^2(\mathcal{S}_{u,v})}, \sup_{u,v}||\nablasl^3\tau
||_{L^2(\mathcal{S}_{u,v})}, \Delta_{\Sigma_{\star}},  \Delta_{\phi},\Delta_{\Phi}),
\end{align*}
such that when~$\varepsilon\leq\varepsilon_{\star}$, we have
\begin{align*}
  &\sup_{u,v}||\{\mu, \lambda, \rho, \sigma, \alpha, \beta,\epsilon\}
  ||_{L^{\infty}(\mathcal{S}_{u,v})}\leq3\Delta_{\Gamma_{\star}}, \\
  &\sup_{u,v}||\{\tau,\chi\}||_{L^{\infty}(\mathcal{S}_{u,v})}\leq
  C(I,\Delta_{e_{\star}},\Delta_{\Gamma_{\star}},\Delta_{\Phi}(\mathcal{S})), \\
  &\sup_{u,v}||\Sigma_2||_{L^{\infty}(\mathcal{S}_{u,v})}\leq3\Delta_{\Sigma_{\star}},\\
  &\sup_{u,v}||\{\Sigma_1,\Sigma_3,\Sigma_4\}||_{L^{\infty}(\mathcal{S}_{u,v})}
  \leq C(I,\Delta_{e_{\star}},\Delta_{\Gamma_{\star}},\Delta_{\Sigma_{\star}},
  \Delta_{\Phi}(\mathcal{S}))\varepsilon, \\
  &\sup_{u,v}||s||_{L^{\infty}(\mathcal{S}_{u,v})}\leq C(\Delta_{e_{\star}},
  \Delta_{\Sigma_{\star}},\Delta_{\Phi})\varepsilon^{1/2},
\end{align*}
on~$\mathcal{D}_{u,v_\bullet}^{\,t}$.

\begin{proof}
$\phantom{X}$

\smallskip
\noindent
\textbf{Basic bootstrap assumption.} Place the following bootstrap
assumptions:
\begin{align*}
  &\sup_{u,v}||\{\mu, \lambda, \rho, \sigma, \alpha, \beta,\epsilon\}
  ||_{L^{\infty}(\mathcal{S}_{u,v})}\leq4\Delta_{\Gamma_{\star}}, \\
  &\sup_{u,v}||\{\Sigma_1,\Sigma_2,\Sigma_3,\Sigma_4\}
  ||_{L^{\infty}(\mathcal{S}_{u,v})}\leq4\Delta_{\Sigma_{\star}}.
\end{align*}

\smallskip
\noindent
\textbf{Estimate for~$\tau$.} First we prove
that~$||\tau||_{L^{\infty}(\mathcal{S}_{u,v})}\leq
C(I,\Delta_{e_{\star}},\Delta_{\Gamma_{\star}},\Delta_{\Phi}(\mathcal{S}))$. We
make use of the $D$-direction equation
of~$\tau$,~\eqref{structureeq2},
\begin{align}
  D\tau=(\epsilon-\bar\epsilon+\rho)\tau+\sigma\bar\tau
  +\bar\pi\rho+\pi\sigma+\Xi\phi_1+\Phi_{01}.
\end{align}
The above equation crucially contains no~$\tau^2$ terms. Making use of
the Sobolev inequality in the
Proposition~\ref{Proposition:EstimateInfinityNorm}, we obtain that
\begin{align*}
&||\phi_i||_{L^{\infty}(\mathcal{S}_{u,v})}\leq\Delta_{\phi}(\mathcal{S})<\infty,
   \qquad \ i=0,1,2,3, \\
&||\Phi_H||_{L^{\infty}(\mathcal{S}_{u,v})}\leq\Delta_{\Phi}(\mathcal{S})<\infty,
\end{align*}  
where~$\Phi_H=\{\Phi_{00},\Phi_{01},\Phi_{02},\Phi_{11},\Phi_{12}\}$. Then
the inequalities in
Proposition~\ref{Proposition:SupremumNormTransportEquations} show that
\begin{align*}
  ||\tau||_{L^{\infty}(\mathcal{S}_{u,v})}&\leq||\tau||_{L^{\infty}(\mathcal{S}_{u,0})}
  +\int_0^v||D\tau||_{L^{\infty}(\mathcal{S}_{u,v'})}\mathrm{d}v'\\
  &\leq\Delta_{\Gamma_{\star}}+C(\Delta_{\Gamma_{\star}},\Delta_{e_{\star}},
  \Delta_{\Phi}(\mathcal{S}))v_\bullet+C(I,\Delta_{\Sigma_{\star}},
  \Delta_{e_{\star}},\Delta_{\phi}(\mathcal{S}))\varepsilon\\
  &\quad+C(\Delta_{\Gamma_{\star}})
  \int_0^v||\tau||_{L^{\infty}({\mathcal{S}}_{u,v'})}\mathrm{d}v'.
\end{align*}  
Now, choosing~$\varepsilon$ sufficiently small, it follows from
Gr\"onwall's inequality that
\begin{align*}
  ||\tau||_{L^{\infty}(\mathcal{S}_{u,v})}\leq
  C(I,\Delta_{e_{\star}},\Delta_{\Gamma_{\star}},\Delta_{\Phi}(\mathcal{S})).
\end{align*}

\smallskip
\noindent
\textbf{Estimate for~$\chi$} In order to estimate~$\chi$, we use
the~$D$-direction equation~\eqref{EqDchi} for~$\chi$. A similar
analysis as before yields
\begin{align*}
  ||\chi||_{L^{\infty}(\mathcal{S}_{u,v})}\leq
  C(I,\Delta_{e_{\star}},\Delta_{\Gamma_{\star}},\Delta_{\Phi}(\mathcal{S})).
\end{align*}

\smallskip
\noindent
\textbf{Estimates for~$\mu$, $\lambda$, $\alpha$, $\beta$ and
  $\epsilon$.}  To estimate the coefficients~$\mu$ and~$\lambda$, we
consider equations~\eqref{structureeq7} and~\eqref{structureeq15}:
\begin{align*}
&\Delta\mu=-\mu^2-\lambda\bar\lambda-\Phi_{22}, \\
&\Delta\lambda=-2\mu\lambda-\Xi\phi_4.
\end{align*}  
Making use of the inequalities in
Proposition~\ref{Proposition:SupremumNormTransportEquations} for the
short direction, we obtain that
\begin{align*}
  ||\mu||_{L^{\infty}(\mathcal{S}_{u,v})}&\leq||\mu||_{L^{\infty}(\mathcal{S}_{0,v})}
  +C(\Delta_{e_\star})\int_0^{\varepsilon}
  ||\Delta\mu||_{L^{\infty}(\mathcal{S}_{u',v})}\mathrm{d}u'\\
  &\leq\Delta_{\Gamma_{\star}}+C(\Delta_{e_\star},\Delta_{\Gamma_{\star}})
  \varepsilon+C(\Delta_{e_\star})\int_0^{u}
  ||\Phi_{22}||_{L^{\infty}(\mathcal{S}_{u',v})}\mathrm{d}u'.
\end{align*}
From the Sobolev and H\"older inequalities, we further find that
\begin{align*}
  \int_0^{u}||\Phi_{22}||_{L^{\infty}(\mathcal{S}_{u',v})}\mathrm{d}u'
  &\leq C(\Delta_{e_\star})\int_0^{u}\sum_{i=0}^2
  ||\nablasl^i\Phi_{22}||_{L^{2}(\mathcal{S}_{u',v})}
  \mathrm{d}u'=\sum_{i=0}^2C(\Delta_{e_\star})\int_0^u
  \left(\int_\mathcal{S}|\nablasl^i\Phi_{22}|^2\right)^{1/2}
  \mathrm{d}u'\\
  &\leq\left(\sum_{i=0}^2C(\Delta_{e_\star})\int_0^u\int_\mathcal{S}
  |\nablasl^i\Phi_{22}|^2\mathrm{d}u'\right)^{1/2}
  \left(\int_0^u 1 \mathrm{d}u'\right)^{1/2}\nonumber\\
  &\leq C(\Delta_{e_\star})\varepsilon^{1/2}
  ||\nablasl^i\Phi_{22}||_{L^{2}(\mathcal{N}'_v(0,u))}.  
\end{align*}
Hence we obtain that
\begin{align*}
  ||\mu||_{L^{\infty}(\mathcal{S}_{u,v})}\leq \Delta_{\Gamma_{\star}}
  +C(\Delta_{e_\star},\Delta_{\Gamma_{\star}})\varepsilon
  +C\varepsilon^{1/2}\Delta_{\Phi}.
\end{align*}
For the connection coefficient~$\lambda$, a similar computation yields
\begin{align*}
  ||\lambda||_{L^{\infty}(\mathcal{S}_{u,v})}&\leq \Delta_{\Gamma_{\star}}
  +C(\Delta_{e_\star},\Delta_{\Gamma_{\star}})\varepsilon
  +C(\Delta_{e_\star})\int_0^u||\Xi\phi_4||_{L^{\infty}(\mathcal{S}_{u',v})}du',\\
  &\leq\Delta_{\Gamma_{\star}}+C(\Delta_{e_\star},\Delta_{\Gamma_{\star}})
  \varepsilon+C\varepsilon^{3/2}\Delta_{\phi}.
\end{align*}
With the same method, we can estimate~$\alpha$, $\beta$ and~$\epsilon$
by using their short direction structure
equations~\eqref{structureeq11}, \eqref{structureeq4}
and~\eqref{structureeq1}:
\begin{align*}
&\Delta\alpha=-\mu\alpha-\lambda\beta-\lambda\tau-\Xi\phi_3, \\
&\Delta\beta=-\bar\lambda\alpha-\mu\beta-\tau\mu-\Phi_{12}, \\
& \Delta\epsilon=-\alpha\bar\pi-\beta\pi-\alpha\tau-\beta\bar\tau
-\pi\tau-\Xi\phi_2-\Phi_{11},
\end{align*}
The details are omitted.

\smallskip
\noindent
\textbf{Estimates for $\rho$ and $\sigma$.}  In this case, the
relevant $\Delta$-transport equations are the structure
equations~\eqref{structureeq9} and~\eqref{structureeq18}:
\begin{align*}
  \Delta\rho&=\bar\delta\tau-\mu\rho-\lambda\sigma-\alpha\tau+\bar\beta\tau
  -\tau\bar\tau-\Xi\phi_2,\\
  \Delta\sigma&=\delta\tau-\bar\lambda\rho-\mu\sigma+\bar\alpha\tau
  -\beta\tau-\tau^2-\Phi_{02}.
\end{align*}
In order to estimate~$\delta\tau$ and~$\bar\delta\tau$, we make use of
the Sobolev inequalities in Corollary~\ref{Corollary:SobolevEmbedding}
and partial integration on~$\mathcal{S}_{u,v}$ to obtain
\begin{align*}
||\nablasl\tau||_{L^{\infty}(\mathcal{S}_{u,v})}\leq
C(\Delta_{e_{\star}})\sum_{i=1}^{3}||\nablasl^i\tau
||_{L^2(\mathcal{S}_{u,v})}\leq C(\Delta_{e_{\star}})
\big(||\tau||_{L^2(\mathcal{S}_{u,v})}+||\nablasl^2
\tau||_{L^2(\mathcal{S}_{u,v})}+||\nablasl^3\tau
||_{L^2(\mathcal{S}_{u,v})}\big).
\end{align*}
Then the H\"older inequality
\begin{align*}
||\tau||_{L^2(\mathcal{S}_{u,v})}\leq||\tau||_{L^{\infty}(\mathcal{S}_{u,v})}
\mbox{Area}(\mathcal{S})^{1/2}
\end{align*}
and the assumptions
\begin{align*}
\sup_{u,v}||\nablasl^2\tau||_{L^2(\mathcal{S}_{u,v})},
\qquad \sup_{u,v}||\nablasl^3\tau||_{L^2(\mathcal{S}_{u,v})}<\infty
\end{align*}
show us
that~$||\nablasl\tau||_{L^{\infty}(\mathcal{S}_{u,v})}$ is
bounded. So we can estimate
the~$||\nablasl\tau||_{L^{\infty}(\mathcal{S}_{u,v})}$ term
in the short direction using equations~\eqref{structureeq9}
and~\eqref{structureeq18} for~$\sigma$ and~$\rho$, respectively.

\smallskip
\noindent
\textbf{Estimate for $s$.} Before estimating the derivatives of the
conformal factor, we first analyse the Friedrich scalar~$s$. Making
use of the conformal Einstein field equations~\eqref{CFEsecond2},
\begin{align*}
\Delta s=-\Sigma_1\Phi_{22}-\Sigma_2\Phi_{11}+\Sigma_3\Phi_{21}+\Sigma_4\Phi_{12}
\end{align*}
and the initial value~$s|_{\mathscr{I}^-}=0$, we readily have that
\begin{align*}
  ||s||_{L^{\infty}(\mathcal{S}_{u,v})}&\leq C(\Delta_{e_{\star}})\int_0^u
  ||\Sigma_2\Phi_{11}-\Sigma_4\Phi_{12}-\Sigma_3\Phi_{21}+\Sigma_1\Phi_{22}
  ||_{L^{\infty}(\mathcal{S}_{u',v})}\mathrm{d}u', \\
  &\leq C(\Delta_{e_{\star}},\Delta_{\Sigma_{\star}},\Delta_{\Phi}
  (\mathcal{S}))\varepsilon+C(\Delta_{e_{\star}},\Delta_{\Sigma_{\star}},
  \Delta_{\Phi})\varepsilon^{1/2}. 
\end{align*}

\smallskip
\noindent
\textbf{Estimate for $\Sigma_2$.} Making use of the conformal Einstein
field equation~\eqref{CFEfirst5}
\begin{align*}
\Delta\Sigma_2=-\Xi\Phi_{22},
\end{align*}
we have that
\begin{align*}
  &||\Sigma_2||_{L^{\infty}(\mathcal{S}_{u,v})}\leq\Delta_{\Sigma_{\star}}
  +C(\Delta_{e_{\star}})\int_0^u||\Xi\Phi_{22}||_{L^{\infty}(\mathcal{S}_{u',v})}
  \mathrm{d}u'
  \leq\Delta_{\Sigma_{\star}}+C(\Delta_{e_{\star}},\Delta_{\Sigma_{\star}},
  \Delta_{\Phi})\varepsilon^{3/2}. 
\end{align*}
Thus, we can choose~$\varepsilon_{\star}$ sufficiently small such
that~$||\Sigma_2||_{L^{\infty}(\mathcal{S}_{u,v})}$ remains close to
its initial value. 

\smallskip
\noindent
\textbf{Estimate for $\Sigma_1$.} Next, equation~\eqref{CFEfirst4} 
\begin{align*}
\Delta\Sigma_1=-\Sigma_4\tau-\Sigma_3\bar\tau+s-\Xi\Phi_{11}
\end{align*}
and the initial value~$\Sigma_1|_{\mathscr{I}^-}=0$, gives that
\begin{align*}
  ||\Sigma_1||_{L^{\infty}(\mathcal{S}_{u,v})}&\leq C(\Delta_{e_{\star}})\int_0^u||
  -\Sigma_4\tau-\Sigma_3\bar\tau+s-\Xi\Phi_{11}||_{L^{\infty}(S_{u',v})}\mathrm{d}u'\\
  &\leq C(I,\Delta_{e_{\star}},\Delta_{\Gamma_{\star}},\Delta_{\Sigma_{\star}},
  \Delta_{\Phi}(\mathcal{S}))\varepsilon
  +C(\Delta_{e_{\star}},\Delta_{\Sigma_{\star}},\Delta_{\Phi})
  \varepsilon^{3/2}+C(I,\Delta_{e_{\star}},\Delta_{\Gamma_{\star}},\Delta_{\Sigma_{\star}},
  \Delta_{\Phi}(\mathcal{S}))\varepsilon^2. 
\end{align*}

\smallskip
\noindent
\textbf{Estimates for~$\Sigma_3$ and~$\Sigma_4$.}
Equation~\eqref{CFEfirst6}
\begin{align*}
\Delta\Sigma_3=-\Sigma_2\tau-\Xi\Phi_{12}
\end{align*}
readily gives that
\begin{align*}
  ||\Sigma_3||_{L^{\infty}(\mathcal{S}_{u,v})}&\leq
  C(\Delta_{e_{\star}})\int_0^u||\Sigma_2\tau
  +\Xi\Phi_{12}||_{L^{\infty}(S_{u',v})}\mathrm{d}u',\\
  &\leq C(I,\Delta_{e_{\star}},\Delta_{\Gamma_{\star}},
  \Delta_{\Sigma_{\star}},\Delta_{\Phi}(\mathcal{S}))\varepsilon
  +C(\Delta_{e_{\star}},\Delta_{\Sigma_{\star}},\Delta_{\Phi}
  (\mathcal{S}))\varepsilon^2. 
\end{align*}
The method is the same for~$\Sigma_4$.

\smallskip
\noindent
\textbf{Concluding the argument.} From the estimates for the NP
connection coefficients and~$\Sigma_{AA'}$ constructed above it
follows that one can choose
\begin{align*}
  \varepsilon_{\star}=\varepsilon_{\star}(I, \Delta_{e_{\star}},
  \Delta_{\Gamma_{\star}},\sup_{u,v}||\nablasl^2
  \tau||_{L^2(\mathcal{S}_{u,v})}, \sup_{u,v}||\nablasl^3\tau
  ||_{L^2(\mathcal{S}_{u,v})}, \Delta_{\Sigma_{\star}},  \Delta_{\phi},
  \Delta_{\Phi},\Delta_{\phi},\Delta_{\Phi}(\mathcal{S})),
\end{align*}
sufficiently small so that
\begin{align*}
 &\sup_{u,v}||\{\mu,\lambda,\alpha,\beta,\epsilon,\rho,\sigma\}
 ||_{L^{\infty}(\mathcal{S}_{u,v})} \leq3\Delta_{\Gamma_\star}, \\
 &\sup_{u,v}||\Sigma_2||_{L^{\infty}(\mathcal{S}_{u,v})} \leq3\Delta_{\Sigma_\star},\\
 &\sup_{u,v}||\{\Sigma_1,\Sigma_3,\Sigma_4\}||_{L^{\infty}(\mathcal{S}_{u,v})}
 \leq C(I,\Delta_{e_{\star}},\Delta_{\Gamma_{\star}},\Delta_{\Sigma_{\star}},
 \Delta_{\Phi}(\mathcal{S}))\varepsilon.
\end{align*}
Accordingly, we have improved our initial bootstrap assumption. 
\end{proof}
\end{proposition}

Now we use the similar method to analyse the~$L^4$ estimate of the
connection coefficients and the derivative of conformal factor.

\begin{proposition}[\textbf{\em control on the~$L^4$-norm of the
    connection coefficients and the derivatives of the conformal factor}]
\label{Proposition:CEFESecondEstimateConnection}
With the same assumptions in
Proposition~\ref{Proposition:CEFEFirstEstimateConnectionSigma}, and
additionally assuming that
\begin{align*}
  \sup_{u,v}||\nablasl\{\mu, \lambda, \alpha, \beta, \epsilon, \rho,
  \sigma,\Sigma_1,\Sigma_2,\Sigma_3,\Sigma_4\}
  ||_{L^4(\mathcal{S}_{u,v})}\leq\Delta_{\Gamma,\Sigma},
\end{align*}
in the truncated diamond~$\mathcal{D}_{u,v_\bullet}^{\,t}$, we find
that there exists
\begin{align*}
  \varepsilon_{\star}=\varepsilon_{\star}
  (I, \Delta_{e_{\star}}, \Delta_{\Gamma_{\star}},
  \sup_{u,v}||\nablasl^2\tau||_{L^2(\mathcal{S}_{u,v})},
  \sup_{u,v}||\nablasl^3\tau||_{L^2(\mathcal{S}_{u,v})},
  \Delta_{\Sigma_{\star}},  \Delta_{\phi},\Delta_{\Phi},
  \Delta_{\phi}(\mathcal{S}),\Delta_{\Phi}(\mathcal{S})),
\end{align*}
such that when~$\varepsilon\leq\varepsilon_{\star}$, we have we have
\begin{align*}
  &\sup_{u,v}||\nablasl\{\tau,\chi\}||_{L^4(\mathcal{S}_{u,v})}
  \leq C(I,\Delta_{e_{\star}},\Delta_{\Gamma_{\star}},\Delta_{\Phi}(\mathcal{S})),\\
  &\sup_{u,v}||\nablasl\{\mu, \lambda, \rho, \sigma, \alpha, \beta,\epsilon\}
  ||_{L^4(\mathcal{S}_{u,v})}\leq3\Delta_{\Gamma_{\star}}, \\
  &\sup_{u,v}||\nablasl\Sigma_2||_{L^4(\mathcal{S}_{u,v})}\leq 3\Delta_{\Sigma_{\star}},\\
  &\sup_{u,v}||\nablasl\{\Sigma_1,\Sigma_3,\Sigma_4\}||_{L^4(\mathcal{S}_{u,v})}
  \leq C(I,\Delta_{e_{\star}},\Delta_{\Gamma_{\star}},\Delta_{\Sigma_{\star}},\Delta_{\Phi}
  (\mathcal{S}))\varepsilon.
\end{align*}
on~$\mathcal{D}_{u,v_\bullet}^{\,t}$.
\end{proposition}

\begin{proof}
$\phantom{X}$

\smallskip
\noindent
\textbf{Basic bootstrap assumption.} 
We make bootstrap assumptions
\begin{align*}
  &\sup_{u,v}||\nablasl(\mu, \lambda, \rho, \sigma, \alpha, \beta,\epsilon)
  ||_{L^4(\mathcal{S}_{u,v})}\leq4\Delta_{\Gamma_{\star}}\\
  &\sup_{u,v}||\nablasl\{\Sigma_1,\Sigma_2,\Sigma_3,\Sigma_4\}
  ||_{L^4(\mathcal{S}_{u,v})}\leq4\Delta_{\Sigma_{\star}}.
\end{align*}

\smallskip
\noindent
\textbf{Estimates for~$\nablasl\tau$.}  First, we estimate
the~$L^4(\mathcal{S}_{u,v})$ norm of~$\nablasl\tau$. Apply
the~$\delta$-derivative to the~$D$-direction equation of~$\tau$ and
the commutator of directional covariant derivatives we obtain
\begin{align*}
   D\delta\tau&=(\rho+\bar\rho+2\epsilon-2\bar\epsilon)
  \delta\tau+\sigma\bar\delta\tau+\sigma\delta\bar\tau
  +\delta(\epsilon-\bar\epsilon+\rho)\tau+\bar\tau\delta\sigma
  +\rho\delta\bar\pi\\
  &\quad+\bar\pi\delta\rho+\sigma\delta\pi
  +\pi\delta\sigma+\Gamma^3+\Sigma_3\phi_1+\Xi\delta\phi_1
  +\Xi\phi_1\Gamma+\delta\Phi_{01}+\Phi_{01}\Gamma.
\end{align*}
In order to estimate the terms
in~$||\Gamma\nablasl\Gamma||_{L^4(\mathcal{S}_{u,v})}$, we use the
H\"older inequality and split it as
\begin{align*}
||\Gamma\nablasl\Gamma||_{L^4(\mathcal{S}_{u,v})}\leq
||\Gamma||_{L^{\infty}(\mathcal{S}_{u,v})}||\nablasl\Gamma||_{L^4(\mathcal{S}_{u,v})}.
\end{align*}
Now, Proposition~\ref{Proposition:CEFEFirstEstimateConnectionSigma}
shows that terms of the
form~$||\Gamma||_{L^{\infty}(\mathcal{S}_{u,v})}$ are, in fact,
bounded. Making use of the Sobolev inequality in
Proposition~\ref{Proposition:Sobolev} and the long direction
inequality in Proposition~\ref{Proposition:L4NormTransportEquations},
we find that
\begin{align*}
  ||\delta\tau||_{L^4(\mathcal{S}_{u,v})}+||\bar\delta\tau||_{L^4(\mathcal{S}_{u,v})}
  &\leq C\big(||\delta\tau||_{L^4(S_{u,0})}+||\bar\delta\tau||_{L^4(S_{u,0})}
  +\int_0^v||D\delta\tau||_{L^4(S_{u,v'})}
  +||D\bar\delta\tau||_{L^4(S_{u,v'})}\mathrm{d}v'\big)\\
  &\leq C(I,\Delta_{e_{\star}},
  \Delta_{\Gamma_{\star}},\Delta_{\Phi}(\mathcal{S}))
  +C(I,\Delta_{e_{\star}},\Delta_{\Gamma_{\star}},\Delta_{\Sigma_{\star}},
  \Delta_{\Phi}(\mathcal{S}),\Delta_{\phi}(\mathcal{S}))\varepsilon\\
  &\quad+C(\Delta_{\Gamma_{\star}})\int_0^v
  \big(||\delta\tau||_{L^4(\mathcal{S}_{u,v'})}
  +||\bar\delta\tau||_{L^4(\mathcal{S}_{u,v'})}\big)\mathrm{d}v'.
\end{align*}
Thus Gr\"onwall's inequality gives
\begin{align*}
 ||\delta\tau||_{L^4(\mathcal{S}_{u,v})}
 +||\bar\delta\tau||_{L^4(\mathcal{S}_{u,v})}\leq
 C(I,\Delta_{e_{\star}},\Delta_{\Gamma_{\star}},\Delta_{\Phi}(\mathcal{S}))
 +C(I,\Delta_{e_{\star}},\Delta_{\Gamma_{\star}},\Delta_{\Sigma_{\star}},
 \Delta_{\Phi}(\mathcal{S}),\Delta_{\phi}(\mathcal{S}))\varepsilon.
\end{align*}
Accordingly, for a small range~$\varepsilon$, we obtain that
\begin{align*}
||\nablasl\tau||_{L^4(\mathcal{S}_{u,v})}\leq
C(I,\Delta_{e_{\star}},\Delta_{\Gamma_{\star}},\Delta_{\Phi}(\mathcal{S})).
\end{align*}

\smallskip
\noindent
\textbf{Estimates for~$\nablasl\chi$.} A direct
computation shows that 
\begin{align*}
  D\delta\chi=(\bar\rho-2\bar\epsilon)\delta\chi+\sigma\bar\delta\chi
  +\Gamma\delta\Gamma-\chi\delta(\epsilon+\bar\epsilon)
  +\Sigma_3(\phi_2+\bar\phi_2)+\Xi\delta(\phi_2+\bar\phi_2)+\delta\Phi_{11},
\end{align*}
where~$\Gamma$ represents a combination of the connection coefficients
whose particular form is not required. A similar equation can be
obtained for~$D\bar\delta\chi$. Using the same method as for the
coefficient~$\tau$, we obtain
that~$||\nablasl\chi||_{L^4(\mathcal{S}_{u,v})}\leq
C(I,\Delta_{e_{\star}},\Delta_{\Gamma_{\star}},\Delta_{\Phi}(\mathcal{S}))$.

\smallskip
\noindent
\textbf{Estimates for~$\nablasl\{\mu, \lambda, \rho,
  \sigma, \alpha, \beta,\epsilon\}$.} Applying the operator~$\Delta$
to equations~\eqref{structureeq7} and~\eqref{structureeq15} we find
that
\begin{align*}
  &\Delta\delta\mu=(\tau-\bar\alpha-\beta)(\mu^2+\lambda\bar\lambda)
  -3\mu\delta\mu-\bar\lambda\bar\delta\mu-\lambda\delta\bar\lambda
  -\bar\lambda\delta\lambda-\delta\Phi_{22}, \\
  &\Delta\delta\lambda=(\tau-\bar\alpha-\beta)
  (2\mu\lambda+\Xi\phi_4)-3\mu\delta\lambda
  -\bar\lambda\bar\delta\lambda-2\lambda\delta\mu
  -\Sigma_3\phi_4-\Xi\delta\phi_4. 
\end{align*}
Now, a direct computation applying
Proposition~\ref{Proposition:TransportLpEstimates} shows that we can
find an~$\varepsilon_{\star}$ such that
when~$\varepsilon\leq\varepsilon_{\star}$, we have
\begin{align*}
||\nablasl\{\mu,\lambda\}||_{L^4(\mathcal{S}_{u,v})}\leq3\Delta_{\Gamma_{\star}}.
\end{align*}
We can estimate~$\delta\alpha$, $\delta\beta$ and~$\delta\epsilon$ by
using the same method. Since we are using the
assumption~$\sup_{u,v}||\nablasl^3\tau||_{L^2(\mathcal{S}_{u,v})}<\infty$
in the truncated causal diamond, the Sobolev inequalities of
Corollary~\ref{Corollary:SobolevEmbedding} show
that~$||\nablasl^2\tau||_{L^4(\mathcal{S}_{u,v})}$ is
finite. Proceeding in a similar way we can estimate~$\delta\sigma$ and
$\delta\rho$ by applying $\delta$ to equations~\eqref{structureeq9}
and~\eqref{structureeq18}.

\smallskip
\noindent 
\textbf{Estimate for~$\nablasl\Sigma_2$.}
Applying~$\delta$ to the short direction equation~\eqref{CFEfirst5}
for~$\Sigma_2$ and using the commutators we find that 
\begin{align*}
  &\Delta\delta\Sigma_2=-\Xi\delta\Phi_{22}-\Sigma_3\Phi_{22}
  +\Xi\Phi_{22}(\tau-\bar\pi)+\Xi\Phi_{21}\bar\lambda+\Xi\Phi_{12}\mu,\\
  &+\Sigma_2(\pi\bar\lambda+\bar\pi\mu)
  -\Sigma_3(\lambda\bar\lambda+\mu^2)-2\Sigma_4\bar\lambda\bar\mu. 
\end{align*}
Similar arguments to the ones used for the connection coefficients
show that
\begin{align*}
  ||\nablasl\Sigma_2||_{L^4(\mathcal{S}_{u,v})}\leq 2\Delta_{\Sigma_{\star}}
  +C(I,\Delta_{e_{\star}},\Delta_{\Gamma_{\star}},\Delta_{\Sigma_{\star}})\varepsilon
  +C(I,\Delta_{e_{\star}},\Delta_{\Gamma_{\star}},\Delta_{\Sigma_{\star}},\Delta_{\Phi}
  (\mathcal{S}))\varepsilon^2+o(\varepsilon^2).
\end{align*}
Accordingly, the~$\varepsilon_{\star}$ can be chosen sufficiently
small to ensure
that~$||\nablasl\Sigma_2||_{L^4(\mathcal{S}_{u,v})}$ is no
more than~$3\Delta_{\Sigma_{\star}}$.

\smallskip
\noindent
\textbf{Estimate for~$\nablasl\Sigma_1$.}  Making use of
the equation for~$\Delta\delta\Sigma_1$:
\begin{align*}
  &\Delta\delta\Sigma_1=-\Sigma_1\Phi_{12}-\Sigma_2\Phi_{01}+\Sigma_4\Phi_{02}
  +s(\bar\pi-\tau)+\Xi\Phi_{11}(\tau-\bar\pi)+\Sigma_4\tau(\tau-\bar\pi) \\
  &+\Sigma_3\bar\tau(\tau-\bar\pi)-\mu\delta\Sigma_1-\bar\tau\delta\Sigma_3
  -\tau\delta\Sigma_4-\Xi\delta\Phi_{11}-\Sigma_4\delta\tau-\Sigma_3\delta\bar\tau
  -\bar\lambda\bar\delta\Sigma_1,
\end{align*}
it follows from the bootstrap assumption, that
\begin{align*}
  ||\nablasl\Sigma_1||_{L^4(\mathcal{S}_{u,v})}\leq
  C(I,\Delta_{e_{\star}},\Delta_{\Gamma_{\star}},\Delta_{\Sigma_{\star}},\Delta_{\Phi}
  (\mathcal{S}))\varepsilon+o(\varepsilon).
\end{align*}

\smallskip
\noindent
\textbf{Estimate for~$\nablasl\Sigma_{3,4}$.} A direct computation
yields the equation
\begin{align*}
  \Delta\delta\Sigma_3=-\Sigma_3\Phi_{12}+\Xi\Phi_{12}(\tau-\bar\pi)
  +\Sigma_2\tau(\tau-\bar\pi)-\tau\delta\Sigma_2-\Sigma_2\delta\tau
  -\mu\delta\Sigma_3-\Xi\delta\Phi_{12}-\bar\lambda\bar\delta\Sigma_3.
\end{align*}
accordingly, one can readily find that 
\begin{align*}
  ||\nablasl\Sigma_3||_{L^4(\mathcal{S}_{u,v})}\leq
  C(I,\Delta_{e_{\star}},\Delta_{\Gamma_{\star}},
  \Delta_{\Sigma_{\star}},\Delta_{\Phi}(\mathcal{S}))
  \varepsilon+o(\varepsilon).
\end{align*}
A similar result holds for~$||\nablasl\Sigma_4||_{L^4(S)}$. It follows
from the previous discussion that when~$\varepsilon$ is suitably
small, we can improve the bootstrap assumption. 

\smallskip
\noindent
\textbf{Concluding the argument.} From the analysis above, it follows
we can choose
\begin{align*}
\varepsilon_{\star}=\varepsilon_{\star}
(I, \Delta_{e_{\star}}, \Delta_{\Gamma_{\star}},
\sup_{u,v}||\nablasl^2\tau||_{L^2(\mathcal{S}_{u,v})}, \sup_{u,v}
||\nablasl^3\tau||_{L^2(\mathcal{S}_{u,v})}, \Delta_{\Sigma_{\star}},
\Delta_{\phi},\Delta_{\Phi},\Delta_{\phi}(\mathcal{S}),
\Delta_{\Phi}(\mathcal{S})),
\end{align*}
sufficiently small so that
\begin{align*}
  &\sup_{u,v}||\nablasl\{\mu, \lambda, \rho, \sigma, \alpha, \beta,\epsilon\}
  ||_{L^4(\mathcal{S}_{u,v})}\leq3\Delta_{\Gamma_{\star}}, \\
  &\sup_{u,v}||\nablasl\Sigma_2||_{L^4(\mathcal{S}_{u,v})}\leq 3\Delta_{\Sigma_{\star}}, \\
  & \sup_{u,v}||\nablasl\{\Sigma_1,\Sigma_3,\Sigma_4\}
  ||_{L^4(\mathcal{S}_{u,v})}\leq
  C(I,\Delta_{e_{\star}},\Delta_{\Gamma_{\star}},\Delta_{\Sigma_{\star}},
  \Delta_{\Phi}(\mathcal{S}))\varepsilon.
\end{align*}
The above estimates improve the bootstrap assumption.
\end{proof}

The discussion of this section is concluded with $L^2$-estimates for
the connection coefficients and the derivative of conformal factor.

\begin{proposition}[\textbf{\em control on the~$L^2$-norm of the
    connection coefficients and the derivatives of the conformal factor}]
\label{Proposition:CEFEThirdEstimateConnection}
Assume that we have a solution of the vacuum conformal Einstein field
equations in Stewart's gauge in a
region~$\mathcal{D}_{u,v_\bullet}^{\,t}$ with
\begin{align*}
  &\sup_{u,v}||\{\mu, \lambda, \alpha, \beta, \epsilon, \rho, \sigma,
  \tau, \chi,\Sigma_1,\Sigma_2,\Sigma_3,\Sigma_4\}
  ||_{L^\infty(S_{u,v})}\leq \Delta_{\Gamma,\Sigma}\,,\\
  & \sup_{u,v}||\nablasl\{\mu, \lambda, \alpha, \beta, \epsilon, \rho,
  \sigma,\Sigma_1,\Sigma_2,\Sigma_3,\Sigma_4\}
  ||_{L^4(\mathcal{S}_{u,v})}\leq\Delta_{\Gamma,\Sigma},\\
  &\sup_{u,v}||\nablasl^2\{\mu, \lambda, \alpha, \beta, \epsilon, \rho,
  \sigma, \tau,\Sigma_1,\Sigma_2,\Sigma_3,\Sigma_4\}
  ||_{L^2(\mathcal{S}_{u,v})}\leq\Delta_{\Gamma,\Sigma},
\end{align*}
for some positive~$\Delta_{\Gamma,\Sigma}$. Assume also
\begin{align*}
  \sup_{u,v}||\nablasl^3\tau||_{L^2(S_{u,v})}<\infty, \qquad
  \Delta_{\Phi}(\mathcal{S})<\infty, \qquad \Delta_{\Phi}<\infty, \qquad
  \Delta_{\phi}(\mathcal{S})<\infty, \qquad \Delta_{\phi}<\infty
\end{align*}
on the same domain. We have that there exists
\begin{align*}
  \varepsilon_\star=\varepsilon_\star
  (I, \Delta_{e_\star},\Delta_{\Gamma_\star},\Delta_{\Sigma_{\star}}
  \sup_{u,v}||\nablasl^3\tau||_{L^2(\mathcal{S}_{u,v})},
  \Delta_{\phi},\Delta_{\Phi},\Delta_{\phi}(\mathcal{S}),
  \Delta_{\Phi}(\mathcal{S})),
\end{align*}
such that when~$\varepsilon\leq\varepsilon_\star$, we have that
\begin{align*}
  & \sup_{u,v}||\nablasl^2\{\tau,\chi\}||_{L^2(\mathcal{S}_{u,v})}\leq
  C(I,\Delta_{e_\star},\Delta_{\Gamma_\star},\Delta_{\Phi}(\mathcal{S})),\\
  &\sup_{u,v}||\nablasl^2\{\mu, \lambda, \alpha, \beta,\epsilon,
  \rho, \sigma\}||_{L^2(\mathcal{S}_{u,v})}\leq3\Delta_{\Gamma_\star},\\
  &\sup_{u,v}||\nablasl^2\Sigma_2||_{L^4(\mathcal{S}_{u,v})}
    \leq 3\Delta_{\Sigma_{\star}},\\
  &\sup_{u,v}||\nablasl^2\{\Sigma_1,\Sigma_3,\Sigma_4\}
   ||_{L^4(\mathcal{S}_{u,v})}\leq
   C(I,\Delta_{e_{\star}},\Delta_{\Gamma_{\star}},\Delta_{\Sigma_{\star}},
   \Delta_{\Phi}(\mathcal{S}))\varepsilon.
\end{align*}
\end{proposition}

\begin{proof}
$\phantom{X}$

\smallskip
\noindent
\textbf{Basic bootstrap assumption.} We make following bootstrap
assumptions:
\begin{align*}
  &\sup_{u,v}||\nablasl^2\{\mu, \lambda, \rho,
    \sigma, \alpha, \beta,\epsilon\}||_{L^2(\mathcal{S}_{u,v})}\leq4\Delta_{\Gamma_{\star}},\\
  &\sup_{u,v}||\nablasl^2\{\Sigma_1,
    \Sigma_2,\Sigma_3,\Sigma_4\}||_{L^2(\mathcal{S}_{u,v})}\leq4\Delta_{\Sigma_{\star}}.
\end{align*}

\smallskip
\noindent
\textbf{Estimates
  for~$||\nablasl^2\tau||_{L^2(\mathcal{S}_{u,v})}$
  and~$||\nablasl^2\chi||_{L^2(\mathcal{S}_{u,v})}$.}
Applying the operator~$\delta$ to the equation for~$D\delta\tau$ and
using the commutators, one obtains following the~$D$-direction
equation of~$\delta^2\tau$:
\begin{align*}
  D\delta^2\tau&=\Gamma\delta^2\tau+\Gamma\delta^2\bar\tau
  +\Gamma\bar\delta\delta\tau+\Gamma\delta\bar\delta\tau
  +\Gamma_1^4+\Gamma_1\delta^2\Gamma_1 \\
  &\quad+\delta\Gamma_1\delta\Gamma_1
  +\Gamma_1^2\delta\Gamma_1+\delta^2\Phi_{01}
  +\Gamma_1\delta\Phi_{01}+\Phi_{01}\delta\Gamma_1
  +\Phi_{01}\Gamma_1^2 \\
  &\quad+\delta\Sigma_3\phi_1+2\Sigma_3\delta\phi_1
  +\Xi\delta\phi_1+\Xi\phi_1\Gamma^2+\Xi\Gamma_1\delta\phi_1
  +\Xi\phi_1\delta\Gamma_1+\Xi\delta^2\phi_1+\Sigma_3\phi_1\Gamma_1,
\end{align*}
where~$\Gamma$ contains a combination of the coefficients~$\rho,
\sigma,\epsilon$,~$\Gamma_1$ contains a combination of~$\tau, \alpha,
\beta, \sigma, \epsilon, \rho $. A similar computation renders
equations for~$D\bar\delta\tau$, $D\delta\bar\delta\tau$. Terms of the
form~$\delta\Gamma_1\delta\Gamma_1$ can be handled using the H\"older
inequality
\begin{align*}
  ||\delta\Gamma_1\delta\Gamma_1||_{L^2(\mathcal{S}_{u,v})}
  \leq||\delta\Gamma_1||_{L^4(\mathcal{S}_{u,v})}
  ||\delta\Gamma_1||_{L^4(\mathcal{S}_{u,v})},
\end{align*}
where Proposition~\ref{Proposition:CEFESecondEstimateConnection} shows
that the bound is finite. The analysis for the
term~$\delta\Sigma_3\phi_1$ is the same. More precisely, one has that
\begin{align*}
||\delta\Sigma_3\phi_1||_{L^2(\mathcal{S}_{u,v})}&\leq||\delta\Sigma_3
||_{L^4(\mathcal{S}_{u,v})}||\phi_1||_{L^4(\mathcal{S}_{u,v})}\nonumber\\
&\leq C(\Delta_{e_{\star}})||\delta\Sigma_3
||_{L^4(\mathcal{S}_{u,v})}\big(||\phi_1||_{L^2(\mathcal{S}_{u,v})}
+||\nablasl\phi_1||_{L^2(\mathcal{S}_{u,v})}\big).
\end{align*}
Similar arguments can be employed in the rest of the terms for the
equation for~$D\delta^2\tau$ so that with the long direction
inequality in Proposition~\ref{Proposition:TransportLpEstimates} we
obtain
\begin{align*}
  ||\delta^2\tau||_{L^2(\mathcal{S}_{u,v})}&\leq  C(I,\Delta_{\Gamma_\star})
  \left(||\delta^2\tau||_{L^2(\mathcal{S}_{u,0})}+\int_0^v||D\delta^2\tau
  ||_{L^2(\mathcal{S}_{u,v'})}\mathrm{d}v' \right), \\
  &\leq C(I,\Delta_{e_{\star}},\Delta_{\Gamma_{\star}},\Delta_{\Phi}(\mathcal{S}))
  +C(I,\Delta_{e_{\star}},\Delta_{\Gamma_{\star}},\Delta_{\Sigma_{\star}},
  \Delta_{\Phi}(\mathcal{S})
  ,\Delta_{\phi}(\mathcal{S}))
  \varepsilon\nonumber\\
  &\quad+C(I,\Delta_{e_{\star}},\Delta_{\Gamma_{\star}})\int_0^v
  ||\nablasl^2\tau||_{L^2(\mathcal{S}_{u,v'})}\mathrm{d}v'.
\end{align*}
Similar estimates can be obtained
for~$\bar\delta^2\tau$,~$\delta\bar\delta\tau$ and
$\bar\delta\delta\tau$. To
estimate~$||\delta\tau||_{L^2(\mathcal{S}_{u,v})}$ we can make use of
the fact that the area of~$\mathcal{S}_{u,v}$ is bounded so that
\begin{align*}
||\delta\tau||_{L^2(\mathcal{S}_{u,v})}\leq
C(\Delta_{e_{\star}},\Delta_{\Gamma_{\star}})||\delta\tau||_{L^4(\mathcal{S}_{u,v})},
\end{align*}
hence, Proposition~\ref{Proposition:CEFESecondEstimateConnection}
shows us that this is also finite. From inequality (32) of Paper~I we
get
\begin{align*}
  ||\nablasl^2\tau||_{L^2(\mathcal{S}_{u,v})}&\leq
  C(I,\Delta_{e_{\star}},\Delta_{\Gamma_{\star}},\Delta_{\Phi}(\mathcal{S}))
  +C(I,\Delta_{e_{\star}},\Delta_{\Gamma_{\star}},\Delta_{\Gamma_{\star}},\Delta_{\Phi}(\mathcal{S})
  ,\Delta_{\phi}(\mathcal{S}))\varepsilon \\
  &\quad+C(I,\Delta_{e_{\star}},\Delta_{\Gamma_{\star}})\int_0^v
  ||\nablasl^2\tau||_{L^2(\mathcal{S}_{u,v'})}\mathrm{d}v',
\end{align*}
so that using Gr\"onwall's inequality we conclude that
\begin{align*}
||\nablasl^2\tau||_{L^2(\mathcal{S}_{u,v})}
\leq C(I,\Delta_{e_{\star}},\Delta_{\Gamma_{\star}},\Delta_{\Phi}(\mathcal{S}))
+C(I,\Delta_{e_{\star}},\Delta_{\Gamma_{\star}},\Delta_{\phi}(\mathcal{S}))\varepsilon.
\end{align*}
Hence, one finds that~$||\nablasl^2\tau||_{L^2(\mathcal{S}_{u,v})}$ is
bounded by a
constant~$C(I,\Delta_{e_{\star}},\Delta_{\Gamma_{\star}},\Delta_{\Phi}(\mathcal{S}))$.
Using the same analysis, we can conclude
that~$||\nablasl^2\chi||_{L^2(\mathcal{S}_{u,v})}$ is bounded.

\smallskip
\noindent
\textbf{Estimates for the the remaining spin connection coefficients.}
Estimates for the remaining connection coefficients can be obtained by
the same methods as in
Proposition~\ref{Proposition:CEFESecondEstimateConnection} namely,
first we compute equations for~$\Delta\delta^2\Gamma$
and~$\Delta\bar\delta\delta\Gamma$, and make use of the short
direction inequality in
Proposition~\ref{Proposition:TransportLpEstimates} to find that
\begin{align*}
||\nablasl^2\{\mu,\lambda,\alpha,\beta,\epsilon,\sigma,\rho\}
||_{L^2(\mathcal{S}_{u,v})}\leq3\Delta_{\Gamma_{\star}}
\end{align*}
for sufficiently small~$\varepsilon$.

\smallskip
\noindent
\textbf{Estimates for~$\nablasl^2\Sigma_2$.} A direct calculation
shows that
\begin{align*}
  \Delta\delta^2\Sigma_2&=\Gamma\delta^2\Sigma_2+\delta\Gamma\delta\Sigma_2
  +\Gamma^2\delta\Sigma_2+\Xi\delta^2\Phi_{22}+\Xi\Phi_{22}\delta\Gamma\\
  &\quad+\Phi_{22}\delta\Sigma_3+\Sigma_3\delta\Phi_{22}+\Xi\Gamma\delta\Phi_{22}
  +\Sigma_3\Phi_{22}\Gamma+\Xi\Phi_{22}\Gamma^2.
\end{align*}
The other short direction equation for the remaining second order
spherical derivatives of~$\Sigma_2$ have the same structure. From
these equations we obtain that
\begin{align*}
||\nablasl^2\Sigma_2||_{L^2(\mathcal{S}_{u,v})}\leq
2\Delta_{\Sigma_{\star}}
+C(I,\Delta_{e_{\star}},\Delta_{\Gamma_{\star}},
\Delta_{\Sigma_{\star}},\Delta_{\Phi}(\mathcal{S}))\varepsilon
+o(\varepsilon).
\end{align*}
The term~$o(\varepsilon)$ arises from the presence of~$\delta^i\Xi$,
$i=0,1,2$. 

\smallskip
\noindent
\textbf{Estimates for~$\nablasl^2\Sigma_1$ and
  $\nablasl^2\Sigma_{3,4}$.} Again, a direct computation
yields the equation
\begin{align*}
\Delta\delta^2\Sigma_1&=\Gamma\delta^2\Sigma_1+\Gamma\delta^2\Sigma'
+\Sigma\Gamma\delta\Gamma+\delta\Sigma\delta\Gamma
+\Sigma\delta^2\Gamma+\Sigma\Gamma^3+\Gamma^2\delta\Sigma+\Gamma\Sigma\Phi,\\
&\quad+\Gamma^2\Xi\Phi+s\Gamma^2+s\delta\Gamma+\Phi\delta\Sigma+\Sigma\delta\Phi
+\Xi\Gamma\delta\Phi+\Xi\Phi\delta\Gamma+\Xi\delta^2\Phi,
\end{align*}
where~$\Gamma$ contains~$\tau$,~$\Sigma$ contains~$\Sigma_2$,
and~$\Sigma'$ does not contain~$\Sigma_1$, while~$\Phi$ does not
contain~$\Phi_{22}$. Making use of the same arguments as
for~$\Sigma_2$, we obtain that
\begin{align*}
  ||\nablasl^2\Sigma_1||_{L^2(\mathcal{S}_{u,v})}\leq
  C(I,\Delta_{e_{\star}},\Delta_{\Gamma_{\star}},\Delta_{\Sigma_{\star}},
  \Delta_{\Phi}(\mathcal{S}))\varepsilon+o(\varepsilon).
\end{align*}
Similar arguments give
\begin{align*}
  ||\nablasl^2\Sigma_3||_{L^2(\mathcal{S}_{u,v})}\leq
  C(I,\Delta_{e_{\star}},\Delta_{\Gamma_{\star}},\Delta_{\Sigma_{\star}},
  \Delta_{\Phi}(\mathcal{S}))\varepsilon+o(\varepsilon).
\end{align*}

\smallskip
\noindent
\textbf{Concluding the argument.} From the analysis in the previous
paragraphs it follows that we can choose
\begin{align*}
  \varepsilon_{\star}=\varepsilon_{\star}(I, \Delta_{e_{\star}}, \Delta_{\Gamma_{\star}},
  \sup_{u,v}||\nablasl^3\tau||_{L^2(\mathcal{S}_{u,v})}, \Delta_{\Sigma_{\star}},
  \Delta_{\phi},\Delta_{\Phi},\Delta_{\phi}(\mathcal{S}),\Delta_{\Phi}(\mathcal{S})),
\end{align*}
sufficiently small so that
\begin{align*}
  &\sup_{u,v}||\nablasl^2\{\mu, \lambda, \rho, \sigma, \alpha, \beta,\epsilon\}
  ||_{L^2(\mathcal{S}_{u,v})}\leq3\Delta_{\Gamma_{\star}}, \\
  &\sup_{u,v}||\nablasl^2\Sigma_2||_{L^2(\mathcal{S}_{u,v})}\leq
  3\Delta_{\Sigma_{\star}}, \\
  &\sup_{u,v}||\nablasl^2\{\Sigma_1,\Sigma_3,\Sigma_4\}
  ||_{L^2(\mathcal{S}_{u,v})}\leq C(I,\Delta_{e_{\star}},\Delta_{\Gamma_{\star}},
  \Delta_{\Sigma_{\star}},\Delta_{\Phi}(\mathcal{S}))\varepsilon.
\end{align*}
The above estimates improve the bootstrap assumptions.
\end{proof}

\subsection{First estimates for the curvature}

Building on the~$L^p$-estimates for the connection coefficients and
the derivative of the conformal factor obtained in the previous
section, we now show that the norms~$\Delta_{\Phi}(\mathcal{S})$
and~$\Delta_{\phi}(\mathcal{S})$ are bounded by the initial data. This
is achieved in the next two propositions.

\begin{proposition}[\textbf{\em basic control of the Ricci curvature}]
  \label{Proposition:CEFEFirstEstimateRicciCurvature}
Assume that we are given a solution to the vacuum CEFEs in
  Stewart's gauge satisfying the assumptions of
  Proposition~\ref{Proposition:CEFEThirdEstimateConnection}. 
Then there exists
\begin{align*}
  \varepsilon_{\star}=\varepsilon_{\star}
  (I,\Delta_{e_{\star}},\Delta_{\Gamma_{\star}}, \Delta_{\Sigma_{\star}},
  \Delta_{\Phi_{\star}},\Delta_{\Phi},\Delta_{\phi},
  \sup_{u,v}||\nablasl^3\tau||_{L^2(\mathcal{S}_{u,v})})
\end{align*}
such that for~$\varepsilon\leq\varepsilon_{\star}$, we have
\begin{align*}
\Delta_{\Phi}(\mathcal{S})<3\Delta_{\Phi_{\star}}.
\end{align*}
on~$\mathcal{D}_{u,v_\bullet}^{\,t}$.
\end{proposition}

\begin{proof}
$\phantom{X}$

\noindent
\textbf{Bootstrap assumption.}  We make the following bootstrap
assumption:
\begin{align*}
  \sup_{u,v}
    ||\nablasl^i\{\Phi_{00},\Phi_{01},\Phi_{02},\Phi_{11},\Phi_{12}\}
    ||_{L^2(\mathcal{S}_{u,v})}\leq4\Delta_{\Phi_{\star}}, \qquad i=0,...,2.
\end{align*}

\smallskip
\noindent
\textbf{$L^2$-norm of the
  components~$\{\Phi_{00},\Phi_{01},\Phi_{02},\Phi_{11},\Phi_{12}\}$.}
We focus on the $L^2(\mathcal{S})$ norm of
$\{\Phi_{00},\Phi_{01},\Phi_{02},\Phi_{11},\Phi_{12}\}$. We will use
the short direction
equations~\eqref{CFEthird1}-\eqref{CFEthird5} to estimate these
components. We take $\Phi_{11}$ as an example. The relevant equation
is in this case given by
\begin{align}
\Delta\Phi_{11}=\delta\Phi_{21}+2\beta\Phi_{21}-\bar\lambda\Phi_{20}
-2\mu\Phi_{11}+\bar\rho\Phi_{22}-\tau\Phi_{21}-\bar\tau\Phi_{21}
+\Sigma_2\bar\phi_2-\Sigma_4\bar\phi_3.
\end{align}
It follows then that
\begin{align*}
  &||\Phi_{11}||_{L^2(\mathcal{S}_{u,v})}\leq2\left(||\Phi_{11}||_{L^2(\mathcal{S}_{0,v})}
  +C(\Delta_{e_{\star}},\Delta_{\Gamma_{\star}})\int_0^u||\Delta\Phi_{11}
  ||_{L^2(\mathcal{S}_{u',v})}\mathrm{d}u'\right),\\
  &\leq2\Delta_{\Phi_{\star}}
  +C(\Delta_{e_{\star}},\Delta_{\Gamma_{\star}})
  \int_0^u\bigg(||\delta\Phi_{21}||_{L^2(\mathcal{S}_{u',v})}
  +||\bar\rho\Phi_{22}||_{L^2(\mathcal{S}_{u',v})}
  +||\Sigma_2\bar\phi_2||_{L^2(\mathcal{S}_{u',v})},\\
  &\quad+||2\beta\Phi_{21}+\bar\lambda\Phi_{20}+2\mu\Phi_{11}
  +\tau\Phi_{21}+\bar\tau\Phi_{21}||_{L^2(\mathcal{S}_{u',v})}
  +||\Sigma_4\bar\phi_3||_{L^2(\mathcal{S}_{u',v})}\bigg) \mathrm{d}u'.
\end{align*}
Using the H\"older inequality, the first three terms can be
transformed to a norm on the light cone. More precisely, one has
\begin{align*}
  \int_0^u||\delta\Phi_{21}||_{L^2(\mathcal{S}_{u',v})}du'&=
  \int_0^u\left(\int_{\mathcal{S}_{u',v}}|\delta\Phi_{21}|^2 \right)^{1/2}
  \mathrm{d}u'\leq\left(\int_0^u\int_{\mathcal{S}_{u',v}}|\delta\Phi_{21}|^2
  \right)^{1/2}\left(\int_0^u1\right)^{1/2}\\
  &\leq\varepsilon^{1/2}||\delta\Phi_{21}||_{L^2(\mathcal{N}_v'(0,u))}
  \leq\Delta_{\Phi}\varepsilon^{1/2}. 
\end{align*}
Similarly, one has that 
\begin{align*}
  &\int_0^u||\bar\rho\Phi_{22}||_{L^2(\mathcal{S}_{u',v})}du'\leq
  C(\Delta_{\Gamma_{\star}},\Delta_{\Phi})\varepsilon^{1/2}, \qquad
  \int_0^u||\Sigma_2\bar\phi_2||_{L^2(\mathcal{S}_{u',v})}du'\leq
  C(\Delta_{\Sigma_{\star}},\Delta_{\phi})\varepsilon^{1/2}. 
\end{align*}
The (large) fourth term can be estimated as follows:
\begin{align*}
  \int_0^u||\Gamma\Phi||_{L^2(\mathcal{S}_{u',v})}\mathrm{d}u'\leq\int_0^u
  ||\Gamma||_{L^{\infty}(\mathcal{S}_{u',v})}||\Phi||_{L^2(\mathcal{S}_{u',v})}\mathrm{d}u'
  \leq C(I,\Delta_{e_\star},\Delta_{\Gamma_\star},\Delta_{\Phi_{\star}})\varepsilon.
\end{align*}
For the last term we have that
\begin{align*}
  &\int_0^u||\Sigma_4\bar\phi_3||_{L^2(\mathcal{S}_{u',v})}du'\leq\int_0^u
  ||\Sigma_4||_{L^{\infty}(\mathcal{S}_{u',v})}
  ||\bar\phi_3||_{L^2(\mathcal{S}_{u',v})}\mathrm{d}u'\\
  &\qquad\leq C\varepsilon||\phi_3||_{L^2(\mathcal{N}_v'(0,u))}
  \varepsilon^{1/2}\leq C(I,\Delta_{e_{\star}},\Delta_{\Gamma_{\star}},
  \Delta_{\Sigma_{\star}},\Delta_{\phi})\varepsilon^{3/2}.
\end{align*}
Hence, we find that
\begin{align*}
  ||\Phi_{11}||_{L^2(\mathcal{S}_{u,v})}&\leq2\Delta_{\Phi_{\star}}
  +C(\Delta_{e_{\star}},\Delta_{\Sigma_{\star}},\Delta_{\Gamma_{\star}},
  \Delta_{\Phi},\Delta_{\phi})\varepsilon^{1/2}
  +C(I,\Delta_{e_\star},\Delta_{\Gamma_\star},\Delta_{\Phi_{\star}})\varepsilon
  \nonumber\\
  &\quad+C(I,\Delta_{e_{\star}},\Delta_{\Gamma_{\star}},
  \Delta_{\Sigma_{\star}},\Delta_{\phi})\varepsilon^{3/2}.
\end{align*}
Accordingly,~$\varepsilon_{\star}$ can be chosen sufficiently small so
that~$||\Phi_{11}||_{L^2(\mathcal{S}_{u,v})}$ is less
than~$3\Delta_{\Phi_{\star}}$, and similarly for the remaining terms.
Consequently, we have improved the bootstrap assumption and finished
Step~1, that is, we have
\begin{align*}
  \sup_{u,v}||(\Phi_{00},\Phi_{01},\Phi_{02},\Phi_{11},\Phi_{12})
  ||_{L^2(\mathcal{S}_{u,v})}\leq3\Delta_{\Phi_{\star}}.
\end{align*}

\smallskip
\noindent
\textbf{Estimates
  for~$||\nablasl\{\Phi_{00},\Phi_{01},\Phi_{02},\Phi_{11},\Phi_{12}\}
  ||_{L^2(\mathcal{S}_{u,v})}$.} We now focus on
the~$L^2(\mathcal{S}_{u,v})$-norm of the first derivative of the Ricci
curvature. We take~$\nablasl\Phi_{11}$ as an example. Using the
results of Proposition~\ref{Proposition:TransportLpEstimates} we
readily have
\begin{align*}
  ||\nablasl\Phi_{11}||_{L^2(\mathcal{S}_{u,v})}&\leq2
  \left(||\nablasl\Phi_{11}||_{L^2(\mathcal{S}_{0,v})}
  +C(\Delta_{e_\star},\Delta_{\Gamma_\star})\int_0^u
  \left(\int_{\mathcal{S}_{u',v}}\Delta\left\langle\nablasl\Phi_{11},
  \nablasl\Phi_{11}\right\rangle_{\sigma} \right)^{1/2}\mathrm{d}u' \right), \\
  &\leq 2\Delta_{\Phi_{\star}}
  +C(\Delta_{e_{\star}},\Delta_{\Gamma_{\star}})\int_0^u
  \left(\int_{\mathcal{S}_{u',v}}|\nablasl\Phi_{11}|(|\Delta\delta\Phi_{11}|
  +|\Delta\bar\delta\Phi_{11}|) \right)^{1/2}\mathrm{d}u' ,
\end{align*}
while the short direction equation for~$\delta\Phi_{11}$ is given by
\begin{align*}
  \Delta\delta\Phi_{11}&=\delta^2\Phi_{21}+\Sigma_2\bar\phi_2(\bar\pi-\tau)
  +\bar\phi_2\delta\Sigma_2+\Sigma_2\delta\bar\phi_2+\Sigma_4\bar\phi_3
  (\tau-\bar\pi)-\bar\phi_3\delta\Sigma_4+\Sigma_4\delta\bar\phi_3 \\
  &\quad+\Phi_{22}\bar\rho(\bar\pi-\tau)+\bar\rho\delta\Phi_{22}
  +\Phi_{22}\delta\bar\rho+\Phi\Gamma^2+\Gamma\delta\Phi+\Phi\delta\Gamma.
\end{align*}
Here the letter~$\Phi$ is used to
denote~$\{\Phi_{20},\Phi_{21},\Phi_{11} \}$. The first term on the
right hand side of the previous equation,~$\delta^2\Phi_{21}$, can be
controlled by
\begin{align*}
  \int_0^u\left(\int_{\mathcal{S}_{u',v}}|\nablasl\Phi_{11}||\nablasl^2\Phi_{21}| \right)^{1/2}
  \mathrm{d}u'&\leq\int_0^u\left(\int_{\mathcal{S}_{u',v}}|\nablasl\Phi_{11}|^2\right)^{1/4}
  \left(\int_{\mathcal{S}_{u',v}}|\nablasl^2\Phi_{21}|^2\right)^{1/4}\mathrm{d}u'\\
  & \leq\sup_{u,v}||\nablasl\Phi_{11}||^{1/2}_{L^2(\mathcal{S}_{u,v})}
  ||\nablasl^2\Phi_{21}||^{1/2}_{L^2(\mathcal{N}_v'(0,u))}\varepsilon^{3/4}\\
  &\leq C(\Delta_{\Phi_{\star}},\Delta_{\Phi})\varepsilon^{3/4} .
\end{align*}
In the case of the terms
\begin{align*}
  \Sigma_2\bar\phi_2(\bar\pi-\tau)+\bar\phi_2\delta\Sigma_2
  +\Sigma_2\delta\bar\phi_2+\Phi_{22}\bar\rho(\bar\pi-\tau)
  +\bar\rho\delta\Phi_{22}+\Phi_{22}\delta\bar\rho,
\end{align*}
the use of the estimates of the curvature of the light cone (rather
than on the sphere) gives a contribution with the same power
of~$\varepsilon$. Furthermore, the terms
\begin{align*}
 \Sigma_4\bar\phi_3(\tau-\bar\pi), \qquad \mbox{and} \qquad
 \Sigma_4\delta\bar\phi_3
\end{align*}
contribute with a power~$\varepsilon^{5/4}$
since~$||\Sigma_4||_{L^{\infty}(\mathcal{S}_{u,v})}$ is controlled
by~$\varepsilon$ in
Proposition~\ref{Proposition:CEFEFirstEstimateConnectionSigma}. For
the term~$\bar\phi_3\delta\Sigma_4$ we have that 
\begin{align*}
  \int_0^u\left(\int_{\mathcal{S}_{u',v}}|\nablasl\Phi_{11}
  ||\bar\phi_3\nablasl\Sigma_4| \right)^{1/2}\mathrm{d}u'
  &\leq\sup_{u,v}||\nablasl\Phi_{11}||^{1/2}_{L^2(\mathcal{S}_{u,v})}
  ||\nablasl\Sigma_4||^{1/2}_{L^2(\mathcal{S}_{u,v})}\int_0^u
  ||\phi_3||^{1/2}_{L^{\infty}(\mathcal{S}_{u',v})}\mathrm{d}u'\\
  &\leq C(\Delta_{e_{\star}})\sup_{u,v}||\nablasl\Phi_{11}||^{1/2}_{L^2(\mathcal{S}_{u,v})}
  ||\nablasl\Sigma_4||^{1/2}_{L^4(\mathcal{S}_{u,v})}\sum_{i=0}^2
  ||\nablasl^i\phi_3||^{1/2}_{L^2(\mathcal{N}_v'(0,u))}\varepsilon^{3/4}\\
  &\leq C(I,\Delta_{e_{\star}},\Delta_{\Gamma_{\star}},
  \Delta_{\Sigma_{\star}},\Delta_{\Phi_{\star}},\Delta_{\phi})\varepsilon^{5/4}. 
\end{align*}
Here we have used the Sobolev inequality and
Proposition~\ref{Proposition:CEFEFirstEstimateConnectionSigma}. Next,
the term~$\Phi\Gamma^2$ gives us
\begin{align*}
  \int_0^u\left(\int_{\mathcal{S}_{u',v}}|\nablasl\Phi_{11}
  ||\Phi\Gamma^2| \right)^{1/2}\mathrm{d}u'&\leq\sum_{i=0}^2
  C(\Delta_{e_{\star}})\sup_{u,v}||\Gamma||_{L^{\infty}(\mathcal{S}_{u,v})}
  ||\nablasl\Phi_{11}||^{1/2}_{L^2(\mathcal{S}_{u,v})}
  ||\nablasl^i\Phi||^{1/2}_{L^2(\mathcal{S}_{u,v})}\varepsilon^{3/4}\\
  &\leq C(I,\Delta_{e_{\star}},\Delta_{\Gamma_{\star}},
  \Delta_{\Phi_{\star}})\varepsilon^{3/4} .
\end{align*}
Terms~$\Gamma\delta\Phi$ and~$\Phi\delta\Gamma$ give a similar
contribution. Putting everything together we find that
\begin{align*}
  ||\nablasl\Phi_{11}||_{L^2(\mathcal{S}_{u,v})}\leq 2\Delta_{\Phi_{\star}}
  +C(I,\Delta_{e_{\star}},\Delta_{\Gamma_{\star}},\Delta_{\Phi_{\star}})
  \varepsilon^{3/4}+C(I,\Delta_{e_{\star}},\Delta_{\Gamma_{\star}},
  \Delta_{\Sigma_{\star}},\Delta_{\Phi_{\star}},\Delta_{\phi})\varepsilon^{5/4},
\end{align*}
so that it is possible to choose a suitably
small~$\varepsilon_{\star}$ to improve the bootstrap assumption.

\smallskip
\noindent
\textbf{Estimates
  for~$||\nablasl^2\{\Phi_{00},\Phi_{01},\Phi_{02},\Phi_{11},\Phi_{12}\}
  ||_{L^2(\mathcal{S}_{u,v})}$.} We present the analysis
of~$\nablasl^2\Phi_{11}$ as an example. The relevant short direction
equation is
\begin{align*}
  \Delta\delta^2\Phi_{11}&=\delta^3\Phi_{21}+\Phi\delta^2\Gamma
  +\Gamma\delta^2\Phi+\delta\Phi\delta\Gamma+\Phi\Gamma\delta\Gamma
  +\Gamma^2\delta\Phi\\
  &\quad+\Phi\Gamma^3 +\Sigma\delta\phi+\phi\delta^2\Sigma
  +\delta\Sigma\delta\phi+\Sigma\Gamma\delta\phi+\phi\Gamma\delta\Sigma
  +\Sigma\phi\Gamma^2.
\end{align*}
Then, making use of the short direction Gr\"onwall-type estimate one
obtains
\begin{align*}
  ||\nablasl^2\Phi_{11}||_{L^2(\mathcal{S}_{u,v})}&\leq
  2\left(||\nablasl^2\Phi_{11}||_{L^2(\mathcal{S}_{0,v})}
  +C(\Delta_{e_{\star}},\Delta_{\Gamma_{\star}})\int_0^u\left(\int_{\mathcal{S}_{u',v}}
  \Delta\left\langle\nablasl^2\Phi_{11},\nablasl^2\Phi_{11}
  \right\rangle_{\sigma}\right)^{1/2} \mathrm{d}u'\right) \\
  &\leq2\Delta_{\Phi_{\star}}+C(\Delta_{e_{\star}},\Delta_{\Gamma_{\star}})
  \int_0^u\left(\int_{\mathcal{S}_{u',v}}|
  \nablasl^2\Phi_{11}|(|\Delta T_1|+|\Delta T_2|)
  \right)^{1/2}\mathrm{d}u',
\end{align*}
where
\begin{align*}
T_1\equiv \bar\delta\bar\delta
\Phi_{11}+(\bar\beta-\alpha)\bar\delta \Phi_{11},
\qquad T_2 \equiv \bar\delta\delta \Phi_{11}+(\alpha-\bar\beta)\delta \Phi_{11}.
\end{align*}
Since~$\Phi$ contains only the
components~$\{\Phi_{11},\Phi_{20},\Phi_{21},\Phi_{22} \}$, we can
analyse terms which contain~$\Phi$ in a similar way. Namely, we make
use of the H\"older inequality to separate the product terms, and then
we make use of the Sobolev embedding theorem. When we encounter the
terms~$\nablasl^i\Phi_{22}$ and~$\nablasl^3\Phi_{21}$, we can make use
of the estimate on the light cone. Finally, a quick inspection of the
remaining terms reveals that only those related to~$\Sigma_2$
contribute to the integration. For example, the
term~$\Sigma_2\delta\phi$ gives
\begin{align*}
  \int_0^u\left(\int_{\mathcal{S}_{u',v}}|\nablasl^2\Phi_{11}
  ||\Sigma_2\delta\phi| \right)^{1/2}\mathrm{d}u'&\leq\int_0^u
  ||\nablasl^2\Phi_{11}||^{1/2}_{L^2(\mathcal{S}_{u',v})}
  ||\Sigma_2||^{1/2}_{L^{\infty}(\mathcal{S}_{u',v})}
  ||\nablasl\phi||^{1/2}_{L^2(\mathcal{S}_{u',v})}\mathrm{d}u'  \\
  &\leq
  \sup_{u,v}||\nablasl^2\Phi_{11}||^{1/2}_{L^2(\mathcal{S}_{u,v})}
  ||\Sigma_2||^{1/2}_{L^{\infty}(\mathcal{S}_{u,v})}
  ||\nablasl\phi||^{1/2}_{L^2(\mathcal{N}_v'(0,u))}\varepsilon^{3/4}\\
  &\leq C(\Delta_{\Sigma_{\star}},
  \Delta_{\Phi_{\star}},\Delta_{\phi})\varepsilon^{3/4}. 
\end{align*}
Similarly, the H\"older and the Sobolev inequalities allow us to
analyse other terms which also controlled by $\varepsilon$. Putting
everything together one finds that  
\begin{align*}
  \sup_{u,v}||\nablasl^2\{\Phi_{00},\Phi_{01},\Phi_{02},\Phi_{11},\Phi_{12}\}
  ||_{L^2(\mathcal{S}_{u,v})}\leq3\Delta_{\Phi_{\star}}.
\end{align*}

\smallskip
\noindent
\textbf{Concluding the argument.}
From the estimates obtained in the previous paragraphs one concludes
that
\begin{align*}
  \sup_{u,v}||\nablasl^i\{\Phi_{00},\Phi_{01},\Phi_{02},\Phi_{11},\Phi_{12}\}
  ||_{L^2(\mathcal{S}_{u,v})}\leq3\Delta_{\Phi_{\star}},\qquad i=0,\dots ,2.
\end{align*}
Hence, we have improved the starting bootstrap assumption.
\end{proof}

Using a similar method, we can obtain the following result:
\begin{proposition}
Assume that we are given a solution to the vacuum CEFEs in
  Stewart's gauge satisfying the same assumptions of
  Proposition~\ref{Proposition:CEFEThirdEstimateConnection}. 
Then there exists
\begin{align*}
  \varepsilon_{\star}=\varepsilon_{\star}
  (I,\Delta_{e_{\star}},\Delta_{\Gamma_{\star}}, \Delta_{\Sigma_{\star}},
  \Delta_{\Phi_{\star}},\Delta_{\Phi},\Delta_{\phi},
  \sup_{u,v}||\nablasl^3\tau||_{L^2(\mathcal{S}_{u,v})})
\end{align*}
such that for~$\varepsilon\leq\varepsilon_{\star}$, we have
\begin{align*}
\Delta_{\phi}(\mathcal{S})<3\Delta_{\phi_{\star}}.
\end{align*}

\end{proposition}
In order to estimate the curvature, we
need~$L^2(\mathcal{S}_{u,v})$-estimates of the connection coefficients
and the derivatives of the conformal factor up to third order. These
estimates can be obtained, except for~$\rho$ and~$\sigma$, by a method
similar to the one used in the previous proof. For these coefficients,
instead of considering their~$n$-direction equations, we make use of
their long direction equations and the Codazzi equation to obtain the
required estimates.

\begin{proposition}[\textbf{\em further control on the $L^2$-norm of
    the connection coefficients}]
\label{Proposition:CEFEImprovedEstimates}

Assume again that we have a solution of the vacuum CEFEs in Stewart's
  gauge in a region~$\mathcal{D}_{u,v_\bullet}^{\,t}$ with
\begin{align*}
  \sup_{u,v}||\{\mu, \lambda, \alpha, \beta, \epsilon, \rho, \sigma,
  \tau, \chi,\Sigma_1,\Sigma_2,\Sigma_3,\Sigma_4\}||_{L^\infty(S_{u,v})}
  &\leq \infty,\\
  \sup_{u,v}||\nablasl\{\mu, \lambda, \alpha, \beta, \epsilon, \rho,
  \sigma,\Sigma_1,\Sigma_2,\Sigma_3,\Sigma_4\} ||_{L^4(\mathcal{S}_{u,v})}
  &\leq\infty,\\
  \sup_{u,v}||\nablasl^2\{\mu, \lambda, \alpha, \beta, \epsilon, \rho,
  \sigma, \tau,\Sigma_1,\Sigma_2,\Sigma_3,\Sigma_4\}||_{L^2(\mathcal{S}_{u,v})}
  &\leq\infty, \\
  \Delta_{\Phi}(\mathcal{S})<\infty, \qquad \Delta_{\Phi}<\infty, \qquad
  \Delta_{\phi}(\mathcal{S})<\infty, \qquad \Delta_{\phi}<\infty
\end{align*}
for some positive~$\Delta_{\Gamma,\Sigma}$ and furthermore that
\begin{align*}
  \sup_{u,v}||\nablasl^3\{\mu,\lambda,\alpha,\beta,\epsilon,\tau,\Sigma_1,
  \Sigma_2,\Sigma_3,\Sigma_4\}||_{L^2(\mathcal{S}_{u,v})}<\infty\,
\end{align*}
on~$\mathcal{D}_{u,v_\bullet}^{\,t}$. Then there
exists~$\varepsilon_{\star}=\varepsilon_{\star}(I,\Delta_{e_{\star}},
\Delta_{\Gamma_{\star}},\Delta_{\Sigma_{\star}},\Delta_{\Phi_{\star}},\Delta_{\Phi},
\Delta_{\phi})$ such that for~$\varepsilon\leq\varepsilon_{\star}$, we
have
\begin{align*}
  &\sup_{u,v}||\nablasl^3\{\mu,\lambda,\alpha,\beta,\epsilon\}
    ||_{L^2(\mathcal{S}_{u,v})}\leq3\Delta_{\Gamma_{\star}},\\
  &\sup_{u,v}||\nablasl^3\{\rho,\sigma\}||_{L^2(\mathcal{S}_{u,v})}
    \leq C(I,\Delta_{e_{\star}},\Delta_{\Gamma_{\star}},
    \Delta_{\Phi_{\star}},\Delta_{\Phi}),\\
  &\sup_{u,v}||\nablasl^3\{\tau,\chi\}||_{L^2(\mathcal{S}_{u,v})}\leq
    C(I,\Delta_{e_{\star}},\Delta_{\Gamma_{\star}},\Delta_{\Phi_{\star}},\Delta_{\Phi}), \\
  &\sup_{u,v}||\nablasl^3\Sigma_2||_{L^2(\mathcal{S}_{u,v})}
    \leq3\Delta_{\Sigma_{\star}}, \\
  &\sup_{u,v}||\nablasl^3\{\Sigma_1,\Sigma_3,\Sigma_4\}
    ||_{L^2(\mathcal{S}_{u,v})}\leq C(I,\Delta_{e_{\star}},\Delta_{\Gamma_{\star}},
    \Delta_{\Sigma_{\star}},\Delta_{\Phi_{\star}})\varepsilon.
\end{align*}
\end{proposition}

\begin{proof}
$\phantom{X}$

\noindent
\textbf{Bootstrap assumption.} We make the following bootstrap
assumption to start the proof: 
\begin{align*}
  &\sup_{u,v}||\nablasl^3\{\mu,\lambda,\alpha,\beta,\epsilon\}
  ||_{L^2(\mathcal{S}_{u,v})}\leq4\Delta_{\Gamma_{\star}},\\
  &\sup_{u,v}||\nablasl^3\tau
  ||_{L^2(\mathcal{S}_{u,v})}\leq\Delta_{\tau}, \\
  &\sup_{u,v}||\nablasl^3\{\Sigma_1,\Sigma_2,\Sigma_3,\Sigma_4\}
  ||_{L^2(\mathcal{S}_{u,v})}\leq4\Delta_{\Sigma_{\star}},
\end{align*}
where~$\Delta_\tau$ is a constant whose value will be fixed later.

\smallskip
\noindent
\textbf{Estimates for~$\rho$ and~$\sigma$.} We first estimate~$\rho$
and~$\sigma$ using the long direction equations~\eqref{structureeq13}
and~\eqref{structureeq6} as we want to avoid the higher derivatives on
sphere in the short direction equations. From the full expression
of~$||\nablasl^3\rho||_{L^2(\mathcal{S}_{u,v})}$ (see Appendix C in
Paper~I), we will analyse four typical terms namely,~$\delta^3\rho$,
$\xi\delta^2\rho$, $\delta\xi\delta\rho$ and~$\xi^2\delta\rho$. For
the term~$\delta^3\rho$, we have
\begin{align*}
  D\delta^3\rho&=\Gamma^5+\Gamma^3\delta\Gamma+\Gamma(\delta\Gamma)^2
  +\Gamma^2\delta^2\Gamma+\delta\Gamma\delta^2\Gamma
  +\rho\delta^3(\epsilon+\bar\epsilon) \\
  &\quad+ (4\epsilon-2\bar\epsilon+5\rho)\delta^3\rho
  +\sigma\delta^3\bar\sigma+\bar\sigma\delta^3\sigma
  +\sigma\delta^2\bar\delta\rho+\delta^3\Phi_{00}.
\end{align*}
The term~$\delta\Gamma\delta^2\Gamma$ can be estimated as
\begin{align*}
  ||\delta\Gamma\delta^2\Gamma||_{L^2(\mathcal{S}_{u,v})}&\leq
  ||\nablasl\Gamma||_{L^4(\mathcal{S}_{u,v})}
  ||\nablasl^2\Gamma||_{L^4(\mathcal{S}_{u,v})}\\
  &\leq C(\Delta_{e_{\star}})
  ||\nablasl\Gamma||_{L^4(\mathcal{S}_{u,v})}
  \left(||\nablasl^2\Gamma||_{L^2(\mathcal{S}_{u,v})}
  +||\nablasl^3\Gamma||_{L^2(\mathcal{S}_{u,v})} \right),
\end{align*}
where~$\Gamma$ contains~$\epsilon$, $\rho$ and~$\sigma$. Then, making
use of the norm of~$\Phi_{00}$ on the long light cone, we find that
\begin{align*}
  \int_0^v||\delta^3\Phi_{00}||_{L^2(\mathcal{S}_{u,v'})}
  \mathrm{d}v'\leq\left(\int_0^v\int_{\mathcal{S}_{u,v'}}|\delta^3\Phi_{00}|^2
    \mathrm{d}v'\right)^{1/2}\left(\int_0^v1\mathrm{d}v'\right)^{1/2} \leq
  C(I)||\nablasl^3\Phi_{00}||_{L^2(\mathcal{N}_u(0,v))}.
\end{align*}
Hence, the long direction of inequality in
Proposition~\ref{Proposition:TransportLpEstimates} yields
\begin{align*}
  ||\delta^3\rho||_{L^2(\mathcal{S}_{u,v})}\leq
  C(I,\Delta_{e_{\star}},\Delta_{\Gamma_{\star}},\Delta_{\Phi_{\star}},\Delta_{\Phi})
  +C(I,\Delta_{\Gamma_\star})\int_0^v\left(||\nablasl^3\rho||_{L^2(\mathcal{S}_{u,v'})}
  +||\nablasl^3\sigma||_{L^2(\mathcal{S}_{u,v'})}\right)\mathrm{d}v'.
\end{align*}
For the term~$\varpi\delta^2\rho$, we readily find that 
\begin{align*}
  ||\varpi\delta^2\rho||_{L^2(\mathcal{S}_{u,v})}\leq||\varpi||_{L^{\infty}(\mathcal{S}_{u,v})}
  ||\nablasl^2\rho||_{L^2(\mathcal{S}_{u,v})}\leq C(\Delta_{\Gamma_{\star}}).
\end{align*}
Similar estimates can be found for~$\delta\varpi\delta\rho$
and~$\varpi^2\delta\rho$. Hence, we conclude that
\begin{align*}
  ||\nablasl^3\rho||_{L^2(\mathcal{S}_{u,v})}\leq
  C(I,\Delta_{e_{\star}},\Delta_{\Gamma_{\star}},\Delta_{\Phi_{\star}},\Delta_{\Phi})
  +C(I,\Delta_{\Gamma_{\star}})\int_0^v\left(||\nablasl^3\rho||_{L^2(\mathcal{S}_{u,v'})}
  +||\nablasl^3\sigma||_{L^2(\mathcal{S}_{u,v'})}\right)\mathrm{d}v'.
\end{align*}
From here, using Gr\"onwall's inequality one finds that 
\begin{align*}
  ||\nablasl^3\rho||_{L^2(\mathcal{S}_{u,v})}\leq
  C(I,\Delta_{e_{\star}},\Delta_{\Gamma_{\star}},\Delta_{\Phi_{\star}},\Delta_{\Phi})
  +C(I,\Delta_{\Gamma_{\star}})\int_0^v
  ||\nablasl^3\sigma||_{L^2(\mathcal{S}_{u,v'})}\mathrm{d}v'.
\end{align*}
In order to estimate~$||\nablasl^3\sigma||_{L^2(\mathcal{S}_{u,v})}$,
we need to control the third order derivatives of~$\sigma$ for
example~$||\delta^3\sigma||_{L^2(\mathcal{S}_{u,v})}$. Using
integration by parts and the structure equation (Codazzi
equation)~\eqref{structureeq17} one finds that
\begin{align*}
  ||\delta^3\sigma||_{L^2(\mathcal{S}_{u,v})}&\leq||\delta^3\rho||_{L^2(\mathcal{S}_{u,v})}
  +C(I,\Delta_{e_{\star}},\Delta_{\Gamma_{\star}},\Delta_{\Phi_{\star}})
  +||\nablasl^2\Phi_{01}||_{L^2(\mathcal{S}_{u,v})} \\
  &\quad+C(I,\Delta_{e_{\star}},\Delta_{\Gamma_{\star}},\Delta_{\Sigma_{\star}},
  \Delta_{\Phi_{\star}})\sum_{i=0}^3||\nablasl^i\phi_0||_{L^2(\mathcal{S}_{u,v})}\varepsilon,
\end{align*}
so that, in fact, one has
\begin{align*}
  ||\nablasl^3\rho||_{L^2(\mathcal{S}_{u,v})}&\leq
  C(I,\Delta_{e_{\star}},\Delta_{\Gamma_{\star}},\Delta_{\Phi_{\star}},\Delta_{\Phi})
  +C(I,\Delta_{e_{\star}},\Delta_{\Gamma_{\star}},\Delta_{\Sigma_{\star}},
  \Delta_{\Phi_{\star}},\Delta_{\phi})\varepsilon \\
  &\quad+C(I,\Delta_{\Gamma_{\star}})\int_0^v
  ||\nablasl^3\rho||_{L^2(\mathcal{S}_{u,v'})}\mathrm{d}v'. 
\end{align*}
Now, using Gr\"onwall's inequality one concludes that
\begin{align*}
  ||\nablasl^3\rho||_{L^2(\mathcal{S}_{u,v})}\leq
  C(I,\Delta_{e_{\star}},\Delta_{\Gamma_{\star}},\Delta_{\Phi_{\star}},\Delta_{\Phi})
  +C(I,\Delta_{e_{\star}},\Delta_{\Gamma_{\star}},\Delta_{\Sigma_{\star}},
  \Delta_{\Phi_{\star}},\Delta_{\phi})\varepsilon,
\end{align*}
so that~$||\nablasl^3\rho||_{L^2(S)_{u,v}}$ is bounded. Moreover, one
has that
\begin{align*}
  ||\nablasl^3\sigma||_{L^2(\mathcal{S}_{u,v})}\leq
  C(I,\Delta_{e_{\star}},\Delta_{\Gamma_{\star}},\Delta_{\Phi_{\star}},\Delta_{\Phi})
  +C(I,\Delta_{e_{\star}},\Delta_{\Gamma_{\star}},\Delta_{\Sigma_{\star}},
  \Delta_{\Phi_{\star}},\Delta_{\phi})\varepsilon.
\end{align*}

\smallskip
\noindent
\textbf{Estimates for~$\tau$ and~$\chi$.} The~$\Delta$-equation
for~$\nablasl^3\tau$ can be obtained from the structure
equation~\eqref{structureeq2} and the commutator relationship. More
precisely, one has that
\begin{align*}
  D\delta^3\tau&=\delta^3(\Xi\phi_1)+\delta^3\Phi_{01}+\Gamma\delta^3\Gamma_1
  +\Gamma\delta^3\tau+\Gamma\delta^2\Psi_1\\
  &\quad+\delta\Gamma\delta^2\Gamma +\Gamma^2\delta^2\Gamma
  +\Gamma^3\delta\Gamma+\Gamma(\delta\Gamma)^2,
\end{align*}
where~$\Gamma_1$ contains~$\epsilon,\alpha,\beta,\rho$
and~$\sigma$. Then, using the bootstrap assumption and the definition
of~$\Delta_{\Psi}$, we obtain
\begin{align*}
||\nablasl^3\tau||_{L^2(\mathcal{S}_{u,v})}&\leq
C(I,\Delta_{e_{\star}},\Delta_{\Gamma_{\star}},\Delta_{\Phi_{\star}},\Delta_{\Phi})
+C(I,\Delta_{e_{\star}},\Delta_{\Sigma_{\star}},\Delta_{\Gamma_{\star}},
\Delta_{\Phi_{\star}},\Delta_{\phi})\varepsilon\nonumber\\
&\quad+C(I,\Delta_{\Gamma_\star})\int_0^v
||\nablasl^3\tau||_{L^2(\mathcal{S}_{u,v'})}\mathrm{d}v',
\end{align*}
so that using Gr\"onwall's inequality we conclude that
\begin{align*}
  ||\nablasl^3\tau||_{L^2(\mathcal{S}_{u,v})}\leq
  C(I,\Delta_{e_{\star}},\Delta_{\Gamma_{\star}},\Delta_{\Phi_{\star}},\Delta_{\Phi})
  +C(\Delta_{e_{\star}},\Delta_{\Sigma_{\star}},\Delta_{\Gamma_{\star}},
  \Delta_{\Phi_{\star}},\Delta_{\phi})\varepsilon.
\end{align*}
We can then choose the constant~$\Delta_\tau$ larger than the right
side above so as to improve the bootstrap assumption. The estimate
of~$\chi$ is similar:
\begin{align*}
  ||\nablasl^3\chi||_{L^2(\mathcal{S}_{u,v})}\leq
  C(I,\Delta_{e_{\star}},\Delta_{\Gamma_{\star}},\Delta_{\Phi_{\star}},\Delta_{\Phi})
  +C(\Delta_{e_{\star}},\Delta_{\Sigma_{\star}},\Delta_{\Gamma_{\star}},\Delta_{\Phi_{\star}},
  \Delta_{\phi})\varepsilon.
\end{align*}

\smallskip
\noindent
\textbf{Estimates for the the remaining spin connection coefficients.}
To obtain the estimates for
\begin{align*}
||\nablasl^3\{\mu,\lambda,\alpha,\beta,\epsilon\}||_{L^2(\mathcal{S}_{u,v})},
\end{align*}
we make use of their short direction equations. Since the proof are
similar, we only show the details of~$\epsilon$ as a representative
example. In this case the relevant equation is
\begin{align*}
  \Delta\delta^3\epsilon&=-\delta^3(\Xi\phi_2+\Phi_{12})+\Gamma\delta^3\Gamma_1
  +\Gamma\delta^3\epsilon+\delta\Gamma\delta^2\Gamma +\Gamma^2\delta^2\Gamma
  +\Gamma^3\delta\Gamma+\Gamma(\delta\Gamma)^2+\Gamma^5,
\end{align*}
where~$\Gamma_1$ does not contain~$\epsilon$. We can then make use of
the short inequality in
Proposition~\ref{Proposition:TransportLpEstimates} and obtain that
\begin{align*}
||\nablasl^3\epsilon||_{L^2(\mathcal{S}_{u,v})}\leq
2\Delta_{\Gamma_{\star}}+C(I,\Delta_{e_{\star}},\Delta_{\Gamma_{\star}},
\Delta_{\Phi_{\star}},\Delta_{\Phi})\varepsilon^{3/4}+o(\varepsilon^{3/4}). 
\end{align*}
Choosing the integral range sufficiently small we conclude that
\begin{align*}
||\nablasl^3\epsilon||_{L^2(\mathcal{S}_{u,v})}\leq 3\Delta_{\Gamma_{\star}}.
\end{align*}
The estimates
of~$||\nablasl^3\{\mu,\lambda,\alpha,\beta\}||_{L^2(\mathcal{S}_{u,v})}$
are are similar. Hence, we have improved the bootstrap assumption for
the connection coefficients.

\smallskip
\noindent
\textbf{Estimates for~$\nablasl^3\Sigma_2$.} The short direction
equation for~$\delta^3\Sigma_2$ can be analysed by the same
method. Starting from
\begin{align*}
  \Delta\delta^3\Sigma_2&=\Gamma^3\delta\Sigma_2
  +\Gamma\delta\Sigma_2\delta\Gamma
  +\Gamma^2\delta^2\Sigma_2+\delta\Gamma\delta^2\Sigma_2
  +\delta\Sigma_2\delta^2\Gamma
  +\Gamma\delta^3\Sigma_2 \\
  &\quad+\sum_{i_1+...+i_4=3}\delta^{i_1}\Xi\delta^{i_2}
  \Gamma^{i_3}\delta^{i_4}\Phi_{22},
\end{align*}
where~$\Gamma$ contains~$\tau$ and it is observed that the terms in
the summation will contribute higher order of~$\varepsilon$ in the
integration. Then applying
Proposition~\ref{Proposition:CEFEFirstEstimateRicciCurvature} we find
that
\begin{align*}
  ||\nablasl^3\Sigma_2||_{L^2(\mathcal{S}_{u,v})}\leq2\Delta_{\Sigma_{\star}}
  +C(I,\Delta_{e_{\star}},\Delta_{\Gamma_{\star}},\Delta_{\Sigma_{\star}},
  \Delta_{\Phi_{\star}})\varepsilon+o(\varepsilon),
\end{align*}
where the term~$o(\varepsilon)$ arises from the summation.

\smallskip
\noindent
\textbf{Estimates for $\nablasl^3\Sigma_1$.} In this case one has that
the~$\Delta$-equation for~$\Delta\delta^3\Sigma_1$ is of the form
\begin{align*}
  \Delta\delta^3\Sigma_1&=\Sigma_2\Phi\Gamma^2+\Phi\Gamma\nablasl\Sigma_2
  +\nablasl\Sigma_2\nablasl\Phi+\Sigma_2\Phi\nablasl\Gamma
  +\Phi\nablasl^2\Sigma_2 \\
  &\quad+s\Gamma^3+s\Gamma\delta\Gamma+s\delta^2\Gamma
  +\sum_{i_1+...+i_4=3}\delta^{i_1}\Xi\delta^{i_2}\Gamma^{i_3}\delta^{i_4}\Phi,
\end{align*}
here the first line on the right hand side contains the leading order
contribution, and~$\Phi$ does not contain~$\Phi_{22}$. From this
equation one readily obtains that
\begin{align*}
  ||\nablasl^3\Sigma_1||_{L^2(\mathcal{S}_{u,v})}\leq
  C(I,\Delta_{e_{\star}},\Delta_{\Gamma_{\star}},\Delta_{\Sigma_{\star}},
  \Delta_{\Phi_{\star}})\varepsilon+o(\varepsilon).
\end{align*}

\smallskip
\noindent
\textbf{Estimates for $\nablasl^3\Sigma_{3,4}$.} In this case the term
contributing to the leading order of the estimate
of~$||\nablasl^3\Sigma_3||_{L^2(\mathcal{S}_{u,v})}$ is
\begin{align*}
  \Delta\delta^3\Sigma_3&=\Sigma_2\Gamma^4+\Gamma\nablasl\Sigma_2
  +\Sigma_2\Gamma^2\nablasl\Gamma+\Sigma_2(\nablasl\Gamma)^2
  +\Gamma^2\nablasl^2\Sigma_2+\Sigma_2\Gamma\nablasl^2\Gamma
  +\Sigma_2\nablasl^3\Gamma+\Gamma\nablasl^3\Sigma_2\\
  &\quad+\sum_{i_1+...+i_4=3}\delta^{i_1}\Xi\delta^{i_2}
  \Gamma^{i_3}\delta^{i_4}\Phi,
\end{align*}
again here the first line of the right hand side offers the leading
contribution, and gives
\begin{align*}
  ||\nablasl^3\Sigma_3||_{L^2(\mathcal{S}_{u,v})}\leq
  C(I,\Delta_{e_{\star}},\Delta_{\Gamma_{\star}},\Delta_{\Sigma_{\star}},
  \Delta_{\Phi_{\star}})\varepsilon+o(\varepsilon).
\end{align*}

\smallskip
\noindent
\textbf{Concluding the argument.} From the analysis above, it follows
that we can choose
\begin{align*}
  \varepsilon_{\star}=\varepsilon_{\star}(I, \Delta_{e_{\star}},
  \Delta_{\Gamma_{\star}},\Delta_{\Sigma_{\star}},\Delta_{\phi_{\star}},
  \Delta_{\Phi_{\star}}, \Delta_{\phi},\Delta_{\Phi}),
\end{align*}
sufficiently small so that
\begin{align*}
  &\sup_{u,v}||\nablasl^3\{\mu,\lambda,\alpha,\beta,\epsilon\}
  ||_{L^2(\mathcal{S}_{u,v})}\leq3\Delta_{\Gamma_{\star}}, \\
  &\sup_{u,v}||\nablasl^3\{\rho,\sigma\}||_{L^2(\mathcal{S}_{u,v})}\leq
  C(I,\Delta_{e_{\star}},\Delta_{\Gamma_{\star}},\Delta_{\Phi_{\star}},
  \Delta_{\Phi}) ,\\
  &\sup_{u,v}||\nablasl^3\{\tau,\chi\}||_{L^2(\mathcal{S}_{u,v})}\leq
  C(I,\Delta_{e_{\star}},\Delta_{\Gamma_{\star}},\Delta_{\Phi_{\star}},
  \Delta_{\Phi}) , \\
  &\sup_{u,v}||\nablasl^3\Sigma_2||_{L^2(\mathcal{S}_{u,v})}
  \leq3\Delta_{\Sigma_{\star}}, \\
  &\sup_{u,v}||\nablasl^3\{\Sigma_1,\Sigma_3,\Sigma_4\}
  ||_{L^2(\mathcal{S}_{u,v})}
  \leq C(I,\Delta_{e_{\star}},\Delta_{\Gamma_{\star}},\Delta_{\Sigma_{\star}},
  \Delta_{\Phi_{\star}})\varepsilon,
\end{align*}
on~$\mathcal{D}_{u,v_\bullet}^{\,t}$.
\end{proof}

\subsection{The energy estimates for the curvature}

In this subsection, we show how to obtain the main energy estimates
for the components of the Ricci and rescaled Weyl curvature.

\subsubsection{Analysis of the rescaled Weyl tensor}

We begin by introducing some integral identities which follow from
using integration by parts in the conformal equations satisfied by the
components of the rescaled Weyl tensor,
equations~\eqref{CFEforth1}-\eqref{CFEforth8}. The proof of these
results follows the same arguments used for the components of the Weyl
tensor in Paper~I as the (vacuum) Bianchi identities have an identical
structure to that of the equations for the rescaled Weyl tensor and
are thus omitted.

\begin{proposition}
[\textbf{\em control of the angular
    derivatives of the components of the rescaled Weyl tensor}]
  \label{Proposition:SecondMainEstimaterescaledWeylCurvature}
Suppose that we are given a solution to the CEFEs in Stewart's gauge
and that~$\mathcal{D}_{u,v}$ is contained in the existence area. The
following~$L^2$ estimates for the components of the rescaled Weyl
curvature hold. First, 
  \begin{align*}
  &\sum_{i=0,1,2}\int_{\mathcal{N}_u(0,v)}|\phi_i|^2
    +\sum_{j=1,2,3}\int_{\mathcal{N}'_v(0,u)}Q^{-1}|\phi_j|^2\\
  &\qquad\qquad  \leq\sum_{i=0,1,2}\int_{\mathcal{N}_{0}(0,v)}|\phi_i|^2
    +\sum_{j=1,2,3}\int_{\mathcal{N}'_{0}(0,u)}Q^{-1}|\phi_j|^2
    + \int_{\mathcal{D}_{u,v}}\phi_H\phi\Gamma,
\end{align*}
then
\begin{align*}
  \sum_{i=0,1,2}\int_{\mathcal{N}_u(0,v)}|\nablasl\phi_i|^2
    +\sum_{j=1,2,3}\int_{\mathcal{N}'_v(0,u)}Q^{-1}|\nablasl\phi_j|^2
    &\leq\sum_{i=0,1,2}\int_{\mathcal{N}_{0}(0,v)}|\nablasl\phi_i|^2
    +\sum_{j=1,2,3}\int_{\mathcal{N}'_{0}(0,u)}Q^{-1}|\nablasl\phi_j|^2 \\
  &+ \int_{\mathcal{D}_{u,v}}|\nablasl\phi_H|
    (\phi\Gamma^2+\Gamma|\nablasl\phi|+\phi|\nablasl\Gamma|),
\end{align*}
next 
\begin{align*}
  &\sum_{i=0,1,2}\int_{\mathcal{N}_u(0,v)}|\nablasl^2\phi_i|^2
    +\sum_{j=1,2,3}\int_{\mathcal{N}'_v(0,u)}Q^{-1}|\nablasl^2\phi_j|^2\leq
    \sum_{i=0,1,2}\int_{\mathcal{N}_{0}(0,v)}|\nablasl^2\phi_i|^2+
    \sum_{j=1,2,3}\int_{\mathcal{N}'_{0}(0,u)}Q^{-1}|\nablasl^2\phi_j|^2\\
  &\qquad\qquad+ \int_{\mathcal{D}_{u,v}}|\nablasl^2\phi_H|
    (\Gamma|\nablasl^2\phi|+\phi|\nablasl^2\Gamma|+|\nablasl\phi|
    |\nablasl\Gamma|+\Gamma^2|\nablasl\phi|+\phi\Gamma|\nablasl\Gamma|
    +\Gamma^3\phi),
\end{align*}
and finally
\begin{align*}
  &\sum_{i=0,1,2}\int_{\mathcal{N}_u(0,v)}|\nablasl^3\phi_i|^2
    +\sum_{j=1,2,3}\int_{\mathcal{N}'_v(0,u)}Q^{-1}|\nablasl^3\phi_j|^2
    \leq\sum_{i=0,1,2}\int_{\mathcal{N}_{0}(0,v)}|\nablasl^3\phi_i|^2
    +\sum_{j=1,2,3}\int_{\mathcal{N}'_{0}(0,u)}Q^{-1}|\nablasl^3\phi_j|^2 \\
  &\qquad+ \int_{\mathcal{D}_{u,v}}|\nablasl^3\phi_H|\big(\Gamma|\nablasl^3\phi|
    +\phi|\nablasl^3\Gamma|+|\nablasl\Gamma||\nablasl^2\phi|
    +|\nablasl\phi||\nablasl^2\Gamma|+\Gamma^2|\nablasl^2\phi|
    +\Gamma\phi|\nablasl^2\Gamma| \\
  &\qquad\qquad\quad+\Gamma|\nablasl\Gamma||\nablasl\phi|
    +\phi|\nablasl\Gamma|^2+\Gamma^3|\nablasl\phi|
    +\phi\Gamma^2|\nablasl\Gamma|+\Gamma^4\phi\big),
\end{align*}
where~$\Gamma$ stands for arbitrary connection coefficients from the
collection~$\{\mu, \lambda, \alpha, \beta, \epsilon, \rho, \sigma,
\tau\}$.
\end{proposition}

To summarise, the previous results can be given a more general
formulation:
\begin{proposition} 
Suppose that we are given a solution to the CEFEs in Stewart's gauge
and that~$\mathcal{D}_{u,v}$ is contained in the existence area. Then
we have that
\begin{align*}
  &\sum_{i=0,1,2}\int_{\mathcal{N}_u(0,v)}|\nablasl^m\phi_i|^2
  +\sum_{j=1,2,3}\int_{\mathcal{N}'_v(0,u)}Q^{-1}|\nablasl^m\phi_j|^2
  \leq\sum_{i=0,1,2}\int_{\mathcal{N}_{0}(0,v)}|\nablasl^m\phi_i|^2\\
  &\qquad+\sum_{j=1,2,3}\int_{\mathcal{N}'_{0}(0,u)}Q^{-1}|\nablasl^m\phi_j|^2
  +\int_{\mathcal{D}_{u,v}}|\nablasl^m\phi_H
    |\sum_{i_1+i_2+i_3+i_4=m}|\nablasl^{i_1}\Gamma^{i_2}
    ||\nablasl^{i_3}\Gamma||\nablasl^{i_4}\phi|,
\end{align*}
where~$\phi$ contains~$\phi_k$, $k=0,...,4$, $\phi_H$
contains~$\phi_k$, $k=0,...,3$.
\end{proposition}

In addition, we have the following proposition:
\begin{proposition} 
[\textbf{\em control of the angular
    derivatives of the ``bad" components of the rescaled Weyl tensor}] 
\label{Proposition:EstimatesDerivativesrescaledWeyl34}
Suppose that we are given a solution to the CEFEs in Stewart's gauge
and that~$\mathcal{D}_{u,v}$ is contained in the existence area. Then
we have that
\begin{align*}
  &\int_{\mathcal{N}_u(0,v)}|\nablasl^m\phi_3|^2
    +\int_{\mathcal{N}'_v(0,u)}Q^{-1}|\nablasl^m\phi_4|^2
    \leq\int_{\mathcal{N}_{0}(0,v)}|\nablasl^m\phi_3|^2
    +\int_{\mathcal{N}'_{0}(0,u)}Q^{-1}|\nablasl^m\phi_4|^2 \\
  &\qquad+ \int_{\mathcal{D}_{u,v}}|\nablasl^m\phi_4
    |\sum_{i_1+i_2+i_3+i_4=m}|\nablasl^{i_1}\Gamma^{i_2}
    ||\nablasl^{i_3}(\rho+\epsilon)||\nablasl^{i_4}\phi_4|\\
  &\qquad+ \int_{\mathcal{D}_{u,v}}|\nablasl^m\phi_3|\sum_{i_1+i_2+i_3+i_4=m}|
    \nablasl^{i_1}\Gamma^{i_2}||\nablasl^{i_3}\Gamma||\nablasl^{i_4}\phi|\\
  &\qquad
  + \int_{\mathcal{D}_{u,v}}|\nablasl^m\phi_4|\sum_{i_1+i_2+i_3+i_4=m}|
    \nablasl^{i_1}\Gamma^{i_2}||\nablasl^{i_3}\Gamma||\nablasl^{i_4}\phi'_H|,
\end{align*}
where~$\phi$ contains~$\phi_3$ and~$\phi_4$, $\phi'_H$
contains~$\phi_2$ and~$\phi_3$.
\end{proposition}

\subsubsection{Analysis of the Ricci curvature}

In order to estimate the~$L^2$-norms of the components of the Ricci
tensor we need inequalities analogous to the ones used for the
rescaled Weyl tensor. In order to obtain these, we first we need to
regroup the conformal equations for the Ricci tensor shown in
Appendix~\ref{AppendCFE3}. More precisely, we pair the
components~$\Phi_{01}$ and~$\Phi_{11}$ by analysing
equations~\eqref{CFEthird2} and~\eqref{CFEthird8}; pair the
components~$\Phi_{02}$ and~$\Phi_{12}$ by
analysing~\eqref{CFEthird3}
and~\eqref{CFEthird7}+\eqref{CFEthird12}; pair the
components~$\Phi_{11}$ and~$\Phi_{12}$ by
analysing~\eqref{CFEthird4} and~\eqref{CFEthird7}; pair the
components~$\Phi_{01}$ and~$\Phi_{02}$ by
analysing~\eqref{CFEthird2}+\eqref{CFEthird12}
and~\eqref{CFEthird8}. Making use of this strategy one obtains the
following:

\begin{proposition} 
Suppose that we are given a solution to the CEFEs in Stewart's gauge
and that~$\mathcal{D}_{u,v}$ is contained in the existence area. Then
we have that
\begin{align*}
  \sum_{\Phi_i\in\Phi_L}\int_{\mathcal{N}_u(0,v)}|\Phi_i|^2
    +\sum_{\Phi_j\in\Phi_S}\int_{\mathcal{N}'_v(0,u)}Q^{-1}|\Phi_j|^2&\leq
    \sum_{\Phi_i\in\Phi_L}\int_{\mathcal{N}_{0}(0,v)}|\Phi_i|^2
    +\sum_{\Phi_j\in\Phi_S}\int_{\mathcal{N}'_{0}(0,u)}Q^{-1}|\Phi_j|^2\\
  &\quad+ \int_{\mathcal{D}_{u,v}}\Phi_H\Gamma\Phi+
    \int_{\mathcal{D}_{u,v}}\Phi_H\Sigma\phi,
\end{align*}
where~$\Phi_L=\equiv\{\Phi_{00},\Phi_{01},\Phi_{02},\Phi_{11}\}$,
$\Phi_S\equiv \{\Phi_{01},\Phi_{02},\Phi_{11},\Phi_{12}\}$,
$\Phi\equiv
\{\Phi_{00},\Phi_{01},\Phi_{02},\Phi_{11},\Phi_{12},\Phi_{22}\}$ and
$\Phi_H\equiv \{\Phi_{00},\Phi_{01},\Phi_{02},\Phi_{11},\Phi_{12}\}$.
\end{proposition} 

\begin{proof}
For simplicity, we demonstrate the argument with the conformal
equations~\eqref{CFEthird1} and~\eqref{CFEthird10} written in
the form
\begin{align*}
& \Delta\Phi_{00}=\delta\Phi_{10}+\Gamma\Phi+\Sigma\phi, \\
& D\Phi_{01}=\delta\Phi_{00}+\Gamma\Phi+\Sigma\phi.
\end{align*}
Integrating by parts we have that 
\begin{align*}
  \int_{\mathcal{N}_u(0,v)}|\Phi_{00}|^2+\int_{\mathcal{N}'_v(0,u)}Q^{-1}|\Phi_{01}|^2
  &\leq\int_{\mathcal{N}_{0}(0,v)}|\Phi_{00}|^2+\int_{\mathcal{N}'_{0}(0,u)}Q^{-1}|\Phi_{01}|^2  \\
  &\quad+\int_{\mathcal{D}_{u,v}}(\Phi_{00},\Phi_{01})\Gamma\Phi
  + \int_{\mathcal{D}_{u,v}}(\Phi_{00},\Phi_{01})\Sigma\phi .
\end{align*}
A similar argument applies to the pairs~$\Phi_{01}$ and~$\Phi_{11}$,
$\Phi_{02}$ and~$\Phi_{12}$, $\Phi_{11}$ and~$\Phi_{12}$, $\Phi_{01}$
and~$\Phi_{02}$. Putting everything together we obtain the required
result.
\end{proof}

Now, applying the angular derivatives to the conformal equations we
obtain the following statement:

\begin{proposition} 
Suppose that we are given a solution to the CEFEs in Stewart's gauge
and that~$\mathcal{D}_{u,v}$ is contained in the existence area. Then
we have first that
\begin{align*}
  &\sum_{\Phi_i\in\Phi_L}\int_{\mathcal{N}_u(0,v)}|\nablasl\Phi_i|^2
    +\sum_{\Phi_j\in\Phi_S}\int_{\mathcal{N}'_v(0,u)}Q^{-1}|\nablasl\Phi_j|^2
    \leq\sum_{\Phi_i\in\Phi_L}\int_{\mathcal{N}_{0}(0,v)}|\nablasl\Phi_i|^2
    +\sum_{\Phi_j\in\Phi_S}\int_{\mathcal{N}'_{0}(0,u)}Q^{-1}|\nablasl\Phi_j|^2\\
  &\qquad+ \int_{\mathcal{D}_{u,v}}|\nablasl\Phi_H|\big(\Phi\Gamma^2
    +\Gamma|\nablasl\Phi|+\Phi|\nablasl\Gamma|\big)
    + \int_{\mathcal{D}_{u,v}}|\nablasl\Phi_H|\big(\Sigma\phi\Gamma
    +\phi|\nablasl\Sigma|+\Sigma|\nablasl\phi|\big),
\end{align*}
and also,
\begin{align*}
  &\sum_{\Phi_i\in\Phi_L}\int_{\mathcal{N}_u(0,v)}|\nablasl^2\Phi_i|^2
    +\sum_{\Phi_j\in\Phi_S}\int_{\mathcal{N}'_v(0,u)}Q^{-1}|\nablasl^2\Phi_j|^2
    \leq\sum_{\Phi_i\in\Phi_L}\int_{\mathcal{N}_{0}(0,v)}|\nablasl^2\Phi_i|^2
    +\sum_{\Phi_j\in\Phi_S}\int_{\mathcal{N}'_{0}(0,u)}Q^{-1}|\nablasl^2\Phi_j|^2\\
  &\quad+ \int_{\mathcal{D}_{u,v}}|\nablasl^2\Phi_H|\big(\Gamma|\nablasl^2\Phi|
    +\Phi|\nablasl^2\Gamma|+|\nablasl\Phi||\nablasl\Gamma|+\Gamma^2
    |\nablasl\Phi|+\Phi\Gamma|\nablasl\Gamma|+\Gamma^3\Phi\big) \\
  &\quad+\int_{\mathcal{D}_{u,v}}|\nablasl^2\Phi_H|\big(\Sigma\phi\Gamma^2
    +\Gamma\phi|\nablasl\Sigma|+\Gamma\Sigma|\nablasl\phi|+\Sigma\phi
    |\nablasl\Gamma|+|\nablasl\phi||\nablasl\Sigma|
    +\phi|\nablasl^2\Sigma|+\Sigma|\nablasl^2\phi|\big),
\end{align*}
and finally,
\begin{align*}
  &\sum_{\Phi_i\in\Phi_L}\int_{\mathcal{N}_u(0,v)}|\nablasl^3\Phi_i|^2
    +\sum_{\Phi_j\in\Phi_S}\int_{\mathcal{N}'_v(0,u)}Q^{-1}|\nablasl^3\Phi_j|^2
    \leq\sum_{\Phi_i\in\Phi_L}\int_{\mathcal{N}_{0}(0,v)}|\nablasl^3\Phi_i|^2
    +\sum_{\Phi_j\in\Phi_S}\int_{\mathcal{N}'_{0}(0,u)}Q^{-1}|\nablasl^3\Phi_j|^2 \\
  &\qquad+ \int_{\mathcal{D}_{u,v}}|\nablasl^3\Phi_H|\big(\Gamma|\nablasl^3\Phi|
    +\Phi|\nablasl^3\Gamma|+|\nablasl\Gamma||\nablasl^2\Phi|
    +|\nablasl\Phi||\nablasl^2\Gamma|+\Gamma^2|\nablasl^2\Phi|
    +\Gamma\Phi|\nablasl^2\Gamma| \\
  &\qquad\qquad\quad+\Gamma|\nablasl\Gamma||\nablasl\Phi|+\Phi|\nablasl\Gamma|^2
    +\Gamma^3|\nablasl\Phi|+\Phi\Gamma^2|\nablasl\Gamma|+\Gamma^4\Phi\big) \\
  &\qquad+ \int_{\mathcal{D}_{u,v}}|\nablasl^3\Phi_H|\big(\Sigma|\nablasl^3\phi|
    +\phi|\nablasl^3\Sigma|+|\nablasl^2\phi||\nablasl\Sigma|
    +|\nablasl^2\Sigma||\nablasl\phi|+\Sigma\Gamma|\nablasl^2\phi|
    +\Sigma\phi|\nablasl^2\Gamma| \\
  &\qquad\qquad\quad+\phi\Gamma|\nablasl^2\Sigma|+\Gamma|\nablasl\Sigma||
    \nablasl\phi|+\Sigma|\nablasl\phi||\nablasl\Gamma|+\phi|
    \nablasl\Sigma||\nablasl\Gamma|\\
  &\qquad\qquad\quad+\phi\Gamma^2|\nablasl\Sigma|+\Sigma\Gamma|\nablasl\phi|
    +\Sigma\phi\Gamma|\nablasl\Gamma|+\Sigma\phi\Gamma^3\big).
\end{align*}

\end{proposition} 

As before, we can summarise the previous estimates in the following
more concise statement:
\begin{proposition}  [\textbf{\em control of the higher angular
    derivatives of the components of the Ricci tensor}]
\label{Proposition:EstimatesDerivativesRiccigood}
Suppose that we are given a solution to the CEFEs in Stewart's gauge
and that~$\mathcal{D}_{u,v}$ is contained in the existence area. Then
we have that
\begin{align*}
  &\sum_{\Phi_i\in\Phi_L}\int_{\mathcal{N}_u(0,v)}|\nablasl^m\Phi_i|^2
  +\sum_{\Phi_j\in\Phi_S}\int_{\mathcal{N}'_v(0,u)}Q^{-1}|\nablasl^m\Phi_j|^2\\
  &\qquad\leq\sum_{\Phi_i\in\Phi_L}\int_{\mathcal{N}_{0}(0,v)}|\nablasl^m\Phi_i|^2
  +\sum_{\Phi_j\in\Phi_S}\int_{\mathcal{N}'_{0}(0,u)}Q^{-1}|\nablasl^m\Phi_j|^2\\
  &\qquad\qquad + \int_{\mathcal{D}_{u,v}}|\nablasl^m\Phi_H|\sum_{i_1+i_2+i_3+i_4=m}
  (|\nablasl^{i_1}\Gamma^{i_2}||\nablasl^{i_3}\Gamma||\nablasl^{i_4}\Phi|
  +|\nablasl^{i_1}\Gamma^{i_2}||\nablasl^{i_3}\Sigma||\nablasl^{i_4}\phi|),
\end{align*}
where~$m=0,\, 1,\, 2, \,3$.
\end{proposition}

Using equations~\eqref{CFEthird5} and~\eqref{CFEthird6} we can
obtain a similar control over the components~$\Phi_{12}$
and~$\Phi_{22}$. More precisely, one has that: 

\begin{proposition} [\textbf{\em control of the higher angular
    derivatives of the ``bad'' components of the Ricci tensor}] 
\label{Proposition:EstimatesDerivativesRicci1222}
Suppose that we are given a solution to the CEFEs in Stewart's gauge
and that~$\mathcal{D}_{u,v}$ is contained in the existence area. Then
we have that
\begin{align*}
  &\int_{\mathcal{N}_u(0,v)}|\nablasl^m\Phi_{12}|^2
    +\int_{\mathcal{N}'_v(0,u)}Q^{-1}|\nablasl^m\Phi_{22}|^2
    \leq\int_{\mathcal{N}_{0}(0,v)}|\nablasl^m\Phi_{12}|^2
    +\int_{\mathcal{N}'_{0}(0,u)}Q^{-1}|\nablasl^m\Phi_{22}|^2 \\
  &+ \int_{\mathcal{D}_{u,v}}|\nablasl^m\Phi_{22}|\sum_{i_1+i_2+i_3+i_4=m}
    |\nablasl^{i_1}\Gamma'^{i_2}||\nablasl^{i_3}\Gamma'||\nablasl^{i_4}\Phi_{22}|\\
  &+ \int_{\mathcal{D}_{u,v}}|\nablasl^m\Phi_{12}
    |\sum_{i_1+i_2+i_3+i_4=m}|\nablasl^{i_1}\Gamma^{i_2}
    ||\nablasl^{i_3}\Gamma||\nablasl^{i_4}\Phi| \\
  &+ \int_{\mathcal{D}_{u,v}}|\nablasl^m\Phi_{22}|
    \sum_{i_1+i_2+i_3+i_4=m}|\nablasl^{i_1}\Gamma^{i_2}
    ||\nablasl^{i_3}\Gamma||\nablasl^{i_4}\Phi'_H| \\
  &+\int_{\mathcal{D}_{u,v}}\left(|\nablasl^m\Phi_{12}
    ||\nablasl^{i_1}\Gamma^{i_2}||\nablasl^{i_3}\Sigma||\nablasl^{i_4}\phi|
    +|\nablasl^m\Phi_{22}||\nablasl^{i_1}\Gamma^{i_2}
    ||\nablasl^{i_3}\Sigma||\nablasl^{i_4}\phi'_H|\right),
\end{align*}
where~$\Gamma'$ does not contain~$\tau$ and~$\chi$, $\Phi$ does not
contain~$\Phi_{00},\Phi'_H$ does not contains~$\Phi_{22}$
and~$\Phi_{00}$, $\phi$ contains~$\phi_3$ and~$\phi_4$, $\phi'_H$
contains~$\phi_2$ and~$\phi_3$.
\end{proposition}

Making use of the previous estimates for the Ricci tensor, we can show
their boundedness in the truncated diamonds:

\begin{proposition}
[\textbf{\em control of the components of the Ricci
    tensor in terms of the initial data}]
\label{Proposition:FinalEstimateRicciCurvature}
Suppose we are given a solution to the vacuum CEFE's in Stewart's
gauge arising from data for the CIVP satisfying
\begin{align*}
  \Delta_{e_\star},\;\Delta_{\Gamma_\star}, \;\Delta_{\Sigma_\star},
  \;\Delta_{\Phi_\star}\;\Delta_{\phi_\star} <\infty,
\end{align*}
with the solution itself satisfying
\begin{align*}
 & \sup_{u,v}||\{\mu, \lambda, \alpha, \beta, \epsilon, \rho, \sigma,
  \tau, \chi,\Sigma_i\}||_{L^\infty(\mathcal{S}_{u,v})}< \infty\,,\quad
  \sup_{u,v}||\nablasl\{\mu, \lambda, \alpha, \beta, \epsilon, \rho,
  \sigma,\Sigma_i\}
  ||_{L^4(\mathcal{S}_{u,v})}<\infty\,,\\
&  \sup_{u,v}||\nablasl^2\{\mu,
  \lambda, \alpha, \beta, \epsilon, \rho, \sigma,
  \tau,\Sigma_i\}||_{L^2(\mathcal{S}_{u,v})}<\infty\,,\quad
  \sup_{u,v}||\nablasl^3\{\mu,\lambda,\alpha,\beta,\epsilon,\tau,\Sigma_i\}
  ||_{L^2(\mathcal{S}_{u,v})}<\infty\,,\\
 & \Delta_{\Phi}(\mathcal{S})<\infty\,, \quad
  \Delta_{\Phi}<\infty\,,\quad \Delta_{\phi}(\mathcal{S})<\infty\,, \quad
  \Delta_{\phi}<\infty,
\end{align*}
on some truncated causal diamond~$\mathcal{D}_{u,v_\bullet}^{\,t}$.
Then there
exists~$\varepsilon_\star=\varepsilon_\star(I,\Delta_{e_\star},\Delta_{\Gamma_\star},
\Delta_{\Sigma_\star},\Delta_{\Phi_\star},\Delta_{\phi})$ such that
for~$\varepsilon_\star\leq\varepsilon$ we have
\begin{eqnarray*}
\Delta_{\Phi}<C(I,\Delta_{e_{\star}},\Delta_{\Gamma_{\star}},\Delta_{\Phi_{\star}}).
\end{eqnarray*}
\end{proposition}

\begin{proof}
We need to control the integration in~$\mathcal{D}_{u,v}$ in
Propositions~\ref{Proposition:EstimatesDerivativesRiccigood}
and~\ref{Proposition:EstimatesDerivativesRicci1222}. Firstly, we focus
on Proposition~\ref{Proposition:EstimatesDerivativesRiccigood}. We
need control
\begin{align*}
  \int_{\mathcal{D}_{u,v}}|\nablasl^m\Phi_H|\sum_{i_1+i_2+i_3+i_4=m}
  (|\nablasl^{i_1}\Gamma^{i_2}||\nablasl^{i_3}\Gamma
  ||\nablasl^{i_4}\Phi|+|\nablasl^{i_1}\Gamma^{i_2}
  ||\nablasl^{i_3}\Sigma||\nablasl^{i_4}\phi|),
\end{align*}
where~$\Phi_H=\{\Phi_{00},\Phi_{01},\Phi_{02},\Phi_{11},\Phi_{12}\}$. We
can separate~$|\nablasl^m\Phi_H|$ and the summation using the H\"older
inequality. In turn, the term~$|\nablasl^m\Phi_H|$ can be controlled
as follows:
\begin{align*}
  ||\nablasl^m\Phi_H||_{L^2(\mathcal{D}_{u,v})}=\left(\int_0^u\int_0^v\int_S
  |\nablasl^m\Phi_H|^2\right)^{1/2}\leq C\Delta_{\Phi}\varepsilon^{1/2}.
\end{align*}
We observe that as~$\Phi$ contains~$\Phi_{22}$, we can only control it
on~$\mathcal{N}'_v$. Accordingly, we have that
\begin{align*}
||\nablasl^m\Phi||_{L^2(\mathcal{D}_{u,v})}\leq C\Delta_{\Phi}.
\end{align*}
Next, we need to analyse the~$L^2$-norm of the summation. Observing
that the first term of the summation has a structure similar to that
of the Weyl tensor~$\Psi$ in vacuum Einstein case, we readily obtain
that this term is controlled by
\begin{align*}
  C(I,\Delta_{e_{\star}},\Delta_{\Gamma_{\star}},\Delta_{\Phi_{\star}},
  \Delta_{\Phi})\varepsilon^{1/2}.
\end{align*}
The second term in the summation can be shown to be less than
\begin{align*}
  C\Delta_{\Phi}\varepsilon^{1/2}\sum_{i_1+i_2+i_3+i_4=m}
  ||\nablasl^{i_1}\Gamma^{i_2}\nablasl^{i_3}
  \Sigma\nablasl^{i_4}\phi||_{L^2(\mathcal{D}_{u,v})}.
\end{align*}
Every time we encounter the components~$\phi_0$ to~$\phi_3$ and their
derivatives, we can control them through the $L^2$-norm on the long
light cone~$\mathcal{N}_u$. Moreover, by analogy to~$\Phi_{22}$, we
control $\phi_4$ and its derivatives on the short light
cone~$\mathcal{N}_v'$. Hence following the same procedure we can
obtain that this norm is less than
\begin{align*}
   C(I,\Delta_{e_{\star}},\Delta_{\Gamma_{\star}},\Delta_{\Sigma_{\star}},
   \Delta_{\Phi_{\star}},\Delta_{\phi},\Delta_{\Phi})\varepsilon^{3/2}.
\end{align*}
In the next step, we consider the terms on the right hand side of the
estimate in
Proposition~\ref{Proposition:EstimatesDerivativesRicci1222}. The terms
\begin{align*}
  \int_{\mathcal{D}_{u,v}}|\nablasl^m\Phi_{12}|\sum_{i_1+i_2+i_3+i_4=m}
  |\nablasl^{i_1}\Gamma^{i_2}||\nablasl^{i_3}\Gamma||\nablasl^{i_4}\Phi|
\end{align*}
can be controlled in the same manner as it was done in
Proposition~\ref{Proposition:EstimatesDerivativesRiccigood} and are
bounded by
\begin{align*}
  C(I,\Delta_{e_{\star}},\Delta_{\Gamma_{\star}},\Delta_{\Phi_{\star}},
  \Delta_{\Phi})\varepsilon^{1/2}.
\end{align*}
Next, the terms
\begin{align*}
  \int_{\mathcal{D}_{u,v}}|\nablasl^m\Phi_{22}|\sum_{i_1+i_2+i_3+i_4=m}
  |\nablasl^{i_1}\Gamma^{i_2}||\nablasl^{i_3}\Gamma||\nablasl^{i_4}\Phi'_H|
\end{align*}
can also be controlled because it does not contains the
term~$(\Phi_{22})^2$. Moreover, the terms
\begin{align*}
  \sum_{i_1+i_2+i_3+i_4=m}\int_{\mathcal{D}_{u,v}}|\nablasl^m\Phi_{12}
  ||\nablasl^{i_1}\Gamma^{i_2}||\nablasl^{i_3}\Sigma||\nablasl^{i_4}\phi|
\end{align*}
can be controlled by
\begin{align*}
  C(I,\Delta_{e_{\star}},\Delta_{\Gamma_{\star}},\Delta_{\Sigma_{\star}},
  \Delta_{\Phi_{\star}},\Delta_{\phi},\Delta_{\Phi})\varepsilon^{3/2}.
\end{align*}
In the case of the term
\begin{align*}
  \int_{\mathcal{D}_{u,v}}|\nablasl^m\Phi_{22}|\sum_{i_1+i_2+i_3+i_4=m}
  |\nablasl^{i_1}\Gamma'^{i_2}||\nablasl^{i_3}\Gamma'||\nablasl^{i_4}\Phi_{22}|,
\end{align*}
we readily found that it is bounded by
\begin{align*}
  &C(I,\Delta_{e_{\star}},\Delta_{\Gamma_{\star}},\Delta_{\Phi_{\star}})
    \int_0^v||\nablasl^m\Phi_{22}||_{L^2(\mathcal{N}'_v(0,u))}
    \sum_{i=0}^m||\nablasl^i\Phi_{22}||_{L^2(\mathcal{N}'_v(0,u))} \\
  &\leq C(I,\Delta_{e_{\star}},\Delta_{\Gamma_{\star}},
    \Delta_{\Phi_{\star}})\int_0^v\sum_{i=0}^m
    ||\nablasl^i\Phi_{22}||^2_{L^2(\mathcal{N}'_v(0,u))}.
\end{align*}
Similarly, we also have that
\begin{align*}
  \sum_{i_1+i_2+i_3+i_4=m}\int_{\mathcal{D}_{u,v}}|\nablasl^m\Phi_{22}\nablasl^{i_1}
    \Gamma^{i_2}\nablasl^{i_3}\Sigma\nablasl^{i_4}\phi'_H|\leq \varepsilon^{3/2}
    C(I,\Delta_{e_{\star}},\Delta_{\Gamma_{\star}},\Delta_{\Sigma_{\star}},\Delta_{\Phi_{\star}},
    \Delta_{\phi})\int_0^v||\nablasl^m\Phi_{22}||^2_{L^2(\mathcal{N}'_v(0,u))}.
\end{align*}
Putting together the above estimates in the inequality of
Proposition~\ref{Proposition:EstimatesDerivativesRicci1222} we have
that
\begin{align*}
  &\sum_{i=0}^3||\nablasl^i\Phi_{22}||^2_{L^2(\mathcal{N}'_v)}
    \leq C\Delta_{\Phi_{\star}}+C(I,\Delta_{e_{\star}},
    \Delta_{\Gamma_{\star}},\Delta_{\Sigma_{\star}},\Delta_{\Phi})
    \varepsilon^{1/2}+C(I,\Delta_{e_{\star}},\Delta_{\Gamma_{\star}},
    \Delta_{\Sigma_{\star}},\Delta_{\Phi_{\star}},\Delta_{\phi},
    \Delta_{\Phi})\varepsilon^{3/2}\\
  &+(C(I,\Delta_{e_{\star}},\Delta_{\Gamma_{\star}},
    \Delta_{\Phi_{\star}})+C(I,\Delta_{e_{\star}},\Delta_{\Gamma_{\star}},
    \Delta_{\Sigma_{\star}},\Delta_{\Phi_{\star}},\Delta_{\phi})
    \varepsilon^{3/2})\int_0^v\sum_{i=0}^m
    ||\nablasl^i\Phi_{22}||^2_{L^2(\mathcal{N}'_v(0,u))}.
\end{align*}
Thus, applying the Gr\"onwall's inequality, we obtain that
\begin{align*}
  \Delta_{\Phi}\leq C(I,\Delta_{e_{\star}},\Delta_{\Gamma_{\star}},
  \Delta_{\Phi_{\star}})+C(I,\Delta_{e_{\star}},\Delta_{\Gamma_{\star}},
  \Delta_{\Sigma_{\star}},\Delta_{\Phi_{\star}},\Delta_{\Phi})
  \varepsilon^{1/2}+o(\epsilon^{1/2}).
\end{align*}
Finally, taking~$\varepsilon$ small enough we prove the proposition.

\end{proof}

The final ingredient in our analysis is the following proposition
whose proof is analogous to that of Proposition~$17$ in Paper~I:

\begin{proposition}
[\textbf{\em control of the components of the Rescaled Weyl
    tensor in terms of the initial data}]
\label{Proposition:FinalEstimateRescaledWeylCurvature}
With the same assumptions in
proposition~\ref{Proposition:FinalEstimateRicciCurvature} on some
truncated causal diamond~$\mathcal{D}_{u,v_\bullet}^{\,t}$. Then there
exists~$\varepsilon_\star=\varepsilon_\star(I,\Delta_{e_\star},\Delta_{\Gamma_\star},
\Delta_{\Sigma_\star},\Delta_{\Phi_\star},\Delta_{\phi_{\star}})$ such
that for~$\varepsilon_\star\leq\varepsilon$ we have
\begin{align*}
  \Delta_{\phi}\leq
  C(I,\Delta_{e_{\star}},\Delta_{\Gamma_{\star}},\Delta_{\Phi_{\star}},
  \Delta_{\phi_{\star}}).
\end{align*}
\end{proposition}

\section{Concluding the argument}
\label{Section:Summary}

The estimates obtained in the previous sections can be used in a
\emph{last slice argument} to obtain our main result. The proof is
completely analogous to that given in Section 7 in Paper~I and is thus
omitted.

\begin{theorem}
\label{Theorem:Main}
Given smooth initial data on~$\mathscr{I}^-\cup\mathcal{N}'_\star$
for~$0\leq v\leq I$ as constructed in Lemma~\ref{Lemma3}, there exists
$\varepsilon$ such that an unique smooth solution to the vacuum
conformal Einstein field equations exists in the region where~$0\leq
v\leq I$ and~$0\leq u\leq \varepsilon$ under the coordinate system and
$\varepsilon$ can be chosen to depend only on~$\Delta_{e_{\star}}$,
$\Delta_{\Gamma_{\star}}$, $\Delta_{\Sigma_{\star}}$,
$\Delta_{\Phi_{\star}}$ and~$\Delta_{\phi_{\star}}$. Moreover, in this
area,
\begin{align*}
  &\sup_{u,v}\sup_{\Gamma\in\{\chi,\mu,\lambda,\rho,\sigma,\alpha,
      \beta,\tau,\epsilon\}}\max\{\sum_{0}^1||\nablasl^i\Gamma
  ||_{L^{\infty}(\mathcal{S}_{u,v})},\sum_{i=0}^2||\nablasl^i
  \Gamma||_{L^4(\mathcal{S}_{u,v})},
    \sum_{i=0}^3||\nablasl^i\Gamma||_{L^2(\mathcal{S}_{u,v})}\} \\
  &+\sup_{u,v}\max\{\sum_{i=0}^1||\nablasl^i\Sigma_2
    ||_{L^{\infty}(\mathcal{S}_{u,v})},\sum_{i=0}^2||\nablasl^i\Sigma_2
    ||_{L^4(\mathcal{S}_{u,v})},\sum_{i=0}^3||\nablasl^i\Sigma_2
    ||_{L^2(\mathcal{S}_{u,v})}\}\\
  &+\Delta_{\Phi}+\Delta_{\phi}\leq C(I,\Delta_{e_{\star}},
  \Delta_{\Gamma_{\star}},\Delta_{\Sigma_{\star}},\Delta_{\Phi_{\star}},
  \Delta_{\phi_{\star}})
\end{align*}
and
\begin{align*}
  &\sup_{u,v}\max\{\sum_{i=0}^1||\nablasl^i\Sigma_{1,3,4}
    ||_{L^{\infty}(\mathcal{S}_{u,v})},\sum_{i=0}^2||\nablasl^i\Sigma_{1,3,4}
    ||_{L^4(\mathcal{S}_{u,v})},\sum_{i=0}^3||\nablasl^i\Sigma_{1,3,4}
    ||_{L^2(\mathcal{S}_{u,v})}\}\\
  &\leq C(I,\Delta_{e_{\star}},\Delta_{\Gamma_{\star}},
    \Delta_{\Sigma_{\star}},\Delta_{\Phi_{\star}})\varepsilon.
\end{align*}

\end{theorem}

\section*{Acknowledgements}

DH was supported by the FCT (Portugal) IF Program IF/00577/2015,
PTDC/MAT-APL/30043/2017 and~Project~No.~UIDB/00099/2020. PZ
acknowledges the support of the China Scholarship Council.

\appendix

\section{The conformal field equations in the NP formalism}
\label{Appendix:CFEinNP}

This appendix serves as a quick reference to our
equations. Throughout, we make use of the NP formalism in the
conventions used in the book~\cite{Ste91} which, in turn, follows the
conventions of~\cite{PenRin84}.

\subsection{The NP field equations}

Given the NP frame~$\{l^a,\ n^a,\ m^a,\ \bar{m}^a \}$, we denote
by~$D\equiv l^a\nabla_a$, $\Delta\equiv n^a\nabla_a$, $\delta\equiv
m^a\nabla_a$, $\bar{\delta} \equiv \bar{m}^a\nabla_a$ the associated
directional derivatives. The commutators are then given by
\begin{subequations}
\begin{align}
& (\Delta D - D\Delta) \psi = \big( (\pink{\gamma}+\pink{\bar{\gamma}}) D +
   (\epsilon+\bar{\epsilon})\Delta -(\bar{\tau} + \pi)\delta
   -(\tau+\bar{\pi})\bar{\delta}\big)\psi, \label{NPCommutator1}\\
   & (\delta D - D \delta) \psi =\big( (\pink{\bar{\alpha}}+\pink{\beta}
   -\pink{\bar{\pi}})D+\pink{\kappa}\Delta -(\bar{\rho}+\epsilon
   -\bar{\epsilon})\delta
   -\sigma\bar{\delta}  \big)\psi, \label{NPCommutator2}\\
& (\delta \Delta -\Delta \delta)\psi = \big( -\pink{\bar{\nu}} D +
   (\tau-\bar{\alpha}-\beta)\Delta + (\mu -\pink{\gamma} +\pink{\bar{\gamma}})\delta
   +\bar{\lambda} \bar{\delta}\big)\psi, \label{NPCommutator3}\\
& (\bar{\delta}\delta - \delta \bar{\delta})\psi = \big(
   (\pink{\bar{\mu}}-\pink{\mu})D + (\pink{\bar{\rho}}-\pink{\rho}) \Delta +
   (\alpha-\bar{\beta})\delta -(\bar{\alpha}-\beta)\bar{\delta} \big)\psi.
   \label{NPCommutator4},
\end{align}
\end{subequations}
where~$\psi$ is any scalar field. Here and in the following we
highlight the terms which vanish in our gauge.

We use the same notation as in Reference~\cite{CFEBook} to denote the
components of Weyl spinor~$\Psi_{ABCD}$ and trace-free Ricci
spinor~$\Phi_{AA^{\prime}BB^{\prime}}$, namely~\{$\Psi_0$, $\Psi_1$,
$\Psi_2$, $\Psi_3$, $\Psi_4$\} and~\{$\Phi_{00}$, $\Phi_{01}$,
$\Phi_{02}$, $\Phi_{11}$, $\Phi_{12}$, $\Phi_{22}$, $\Lambda$\}. In
particular, we have that
\begin{align*}
&\Phi_{00}\equiv\frac{1}{2}R_{\{ab\}}l^al^b, \qquad \Phi_{01}\equiv
  \frac{1}{2}R_{\{ab\}}l^am^b, \qquad \Phi_{02}\qquad \frac{1}{2}R_{\{ab\}}m^am^b,& \\
&\Phi_{11}\equiv \frac{1}{4}R_{\{ab\}}(l^an^b+m^a\bar{m}^b), \qquad 
  \Phi_{12}\equiv \frac{1}{2}R_{\{ab\}}n^am^b,& \\
&\Phi_{22} \equiv\frac{1}{2}R_{\{ab\}}n^an^b, \qquad  \Lambda\equiv -\frac{1}{24}R,&
\end{align*}
where~$R_{\{ab\}}\equiv R_{ab}-\frac{1}{4}Rg_{ab}$. 

The structure equations, in turn, have the form 
\begin{subequations}
\begin{align}
  &\Delta\epsilon-D\pink{\gamma}=\Lambda-\Phi_{11}-\Psi_2
  +\epsilon(2\pink{\gamma}+\pink{\bar{\gamma}})+\pink{\gamma}\bar{\epsilon}
  +\pink{\kappa\nu}-\beta\pi-\alpha\bar{\pi}-\alpha\tau
  -\pi\tau-\beta\bar{\tau} \label{structureeq1},\\ 
  &\Delta\pink{\kappa}-D\tau=-\Phi_{01}-\Psi_1+3\pink{\gamma\kappa}
  +\pink{\bar{\gamma}\kappa}-\bar{\pi}\rho-\pi\sigma-\epsilon\tau
  +\bar{\epsilon}\tau-\rho\tau-\sigma\bar{\tau} \label{structureeq2},
  \\
  &\Delta\pi-D\pink{\nu}=-\Phi_{21}-\Psi_3+3\epsilon\pink{\nu}
  +\bar{\epsilon}\pink{\nu}-\pink{\gamma}\pi+\pink{\bar{\gamma}}\pi-\mu\pi
  -\lambda\bar{\pi}-\lambda\tau-\mu\bar{\tau}
  \label{structureeq3},\\
  &\delta\pink{\gamma}-\Delta\beta=\Phi_{12}-\bar{\alpha}\pink{\gamma}
  -2\beta\pink{\gamma}
  +\beta\pink{\bar{\gamma}}+\alpha\bar{\lambda}+\beta\mu-\epsilon\pink{\bar{\nu}}
  -\pink{\nu}\sigma+\pink{\gamma}\tau+\mu\tau \label{structureeq4},
  \\
  &\delta\epsilon-D\beta=-\Psi_{1}+\bar{\alpha}\epsilon+\beta\bar{\epsilon}
  +\pink{\gamma\kappa}+\pink{\kappa}\mu-\epsilon\bar{\pi}-\beta\bar{\rho}
  -\alpha\sigma-\pi\sigma \label{structureeq5},\\
  &\delta\pink{\kappa}-D\sigma=-\Psi_{0}+\bar{\alpha}\pink{\kappa}
  +3\beta\pink{\kappa}
  -\pink{\kappa}\bar{\pi}-3\epsilon\sigma+\bar{\epsilon}\sigma-\rho\sigma
  -\bar{\rho}\sigma+\pink{\kappa}\tau \label{structureeq6},\\
  &\delta\pink{\nu}-\Delta\mu=\Phi_{22}+\lambda\bar{\lambda}+\pink{\gamma}\mu
  +\pink{\bar{\gamma}}\mu+\mu^2-\bar{\alpha}\pink{\nu}-3\beta\pink{\nu}
  -\pink{\bar{\nu}}\pi+\pink{\nu}\tau \label{structureeq7},\\
  &\delta\pi-D\mu=-2\Lambda-\Psi_2+\epsilon\mu+\bar{\epsilon}\mu
  +\pink{\kappa}\nu+\bar{\alpha}\pi-\beta\pi-\pi\bar{\pi}-\mu\bar{\rho}
  -\lambda\sigma \label{structureeq8},\\
  &\delta\tau-\Delta\sigma=\Phi_{02}-\pink{\kappa\bar{\nu}}+\bar{\lambda}\rho
  -3\pink{\gamma}\sigma+\pink{\bar{\gamma}}\sigma+\mu\sigma-\bar{\alpha}\tau
  +\beta\tau+\tau^2 \label{structureeq9},\\
  &\bar{\delta}\beta-\delta\alpha=-\Lambda-\Phi_{11}+\Psi_2
  -\alpha\bar{\alpha}+2\alpha\beta-\beta\bar{\beta}-\epsilon\mu
  +\epsilon\bar{\mu}-\pink{\gamma}\rho-\mu\rho+\pink{\gamma}\bar{\rho}
  +\lambda\sigma \label{structureeq10},\\
  &\bar{\delta}\pink{\gamma}-\Delta\alpha=\Psi_3-\bar{\beta}\pink{\gamma}
  -\alpha\pink{\bar{\gamma}}+\beta\lambda+\alpha\bar{\mu}-\epsilon\pink{\nu}
  -\pink{\nu}\rho+\lambda\tau+\pink{\gamma}\bar{\tau}\label{structureeq11},\\
  &\bar{\delta}\epsilon-D\alpha=-\Phi_{10}+2\alpha\epsilon
  +\bar{\beta}\epsilon-\alpha\bar{\epsilon}+\pink{\gamma\bar{\kappa}}
  +\pink{\kappa}\lambda-\epsilon\pi-\alpha\rho-\pi\rho
  -\beta\bar{\sigma} \label{structureeq12},\\
  &\bar{\delta}\pink{\kappa}-D\rho=-\Phi_{00}+3\alpha\pink{\kappa}
  +\bar{\beta}\pink{\kappa}
  -\pink{\kappa}\pi-\epsilon\rho-\bar{\epsilon}\rho-\rho^2-\sigma\bar{\sigma}
  +\pink{\bar{\kappa}}\tau \label{structureeq13},\\
  &\bar{\delta}\mu-\delta\lambda=-\Phi_{21}+\Psi_3-\bar{\alpha}\lambda
  +3\beta\lambda-\alpha\mu-\bar{\beta}\mu-\mu\pi+\bar{\mu}\pi-\pink{\nu}\rho
  +\pink{\nu}\bar{\rho} \label{structureeq14},\\
  &\bar{\delta}\pink{\nu}-\Delta\lambda=\Psi_4+3\pink{\gamma}\lambda
  -\pink{\bar{\gamma}}\lambda
  +\lambda\mu+\lambda\bar{\mu}-3\alpha\pink{\nu}-\bar{\beta}\pink{\nu}-\pink{\nu}\pi
  +\pink{\nu}\bar{\tau} \label{structureeq15},\\
  &\bar{\delta}\pi-D\lambda=-\Phi_{20}+3\epsilon\lambda
  -\bar{\epsilon}\lambda+\pink{\bar{\kappa}\nu}-\alpha\pi+\bar{\beta}\pi
  -\pi^2-\lambda\rho-\mu\bar{\sigma} \label{structureeq16},\\
  &\bar{\delta}\sigma-\delta\rho=-\Phi_{01}+\Psi_1-\pink{\kappa}\mu
  +\pink{\kappa}\bar{\mu}-\bar{\alpha}\rho-\beta\rho+3\alpha\sigma
  -\bar{\beta}\sigma-\rho\tau-\bar{\rho}\tau \label{structureeq17},\\
  &\bar{\delta}\tau-\Delta\rho=2\Lambda+\Psi_2-\pink{\kappa\nu}
  -\pink{\gamma}\rho
  -\pink{\bar{\gamma}}\rho+\bar{\mu}\rho+\lambda\sigma+\alpha\tau
  -\bar{\beta}\tau+\tau\bar{\tau} \label{structureeq18}.
\end{align}
\end{subequations}

\subsection{Conformal vacuum Einstein field equations}
\label{SpinorialCFE}

The rescaled Weyl spinor~$\phi_{ABCD}$ is related to the Weyl
spinor~$\Psi_{ABCD}$ of the NP formalism by 
\begin{align*}
\phi_{ABCD}=\Xi^{-1}\Psi_{ABCD}.
\end{align*}
Moreover,~$\phi_{ABCD}$ is related to the spinorial counterpart
of~$d_{abcd}$ via
\begin{align*}
d_{AA^{\prime}BB^{\prime}CC^{\prime}DD^{\prime}}=-\phi_{ABCD}
\epsilon_{A^{\prime}B^{\prime}}\epsilon_{C^{\prime}D^{\prime}}
-\phi_{A^{\prime}B^{\prime}C^{\prime}D^{\prime}}\epsilon_{AB}\epsilon_{CD}.
\end{align*}
We denote the components of~$\phi_{ABCD}$ by
\begin{align*}
\phi_0\equiv\Xi^{-1}\Psi_{0}, \quad \phi_0\equiv\Xi^{-1}\Psi_{0}, \quad 
\phi_1\equiv\Xi^{-1}\Psi_{1}, \quad \phi_2\equiv\Xi^{-1}\Psi_{2}, \quad 
\phi_3\equiv\Xi^{-1}\Psi_{3}, \quad \phi_4\equiv\Xi^{-1}\Psi_{4}.
\end{align*}
The spinorial counterpart of the Schouten tensor is
\begin{align*}
  L_{AA^{\prime}BB^{\prime}}=-\Lambda\epsilon_{AB}
  \epsilon_{A^{\prime}B^{\prime}}+\Phi_{ABA^{\prime}B^{\prime}}.
\end{align*}
Finally, we denote the components of the derivative of the conformal
factor~$\Xi$ by
\begin{align*}
  \Sigma_1 \equiv D\Xi, \quad \Sigma_2\equiv \Delta\Xi,
  \quad
  \Sigma_3=\delta\Xi=\Xi_{01^{\prime}}, \quad
  \Sigma_4=\bar{\delta}\Xi=\Xi_{10^{\prime}}.
\end{align*}

\subsubsection{The first conformal Einstein field equation}

The spinorial counterpart of the first conformal Einstein equation,
equation~\eqref{CFE1}, is
\begin{align}
\nabla_{BB^{\prime}}\nabla_{AA^{\prime}}\Xi=-\Xi\Phi_{ABA^{\prime}B^{\prime}}
+s\epsilon_{AB}\bar{\epsilon}_{A^{\prime}B^{\prime}}
+\Xi\Lambda\epsilon_{AB}\bar{\epsilon}_{A^{\prime}B^{\prime}}.\label{CFEfirst}
\end{align}
When we decompose it in terms of the NP null tetrad we obtain
\begin{subequations}
\begin{align}
  -\Sigma_1(\epsilon+\bar{\epsilon})+\Sigma_4\pink{\kappa}
  +\Sigma_3\pink{\bar{\kappa}}+D\Sigma_1&=-\Xi\Phi_{00} \label{CFEfirst1},\\
  \Sigma_2(\epsilon+\bar{\epsilon})-\Sigma_3\pi-\Sigma_4\bar{\pi}
  +D\Sigma_2&=s+\Xi\Lambda-\Xi\Phi_{11} \label{CFEfirst2}, \\
  -\Sigma_3(\epsilon-\bar{\epsilon})+\Sigma_2\pink{\kappa}
  -\Sigma_1\bar{\pi}+D\Sigma_3&=-\Xi\Phi_{01} \label{CFEfirst3},\\
  -\Sigma_1(\pink{\gamma}+\pink{\bar{\gamma}})+\Sigma_4\tau+\Sigma_3\bar{\tau}
  +\Delta\Sigma_1&=s+\Xi\Lambda-\Xi\Phi_{11} \label{CFEfirst4},\\
  \Sigma_2(\pink{\gamma}+\pink{\bar{\gamma}})-\Sigma_3\pink{\nu}
  -\Sigma_4\pink{\bar{\nu}}
  +\Delta\Sigma_2&=-\Xi\Phi_{22} \label{CFEfirst5},\\
  -\Sigma_3(\pink{\gamma}-\pink{\bar{\gamma}})
  -\Sigma_1\pink{\bar{\nu}}+\Sigma_2\tau
  +\Delta\Sigma_3&=-\Xi\Phi_{12} \label{CFEfirst6},\\
  -\Sigma_1({\bar{\alpha}+\beta})+\Sigma_3\bar{\rho}
  +\Sigma_4\sigma+\delta\Sigma_1&=-\Xi\Phi_{01} \label{CFEfirst7}, \\
  \Sigma_2({\bar{\alpha}+\beta})-\Sigma_4\bar{\lambda}
  -\Sigma_3\mu+\delta\Sigma_2&=-\Xi\Phi_{12} \label{CFEfirst8},\\
  -\Sigma_3(-\bar{\alpha}+\beta)-\Sigma_1\bar{\lambda}
  +\Sigma_2\sigma+\delta\Sigma_3&=-\Xi\Phi_{02} \label{CFEfirst9} ,\\
  \Sigma_4(-\bar{\alpha}+\beta)-\Sigma_1\mu+\Sigma_2\bar{\rho}
  +\delta\Sigma_4&=-s-\Xi\Lambda-\Xi\Phi_{11} \label{CFEfirst10}.
\end{align}
\end{subequations}

\subsubsection{The second conformal Einstein field equation}

The spinorial counterpart of the second conformal Einstein equation,
equation~\eqref{CFE2}, is given by
\begin{align}
  \nabla_{AA^{\prime}}s=\Lambda\nabla_{AA^{\prime}}\Xi
  -\Phi_{ABA^{\prime}B^{\prime}}\nabla^{BB^{\prime}}\Xi. \label{CFEsecond}
\end{align}
Its decomposition in terms of the NP null tetrad is given by
\begin{subequations}
\begin{align}
  &-Ds=-\Sigma_1\Lambda+\Sigma_2\Phi_{00}-\Sigma_4\Phi_{01}
  -\Sigma_3\Phi_{10}+\Sigma_1\Phi_{11} \label{CFEsecond1}, \\
  &-\Delta s=-\Sigma_2\Lambda+\Sigma_2\Phi_{11}
  -\Sigma_4\Phi_{12}-\Sigma_3\Phi_{21}+\Sigma_1\Phi_{22}
  \label{CFEsecond2},\\
  &-\delta s=-\Sigma_3\Lambda+\Sigma_2\Phi_{01}-\Sigma_4\Phi_{02}
  -\Sigma_3\Phi_{11}+\Sigma_1\Phi_{12} \label{CFEsecond3}.
\end{align}
\end{subequations}

\subsubsection{The third conformal Einstein field equation} 
\label{AppendCFE3}

The spinorial counterpart of the third equation is
\begin{align}
  &\nabla_{AA^{\prime}}\Phi_{BCB^{\prime}C^{\prime}}-\nabla_{BB^{\prime}}
  \Phi_{ACA^{\prime}C^{\prime}}=\epsilon_{BC}
  \bar{\epsilon}_{B^{\prime}C^{\prime}}\nabla_{AA^{\prime}}\Lambda
  -\epsilon_{AC}\bar{\epsilon}_{A^{\prime}C^{\prime}}\nabla_{BB^{\prime}}\Lambda
  \nonumber\\
  &-\bar{\phi}_{A^\prime B^\prime C^\prime
    D^\prime}\epsilon_{AB}\nabla_C^{\phantom{C}D^\prime}\Xi
  -\phi_{ABCD}\bar{\epsilon}_{A^{\prime}B^{\prime}}
  \nabla^D_{\phantom{D}C^\prime}\Xi.
\label{CFEthird}
\end{align}
The independent components of this equation can be found to be
\begin{subequations}
\begin{align}
  \Delta\Phi_{00}-\delta\Phi_{10}+2D\Lambda&=\Phi_{00}
  (2\pink{\gamma}+2\pink{\bar\gamma}-\mu)-2\Phi_{10}(\bar\alpha+\tau)
  -2\Phi_{01}\bar\tau+2\Phi_{11}\bar\rho+\Phi_{20}\sigma \nonumber \\
  &\quad+\Sigma_3\bar\phi_1-\Sigma_1\bar\phi_2 \label{CFEthird1} \\
  \Delta\Phi_{01}-\delta\Phi_{11}+\delta\Lambda&=\Phi_{01}(2\pink{\gamma}-\mu)
  +\Phi_{00}\pink{\bar{\nu}}+\Phi_{12}\bar{\rho}+\Phi_{21}\sigma
  -\Phi_{10}\bar{\lambda}-2\Phi_{11}\tau-\Phi_{02}\bar{\tau}\nonumber\\
  &\quad+\Sigma_3\bar{\phi}_2-\Sigma_1\bar{\phi}_1 \label{CFEthird2},\\
  \Delta\Phi_{02}-\delta\Phi_{12}&=\Phi_{02}(2\pink{\gamma}
  -2\pink{\bar{\gamma}}-\mu)
  +2\Phi_{12}(\bar{\alpha}-\tau)+2\Phi_{01}\pink{\bar{\nu}}
  +\Phi_{22}\sigma-2\Phi_{11}\bar{\lambda}\nonumber\\
  &\quad+\Sigma_3\bar{\phi}_3-\Sigma_1\bar{\phi}_4 \label{CFEthird3},\\
  \Delta\Phi_{11}-\delta\Phi_{21}+\Delta\Lambda&=\Phi_{01}\pink{\nu}
  +\Phi_{10}\pink{\bar{\nu}}+\Phi_{21}(2\beta-\tau)+\Phi_{22}\bar{\rho}
  -\Phi_{20}\bar{\lambda}-2\Phi_{11}\mu-\Phi_{12}\bar{\tau}\nonumber\\
  &\quad+\Sigma_2\bar{\phi}_2-\Sigma_4\bar{\phi}_3 \label{CFEthird4},\\
  \Delta\Phi_{12}-\delta\Phi_{22}&=\Phi_{22}(2\bar{\alpha}+2\beta-\tau)
  +\Phi_{02}\pink{\nu}+2\Phi_{11}\pink{\bar{\nu}}-2\Phi_{12}(\pink{\bar{\gamma}}+\mu)
  -2\Phi_{21}\bar{\lambda}\nonumber\\
  &\quad+\Sigma_2\bar{\phi}_3-\Sigma_4\bar{\phi}_4 \label{CFEthird5},\\
  D\Phi_{22}-\delta\Phi_{21}+2\Delta\Lambda&=\Phi_{22}(\bar\rho-2\epsilon
  -2\bar\epsilon)+2\Phi_{21}(\beta+\bar\pi)+2\Phi_{12}\pi
  -\Phi_{20}\bar\lambda-2\Phi_{11}\mu\nonumber \\
  &\quad+\Sigma_3\phi_3-\Sigma_2\phi_2  \label{CFEthird6} \\
  D\Phi_{12}-\delta\Phi_{11}+\delta\Lambda&=\Phi_{02}\pi+2\Phi_{11}\bar{\pi}
  +\Phi_{12}(\bar{\rho}-2\bar{\epsilon})+\Phi_{21}\sigma-\Phi_{22}\pink{\kappa}
  -\Phi_{10}\bar{\lambda}-\Phi_{01}\mu\nonumber\\
  &\quad-\Sigma_2\phi_1+\Sigma_3\phi_2 \label{CFEthird7},\\
  D\Phi_{11}-\delta\Phi_{10}+D\Lambda&=\Phi_{01}\pi+\Phi_{10}(\bar{\pi}
  -2\bar{\alpha})+2\Phi_{11}\bar{\rho}+\Phi_{20}\sigma-\Phi_{21}\pink{\kappa}
  -\Phi_{12}\pink{\bar{\kappa}}-\Phi_{00}\mu\nonumber\\
  &\quad-\Sigma_4\phi_1+\Sigma_1\phi_2 \label{CFEthird8},\\
  D\Phi_{02}-\delta\Phi_{01}&=\Phi_{02}(2\epsilon-2\bar{\epsilon}
  +\bar{\rho})+2\Phi_{01}(\bar{\pi}-\beta)+2\Phi_{11}\sigma
  -2\Phi_{12}\pink{\kappa}-\Phi_{00}\bar{\lambda}\nonumber\\
  &\quad-\Sigma_2\phi_0+\Sigma_3\phi_1 \label{CFEthird9},\\
  D\Phi_{01}-\delta\Phi_{00}&=2\Phi_{01}(\epsilon+\bar{\rho})
  +\Phi_{00}\bar{\pi}+2\Phi_{10}\sigma-2\Phi_{00}(\bar{\alpha}+\beta)
  -2\Phi_{11}\pink{\kappa}-\Phi_{02}\pink{\bar{\kappa}}\nonumber\\
  &\quad-\Sigma_4\phi_0+\Sigma_1\phi_1 \label{CFEthird10},\\
  \delta\Phi_{10}-\bar{\delta}\Phi_{01}&=\Phi_{00}(\mu-\bar{\mu})
  +2\Phi_{11}(\rho-\bar{\rho})+2\Phi_{10}\bar{\alpha}+\Phi_{02}
  \bar{\sigma}-2\Phi_{01}\alpha-\Phi_{20}\sigma\nonumber\\
  &\quad+\Sigma_4\phi_1-\Sigma_3\bar{\phi}_1-\Sigma_1\phi_2
  +\Sigma_1\bar{\phi}_2 \label{CFEthird11},\\
  \delta\Phi_{11}-\bar{\delta}\Phi_{02}+\delta\Lambda&=2\Phi_{02}
  (\bar{\beta}-\alpha)+\Phi_{01}(\mu-2\bar{\mu})+\Phi_{12}(2\rho
  -\bar{\rho})+\Phi_{10}\bar{\lambda}-\Phi_{21}\sigma\nonumber\\
  &\quad+\Sigma_2\phi_1-\Sigma_3(\phi_2+\bar{\phi}_2)+\Sigma_1\bar{\phi}_3
  \label{CFEthird12},\\
  \delta\Phi_{21}-\bar{\delta}\Phi_{12}&=2\Phi_{11}(\mu-\bar{\mu})
  +\Phi_{22}(\rho-\bar{\rho})+2\Phi_{12}\bar{\beta}+\Phi_{20}\bar{\lambda}
  -2\Phi_{21}\beta-\Phi_{02}\lambda\nonumber\\ 
  &\quad+\Sigma_2(\phi_2-\bar{\phi}_2)-\Sigma_3\phi_3
  +\Sigma_4\bar{\phi}_3 \label{CFEthird13}.
\end{align}
\end{subequations}

\subsubsection{The fourth conformal Einstein field equation}

The spinorial counterpart of the fuorth conformal Einstein equation,
equation~\eqref{CFE4}, is given by
\begin{align}
\nabla_{DC^{\prime}}\phi_{ABC}^{\phantom{ABC}D}=0.
\label{CFEforth}
\end{align}
Its decomposition in terms of the NP null tetrad is given by
\begin{subequations}
\begin{align}
  &\Delta\phi_0-\delta\phi_1=-2\phi_1(\beta+2\tau)
  +\phi_0(4\pink{\gamma}-\mu)+3\phi_2\sigma  \label{CFEforth1},\\
  &\Delta\phi_1-\delta\phi_2=2\phi_1(\pink{\gamma}-\mu)+\phi_0\pink{\nu}
  +2\phi_3\sigma-3\phi_2\tau \label{CFEforth2},\\
  &\Delta\phi_2-\delta\phi_3=2\phi_3(\beta-\tau)-3\phi_2\mu
  +2\phi_1\pink{\nu}+\phi_4\sigma \label{CFEforth3},\\
  &\Delta\phi_3-\delta\phi_4=\phi_4(4\beta-\tau)+3\phi_2\pink{\nu}
  -2\phi_3(\pink{\gamma}+2\mu) \label{CFEforth4},\\
  &D\phi_1-\bar{\delta}\phi_0=\phi_0(\pi-4\alpha)+2\phi_1
  (\epsilon+2\rho)-3\phi_2\pink{\kappa} \label{CFEforth5},\\
  &D\phi_2-\bar{\delta}\phi_1=2\phi_1(\pi-\alpha)-\phi_0\lambda
  +3\phi_2\rho-2\phi_3\pink{\kappa} \label{CFEforth6},\\
  &D\phi_3-\bar{\delta}\phi_2=2\phi_3(\rho-\epsilon)
  -2\phi_1\lambda+3\phi_2\pi-\phi_4\pink{\kappa} \label{CFEforth7},\\
  &D\phi_4-\bar{\delta}\phi_3=\phi_4(\rho-4\epsilon)
  +2\phi_3(\alpha+2\pi)-3\phi_2\lambda \label{CFEforth8}.
\end{align}
\end{subequations}


\begin{thebibliography}{10}

\bibitem{HilValZha19}
D.~Hilditch, J.~A. {Valiente Kroon}, and P.~Zhao.
\newblock Revisiting the characteristic initial value problem for the vacuum
  einstein field equations.
\newblock in {\tt arXiv 1911.00047}, 2019.

\bibitem{Luk12}
J.~Luk.
\newblock On the local existence for the characteristic initial value problem
  in general relativity.
\newblock {\em Int. Math. Res. Not.}, 20:4625, 2012.

\bibitem{SteFri82}
J.~M. Stewart and H.~Friedrich.
\newblock Numerical relativity. {T}he characteristic initial value problem.
\newblock {\em Proc. Roy. Soc. Lond. A}, 384:427, 1982.

\bibitem{Fri91}
H.~Friedrich.
\newblock On the global existence and the asymptotic behaviour of solutions to
  the einstein-maxwell-yang-mills equations.
\newblock {\em J. Diff. Geom.}, 34:275, 1991.

\bibitem{Kan96b}
J.~K\'{a}nn\'{a}r.
\newblock On the existence of $\mbox{{C}}^\infty$ solutions to the asymptotic
  characteristic initial value problem in general relativity.
\newblock {\em Proc. Roy. Soc. Lond. A}, 452:945, 1996.

\bibitem{CFEBook}
J.~A. {Valiente Kroon}.
\newblock {\em Conformal Methods in General Relativity}.
\newblock Cambridge University Press, 2016.

\bibitem{Fri81a}
H.~Friedrich.
\newblock On the regular and the asymptotic characteristic initial value
  problem for {Einstein}'s vacuum field equations.
\newblock {\em Proc. Roy. Soc. Lond. A}, 375:169, 1981.

\bibitem{Fri81b}
H.~Friedrich.
\newblock The asymptotic characteristic initial value problem for {Einstein}'s
  vacuum field equations as an initial value problem for a first-order
  quasilinear symmetric hyperbolic system.
\newblock {\em Proc. Roy. Soc. Lond. A}, 378:401, 1981.

\bibitem{Fri82}
H.~Friedrich.
\newblock On the existence of analytic null asymptotically flat solutions of
  {Einstein}'s vacuum field equations.
\newblock {\em Proc. Roy. Soc. Lond. A}, 381:361, 1982.

\bibitem{CabChrWaf14}
Aurore Cabet, Piotr~T. Chru{\'s}ciel, and Roger~Tagne Wafo.
\newblock {On the characteristic initial value problem for nonlinear symmetric
  hyperbolic systems, including Einstein equations}.
\newblock 2014.

\bibitem{Fri86a}
H.~Friedrich.
\newblock On purely radiative space-times.
\newblock {\em Comm. Math. Phys.}, 103:35, 1986.

\bibitem{Fri88}
H.~Friedrich.
\newblock On static and radiative space-times.
\newblock {\em Comm. Math. Phys.}, 119:51, 1988.

\bibitem{ChrPae13a}
P.~T. Chru\'{s}ciel and T.-T. Paetz.
\newblock Kids like cones.
\newblock {\em Class. Quantum Grav.}, 30:235036, 2013.

\bibitem{ChrPae13b}
P.~T. Chru\'{s}ciel and T.-T. Paetz.
\newblock Solutions of the vacuum einstein equations with initial data on past
  null infinity.
\newblock {\em Class. Quantum Grav.}, 30:235037, 2013.

\bibitem{Fri14b}
H.~Friedrich.
\newblock The taylor expansion at past time-like infinity.
\newblock {\em Comm. Math. Phys.}, 324:263, 2014.

\bibitem{Cag80}
F.~Cagnac.
\newblock Probl\`eme de cauchy sur un cono\"{\i}de caract\'eristique.
\newblock {\em Ann. Fac. Sci. Toulouse $2^e$ s\'{e}rie}, 2:11, 1980.

\bibitem{Cag81}
F.~Cagnac.
\newblock Probl\`eme de cauchy sur un cono\"{\i}de caract\'eristique pour des
  equations quasi-lineaires.
\newblock {\em Ann. Mat. Pure Appl.}, 129:13, 1981.

\bibitem{Ste91}
J.~Stewart.
\newblock {\em Advanced general relativity}.
\newblock Cambridge University Press, 1991.

\bibitem{PenRin84}
R.~Penrose and W.~Rindler.
\newblock {\em Spinors and space-time. {V}olume 1. {T}wo-spinor calculus and
  relativistic fields}.
\newblock Cambridge University Press, 1984.

\bibitem{Fri83}
H.~Friedrich.
\newblock Cauchy problems for the conformal vacuum field equations in general
  relativity.
\newblock {\em Comm. Math. Phys.}, 91:445, 1983.

\bibitem{PenRin86}
R.~Penrose and W.~Rindler.
\newblock {\em Spinors and space-time. {V}olume 2. {S}pinor and twistor methods
  in space-time geometry}.
\newblock Cambridge University Press, 1986.

\end{thebibliography}


\end{document}